\documentclass[twoside,11pt]{article}
\usepackage{jair}
\usepackage[utf8]{inputenc}
\usepackage[T1]{fontenc}
\usepackage[colorlinks]{hyperref}
\usepackage{url}
\usepackage{booktabs}
\usepackage{nicefrac}
\usepackage{xcolor}
\usepackage{natbib}
\usepackage{tabulary}
\usepackage{parskip}
\usepackage{pifont}
\usepackage{amsmath}
\usepackage{amsthm}
\usepackage{bbold}
\usepackage{lmodern}
\usepackage{thm-restate}
\usepackage{xspace}
\usepackage{tikz}

\usetikzlibrary{shapes}
\usetikzlibrary{arrows}
\usetikzlibrary{automata}
\usetikzlibrary{shadows}

\hypersetup{linkcolor=sthlmRed,citecolor=sthlmGreen,urlcolor=sthlmBlue}

\newtheorem{definition}{Definition}
\newtheorem{example}{Example}

\def\ie{{\em i.e.,}\xspace}
\def\eg{{\em e.g.,}\xspace}
\def\st{{\em s.t.}\xspace}
\def\cf{{\em cf.}\xspace}
\def\wrt{{\em w.r.t.}\xspace}

\DeclareMathOperator*{\argmax}{\arg\,\max}

\definecolor{sthlmBlue}{RGB}{0,110,191} 

\definecolor{sthlmLightGreen}{RGB}{213,247,244} 

\definecolor{sthlmGreen}{RGB}{0,134,127} 

\definecolor{sthlmOrange}{RGB}{221,74,44} 

\definecolor{sthlmLightRed}{RGB}{254,222,237} 

\definecolor{sthlmRed}{RGB}{196,0,100} 

\definecolor{sthlmYellow}{RGB}{252,191,10} 

\firstpageno{1}

\title{On Convex Optimal Value Functions For POSGs}

\author{\name Rafael F. Cunha \\
\addr University of Groningen, Bernoulli Institute, The Netherlands \\
\email \href{mailto:r.f.cunha@rug.nl}{r.f.cunha@rug.nl} \\
\AND
\name Jacopo Castellini \\
\addr Univ Lyon, INSA Lyon, Inria, CITI, EA3720, 69621 Villeurbanne, France \\
\email \href{mailto:jacopo.castellini@insa-lyon.fr}{jacopo.castellini@insa-lyon.fr} \\
\AND
\name Johan Peralez \\
\addr Univ Lyon, INSA Lyon, Inria, CITI, EA3720, 69621 Villeurbanne, France \\
\email \href{mailto:johan.peralez@insa-lyon.fr}{johan.peralez@insa-lyon.fr} \\
\AND
\name Jilles S. Dibangoye \\
\addr University of Groningen, Bernoulli Institute, The Netherlands \\
\email \href{mailto:j.s.dibangoye@rug.nl}{j.s.dibangoye@rug.nl}}

\begin{document}
\maketitle

\begin{abstract}
Multi-agent planning and reinforcement learning can be challenging when agents cannot see the state of the world or communicate with each other due to communication costs, latency, or noise. Partially Observable Stochastic Games (POSGs) provide a mathematical framework for modelling such scenarios. This paper aims to improve the efficiency of planning and reinforcement learning algorithms for POSGs by identifying the underlying structure of optimal state-value functions. The approach involves reformulating the original game from the perspective of a trusted third party who plans on behalf of the agents simultaneously. From this viewpoint, the original POSGs can be viewed as Markov games where states are occupancy states, \ie posterior probability distributions over the hidden states of the world and the stream of actions and observations that agents have experienced so far. This study mainly proves that the optimal state-value function is a convex function of occupancy states expressed on an appropriate basis in all zero-sum, common-payoff, and Stackelberg POSGs.
\end{abstract}

\section{Introduction}
The Partially Observable Stochastic Game (POSG) framework is widely used for formalizing problems in multi-agent sequential decision-making \citep{bayesian,dynamic,multiagent}. This framework is designed to address situations where multiple agents interact with each other to control a dynamic environment and optimize their individual preferences despite stochasticity and sensing uncertainty. Recently, researchers have focused on the subclasses of zero-sum, common-payoff, and Stackelberg POSGs (\emph{st}-POSGs) \citep{economic,decentralized,security}. Various approaches have been developed to solve POSGs, with many relying on dynamic programming. These approaches aim to find a solution for the original game by recursively finding solutions for its subgames. However, often solutions of interest exist within a continuum, which requires leveraging uniform continuity properties of optimal state-value functions. Unfortunately, there is a lack of knowledge regarding the uniform continuity properties of optimal state-value functions, which can make it challenging to generalize knowledge from one partial solution to another, restricting the efficiency of existing techniques.

Recent advancements in studying optimal state-value functions for games have showcased innovative ways to generalize knowledge from past experiences. These approaches utilize uniform continuity properties, providing insights into the underlying structure of optimal state-value functions for specific games. The structural analysis of optimal state-value functions can be traced back to the seventies, with \cite{finite} demonstrating that the optimal state-value functions of finite-horizon Partially Observable Markov Decision Processes (POMDPs) are piecewise linear and convex functions of belief states. This discovery paved the way for a collection of efficient algorithms, including those developed by \citet{pomdp,pb,survey}. As such, many authors have since investigated the underlying structure of optimal value functions for subclasses of finite-horizon POSGs. For instance, \cite{framework} have shown that piecewise linearity and convexity of optimal state-value functions of belief states extends to Interactive Partially Observable Markov Decision Processes (I-POMDPs). \cite{rho} demonstrates that POMDPs with Lipschitz reward functions have Lipschitz-continuous near-optimal value functions. Similarly, \cite{heuristic,continuous} have proven that the optimal state-value functions of finite-horizon common-payoff POSGs (\emph{dec}-POMDPs) are piecewise linear and convex functions of occupancy states on the standard basis. In finite-horizon one-sided zero-sum POSGs (\emph{zs}-POSGs), \cite{stochastic,onesided} have demonstrated the convexity of the optimal state-value function over belief states. Moreover, \cite{structure} have established the linearity of optimal state-value functions of finite-horizon best-response problems in the space of distributions over the histories of an agent. Recently, \cite{zerosum} have complemented this linearity property by demonstrating the Lipchitz-continuity, which generalizes values between occupancy states for finite-horizon \emph{zs}-POSGs.

It is worth noting that uniform continuity properties are not all equal. For instance, the linearity property highlighted in \citet{structure} does not facilitate the generalization of values from one experience to another. Utilizing this property in isolation proves challenging. Undoubtedly, the Lipchitz-continuity property from \citet{zerosum} is useful in generalizing values across different experiences. However, it is worth noting that this approach can sometimes result in poor generalization. Generally, the stronger the uniform continuity property, the better. Unfortunately, there is no constructive approach for revealing the underlying structure of optimal state-value functions. Presently, exhibiting a uniform continuity property appears more like an art than science. This paper introduces a methodology for identifying the underlying structure of optimal state-value functions to address this bottleneck and enhance the ability to create efficient planning and reinforcement learning algorithms for finite-horizon POSGs. The approach involves reformulating the original game from the perspective of a trusted third party who plans on behalf of the agents simultaneously. From this viewpoint, the original finite-horizon POSGs can be viewed as finite-horizon Markov games where states are occupancy states, which are posterior probability distributions over the hidden states of the world and the stream of actions and observations that agents have experienced so far (\ie their private histories).

\begin{figure}[!ht]
\centering
\begin{tikzpicture}[scale=1.35]
\node[inner sep=0pt] (map) at (0,0) {\includegraphics[width=1\textwidth]{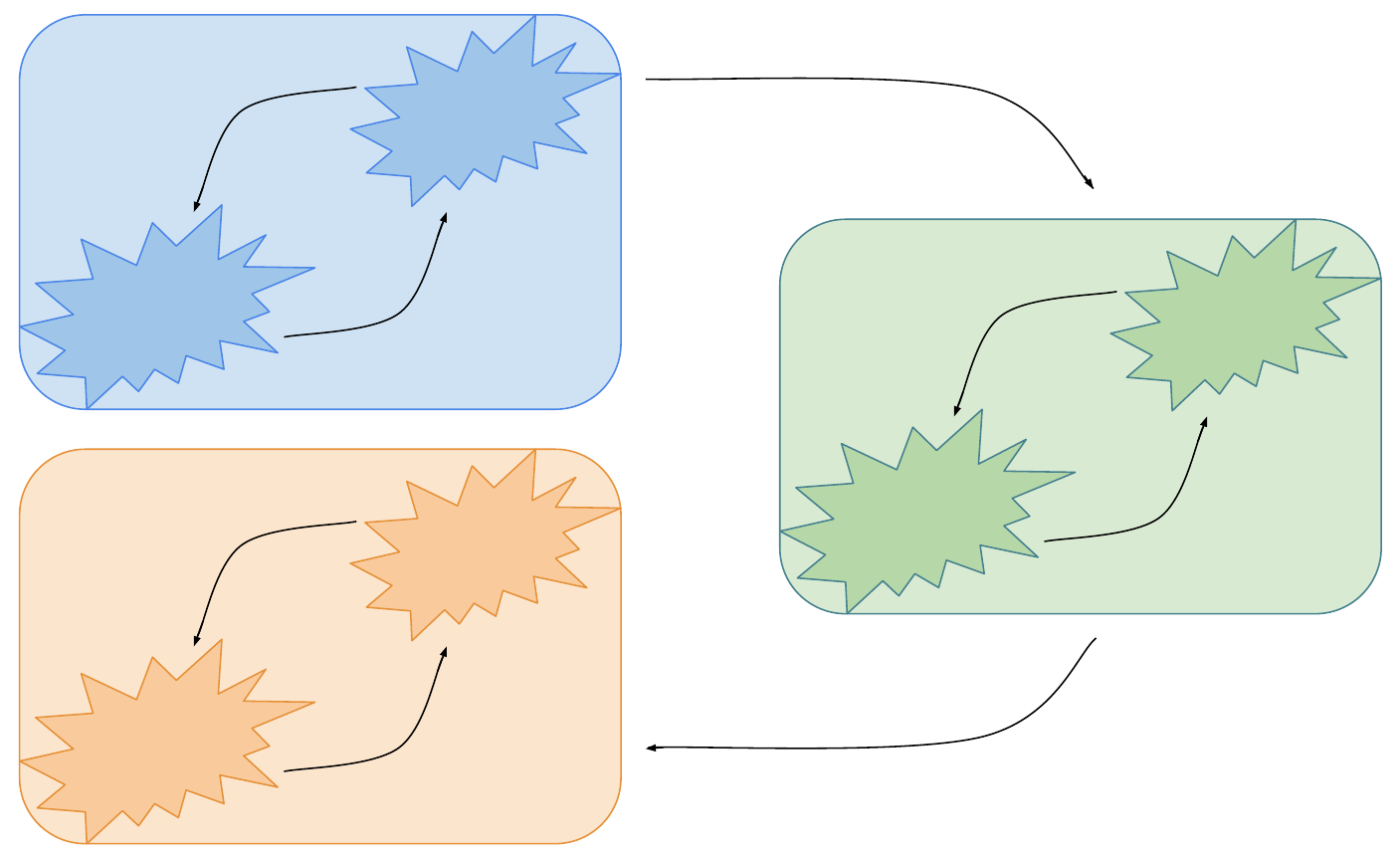}};
\node[scale=.55, text width = 3in, minimum height=1in] (node1) at (-3.7,2.95) {\textit{\bf Section \ref{sec:background}}
\begin{itemize}
\item[\textcolor{sthlmGreen}{\ding{51}}]{\scriptsize Original models}
\item[\textcolor{sthlmRed}{\ding{55}}]{\scriptsize Optimal decision-making}
\end{itemize}\par};
\node[scale=.55] (node1) at (-4.4,1.05) {\textbf{POSG} $M$};
\node[scale=.45] (node1) at (-4.4,.875) {\it (master game)};
\node[scale=.45, color=sthlmBlue] (node1) at (-4.2,1.8) {$\mathbf{1}$};
\node[scale=.45, color=sthlmBlue] (node1) at (-3.3,.615) {$\mathbf{1}$};
\node[scale=.55] (node1) at (-1.82,2.65) {\textbf{POMDP} $M(a_{0:}^{\neg i})$};
\node[scale=.45] (node1) at (-1.82,2.425) {\it (slave game)};
\node[scale=.45, color=sthlmBlue] (node1) at (-2.85,2.85) {$\mathbf{*}$};
\node[scale=.45, color=sthlmBlue] (node1) at (-1.95,1.65) {$\mathbf{*}$};
\node[scale=.55, text width = 3in, minimum height=1in] (node1) at (-3.7,-.55) {\textit{\bf Sections \ref{sec:connections} and \ref{sec:main:results}}
\begin{itemize}
\item[\textcolor{sthlmGreen}{\ding{51}}]{\scriptsize Optimal decision-making}
\item[\textcolor{sthlmGreen}{\ding{51}}]{\scriptsize Uniform continuity}
\end{itemize}\par};
\node[scale=.55] (node1) at (-4.4,-2.45) {\textbf{OMG} $M''$};
\node[scale=.45] (node1) at (-4.4,-2.625) {\it (master game)};
\node[scale=.45, color=sthlmOrange] (node1) at (-4.2,-1.7) {$\mathbf{1}$};
\node[scale=.45, color=sthlmOrange] (node1) at (-3.3,-2.915) {$\mathbf{1}$};
\node[scale=.55] (node1) at (-1.82,-.9) {\textbf{MDP} $M''(a_{0:}^{\neg i})$};
\node[scale=.45] (node1) at (-1.82,-1.075) {\it (slave game)};
\node[scale=.45, color=sthlmOrange] (node1) at (-2.85,-.655) {$\mathbf{*}$};
\node[scale=.45, color=sthlmOrange] (node1) at (-1.95,-1.85) {$\mathbf{*}$};
\node[scale=.55, text width = 3in, minimum height=1in] (node1) at (2.4,1.285) {\textit{\bf Section \ref{sec:reformulations}}
\begin{itemize}
\item[\textcolor{sthlmGreen}{\ding{51}}]{\scriptsize Optimal decision-making}
\item[\textcolor{sthlmRed}{\ding{55}}]{\scriptsize Uniform continuity}
\end{itemize}\par};
\node[scale=.55] (node1) at (1.8,-.6) {\textbf{PTMG} $M'$};
\node[scale=.45] (node1) at (1.8,-.775) {\it (master game)};
\node[scale=.45, color=sthlmGreen] (node1) at (1.95,.15) {$\mathbf{1}$};
\node[scale=.45, color=sthlmGreen] (node1) at (2.85,-1.05) {$\mathbf{1}$};
\node[scale=.55] (node1) at (4.35,0.95) {\textbf{MDP} $M'(a_{0:}^{\neg i})$};
\node[scale=.45] (node1) at (4.35,0.775) {\it (slave game)};
\node[scale=.45, color=sthlmGreen] (node1) at (3.325,1.2) {$\mathbf{*}$};
\node[scale=.45, color=sthlmGreen] (node1) at (4.18,0) {$\mathbf{*}$};
\node[scale=.55] (node1) at (0.6,2.65) {\it Central planner raw viewpoint};
\node[scale=.55] (node1) at (0.7,-2.4) {\it Central planner concise viewpoint};
\end{tikzpicture}
\caption{The paper employs a three-step transformation methodology to convert original games into appropriate representations for optimal decision-making. The methodology involves the use of several game-theoretic models, including Partially Observable Stochastic Game (POSG), Partially Observable Markov Decision Process (POMDP), Markov Decision Process (MDP), Plan-Time Markov Game (PTMG), and Occupancy-State Markov Game (OMG). \textbf{Best viewed in color}.}
\label{fig:three:step:method}
\end{figure}
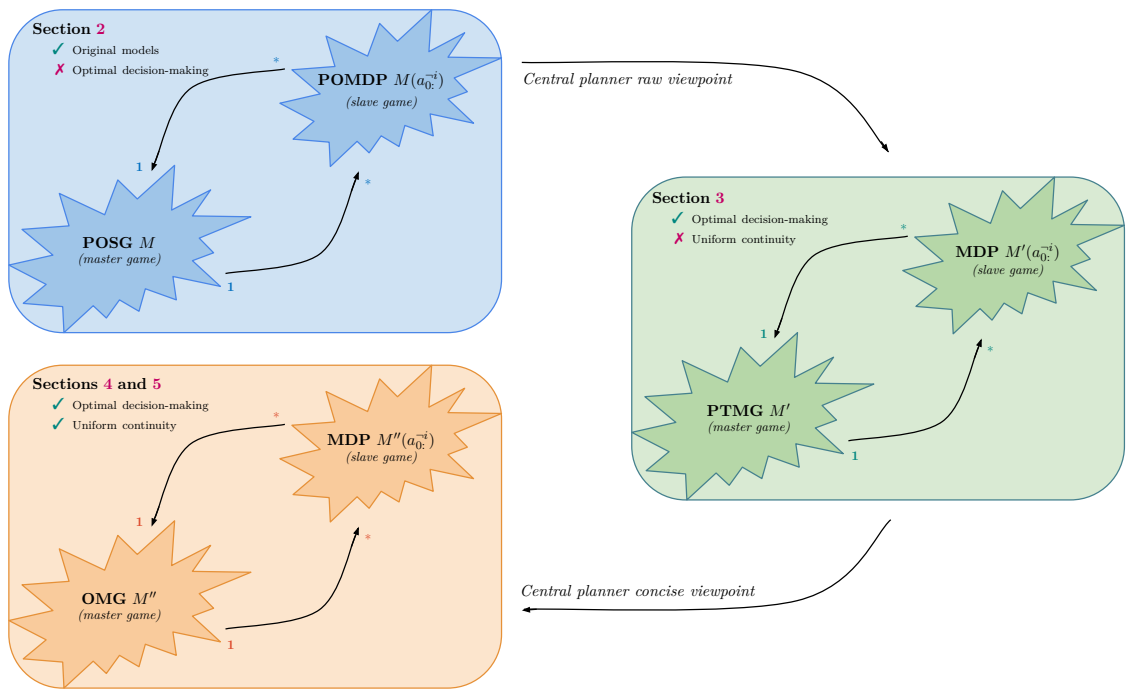

\paragraph{Contributions.} This study provides a proof that the optimal state-value function for the best-response problem, specifically for a slave game, is piecewise linear and convex over occupancy states when expressed on the standard basis, and linear over occupancy states when expressed appropriately. Additionally, the study presents the uniform continuity properties of optimal state-value functions for certain well-known subclasses of finite-horizon POSGs, which are master games. Specifically, when expressed appropriately, the optimal state-value functions for two subclasses of POSGs, namely \textit{zs}-POSGs and \textit{st}-POSGs, are convex over occupancy states. These uniform continuity properties are stronger than the previous ones. Moreover, the study demonstrates that the optimal state-value functions of \textit{dec}-POMDPs are piecewise linear and convex over occupancy states when expressed on the standard basis---a property originally established by \citet{continuous}. We instead provide a simpler proof. The authors emphasize the tightness of these properties, which is evidenced by the fact that they match pre-existing results in narrower settings under further assumptions. The study highlights a methodology to reveal uniform continuity properties of optimal value functions of master games, which is perhaps the most significant takeaway of this paper. This methodology comprises a series of transformations that produce three models for master and slave games, each sufficient to solve the others optimally. However, the last master and slave game formulations make it possible to exhibit uniform continuity properties that were not possible with the previous ones, yet the sufficiency of the former required the introduction of the latter. The key to establishing uniform continuity properties for master games appears to be the ability to appropriately express the optimal value function of the master games with those of the slave games. The authors hope that these findings will serve as a foundation for more efficient planning and reinforcement learning algorithms to exploit the structural properties inherent in POSGs.

The current paper is structured as follows, with each section delving into a unique representation of master and slave games, as illustrated in Figure \ref{fig:three:step:method}. Section \ref{sec:background} presents a formal description of these games, along with the relevant solution concepts. It also introduces the notations used throughout the document, preliminary properties, and \citeauthor{bellman}'s equations for the state- and action-value functions under fixed joint policies. In Section \ref{sec:reformulations}, a reformulation of master and slave games from the perspective of a trusted third party (or central planner) who plans on behalf of the agents is presented, providing a framework for optimal decision-making at the planning phase and decentralized decision-making at the execution phase. Section \ref{sec:connections} examines the connections between master and slave games through summaries of the central planner's knowledge, namely occupancy states. Finally, the main results of this paper are provided in Section \ref{sec:main:results}. The convexity properties of the optimal state-value functions of all slave and master games are presented here. The paper concludes with a discussion on the potential impact of the current findings.

\section{Background}
\label{sec:background}
This section formally describes master and slave games and the solution concepts of interest. Additionally, this section introduces the notation used throughout the document, preliminary properties, and \citeauthor{bellman}'s equations for the state- and action-value functions.

\begin{figure}[!ht]
\centering
\begin{tikzpicture}[scale=1.35]
\def \agentCentral{(0,-3) ellipse (1cm and 0.9cm);
\node[inner sep=0pt] (russell) at (0,-3.25) {\includegraphics[width=.06125\textwidth]{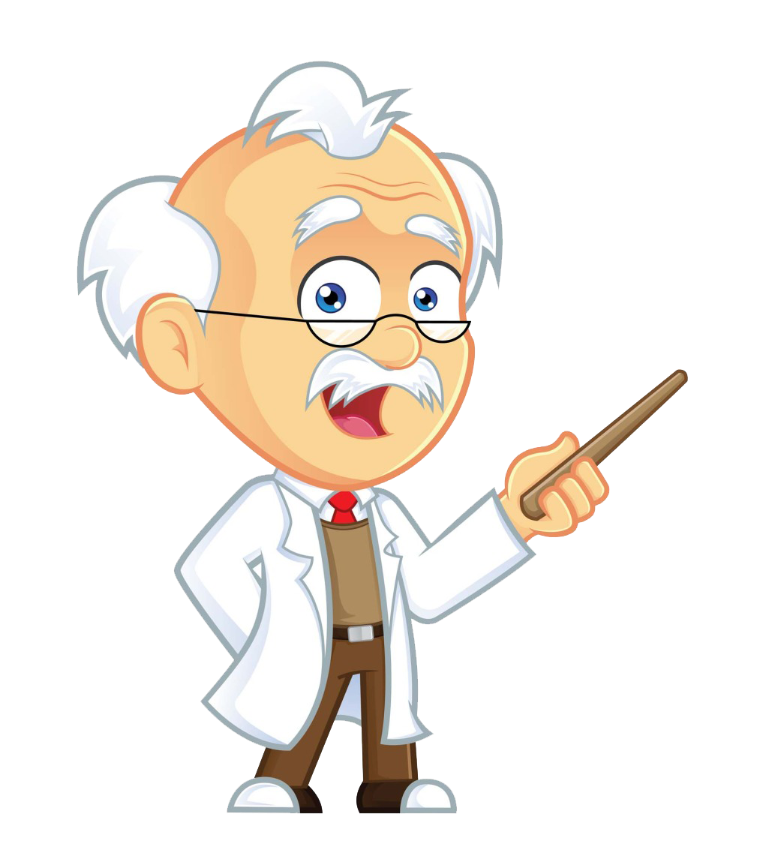}};}
\def \agentA{(-.5,3) ellipse (1cm and 0.9cm);
\node[inner sep=0pt] (calvin) at (-.5,2.7) {\includegraphics[width=.03725\textwidth]{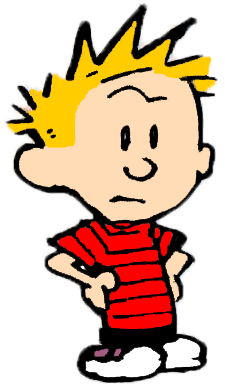}};}
\def \agentD{(-.5,0) ellipse (1cm and 0.9cm);
\node[inner sep=0pt] (susie) at (-.5,-.25) {\includegraphics[width=.03725\textwidth]{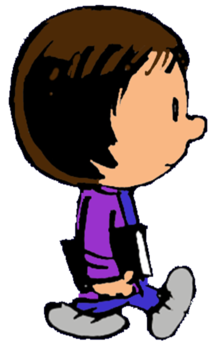}}; }
\draw[rounded corners, fill=white] (3.5,-2.5) rectangle +(3.5,5.5);
\node (agentCentral) at (-2,.5) {};
\node (agentA) at (-.25,2.25) {};
\node (agentD) at (-.25,-1) {};
\begin{scope}[scale=.5]
\def \agentA{(-1.5,4.5) ellipse (1cm and 0.9cm);
\node[inner sep=0pt] (calvin) at (-1.5,4.2) {\includegraphics[width=.075\textwidth]{figures/Calvin_Color.png}};}
\def \agentD{(-1.5,-2) ellipse (1cm and 0.9cm);
\node[inner sep=0pt] (susie) at (-1.5,-2.25) {\includegraphics[width=.075\textwidth]{figures/Susie.png}}; }
\draw[draw=white] \agentA;
\draw[draw=white] \agentD;
\node[inner sep=0pt] (hobbes) at (11.6,3) {\includegraphics[width=.075\textwidth]{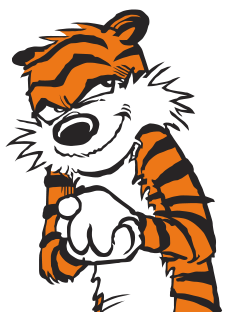}};
\node[inner sep=0pt] (door1) at (9,3) {\includegraphics[width=.075\textwidth]{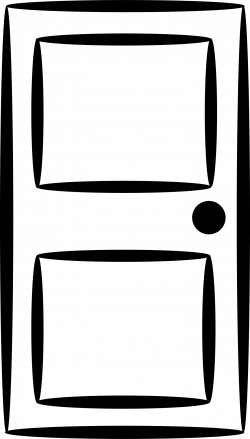}};
\node[inner sep=0pt] (door2) at (9,-2) {\includegraphics[width=.075\textwidth]{figures/Door.png}};
\node[inner sep=0pt] (treasure) at (11.75,-3) {\includegraphics[width=.075\textwidth]{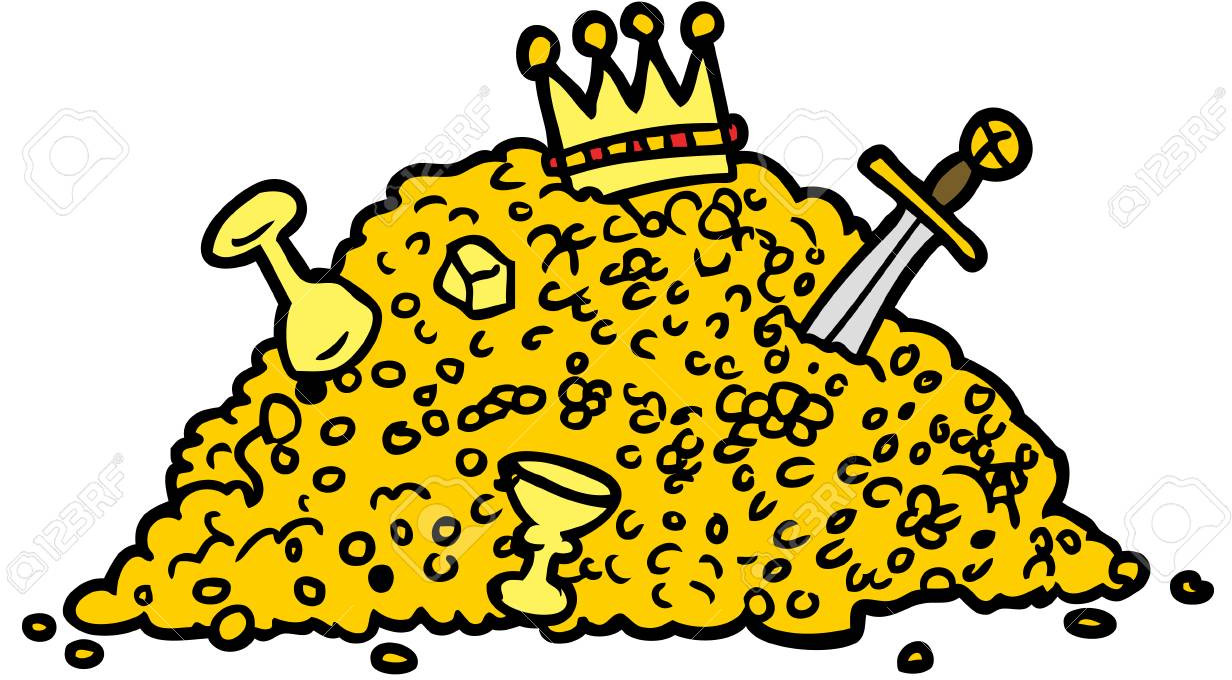}};
\end{scope}
\draw[->,-latex] (.5,-1) -- node[draw=none, fill=white,scale=.75,rotate=28] {decision 1} (2,-.25);
\draw[->,-latex] (2.75,-.25) -- node[draw=none, fill=white,scale=.75,rotate=28] {perception 1} (1.25,-1);
\draw[->,-latex] (.5,2) -- node[draw=none, fill=white,scale=.75,rotate=-28] {decision 2} (2,1.25);
\draw[->,-latex] (2.75,1.25) -- node[draw=none, fill=white,scale=.75,rotate=-28] {perception 2} (1.25,2);
\node[scale=1] at (-.75,-3) {\sc agents};
\node[scale=1] at (5.25,-3) {\sc environment};
\end{tikzpicture}
\caption{Illustration of agents and their environment for the tiger problem.}
\label{fig:agent:environment}
\end{figure}

\paragraph{Notation.}
This paper uses uppercase letter $Y$ to refer to random variables and lowercase letter $y$ to refer to their realizations. The corresponding domains are represented by calligraphic uppercase letter $\mathcal{Y}$. When we want to represent a vector containing all values of $y$ between two integers $t_0$ and $t_1$, we use the notation $y_{t_0:t_1}$. Similarly, $y^{i_0:i_1}$ represents a shorthand notation for the vector $(y^{i_0},\ldots,y^{i_1})$, where $i_0$ and $i_1$ are integers satisfying $i_0\leq i_1$. If $i_1<i_0$, we represent this as $y^{i_1:i_0} \doteq \emptyset$. However, when we want to represent the Cartesian product of all values of domains $\mathcal{Y}^{i_0},\ldots,\mathcal{Y}^{i_1}$ between two integers $i_0$ and $i_1$, we use the notation $\mathcal{Y}^{i_0:i_1} \doteq \mathcal{Y}^{i_0}\times \ldots \times \mathcal{Y}^{i_1}$. We use subscripts to represent control interval indexes and superscripts to index agents. The probability of an event is denoted by $\Pr\{\cdot\}$, and the expectation of a random variable is represented by $\mathbb{E}\{\cdot\}$. When dealing with a finite set $\mathcal{Y}$, we represent the cardinality of $\mathcal{Y}$ with $|\mathcal{Y}|$, and $\mathbb{N}_{\leq |\mathcal{Y}|}$ as a set of integers ranging from 0 to $|\mathcal{Y}|$. Additionally, $\mathcal{P}(\mathcal{Y})$ represents the $(|\mathcal{Y}|-1)$-dimensional real simplex. The set-theoretic support of a function $f\colon \mathcal{Y} \to \mathbb{R}$ is represented by $\mathtt{supp}(f)$, which is the set of points in $\mathcal{Y}$ where $f$ is non-zero. We use the short-hand notations $y^{i_0:}$ and $y^{:i_1}$ to represent the vector $y^{i_0:n}$ and $y^{0:i_1}$, respectively, where $i_0$ and $i_1$ are integers satisfying $i_0,i_1 \in \mathbb{N}_{\leq n}$. Finally, we use the Kronecker delta $\delta_{\cdot}(\cdot) \colon \mathcal{Y}_0\times \mathcal{Y}_1 \to \{0,1\}$ in our calculations.

\subsection{Master Games as POSGs}
This subsection will discuss finite-horizon POSGs and their solution concepts, including the \citeauthor{bellman}'s equations. To help readers better understand the concepts presented in this paper, we will use a simple toy master game called the tiger problem, illustrated in Figure \ref{fig:agent:environment}.

\begin{example}
In the tiger problem, there are two doors, one hiding a tiger named Hobbes, and the other hiding a pile of gold. Two agents, Calvin and Susie, aim to determine which door contains the treasure so that they can claim it as quickly as possible. The tiger can be on the left, referred to as $x_{\textsc{tl}}$, or on the right, referred to as $x_{\textsc{tr}}$. The agents can perform one of three actions: $u_{\textsc{ol}}$ (open left), $u_{\textsc{or}}$ (open right), and $u_{\textsc{l}}$ (listen). Even after opening a door, there are only two possible observations for each agent: to hear the tiger on the left, $z_{\textsc{hl}}$, or to hear the tiger on the right, $z_{\textsc{hr}}$. If both agents perform $u_{\textsc{l}}$, the state of the world remains unchanged. However, any other joint action results in a transition to either state with a probability of $0.5$, effectively resetting the problem. When the world is in state $x_{\textsc{tl}}$, the joint action $(u_{\textsc{l}},u_{\textsc{l}})$ results in observation $z_{\textsc{hl}}$ for any agent with a probability of 0.85 and observation $z_{\textsc{hr}}$ with a probability of 0.15. This applies similarly for state $x_{\textsc{tr}}$. Regardless of the state, all other joint actions result in either observation with a probability of 0.5.
\end{example}

\begin{figure}[!ht]
\centering
\begin{tikzpicture}[->,>=triangle 45,shorten >=2pt,auto,node distance=4cm,semithick]
\tikzstyle{every state}=[draw=black,text=black,,inner color= white,outer color= white,draw= black,text=black, drop shadow]
\tikzstyle{place}=[circle,thick,draw=sthlmBlue,fill=blue!20,minimum size=6mm]
\tikzstyle{placesplit}=[circle split,thick,draw=sthlmBlue,fill=sthlmBlue!20,minimum size=6mm]
\tikzstyle{red place}=[place,draw=sthlmRed,fill=sthlmRed!20]
\tikzstyle{green place}=[place,draw=sthlmGreen,fill=sthlmGreen!20]
\tikzstyle{green placesplit}=[placesplit,draw=sthlmGreen,fill=sthlmGreen!20]
\draw[rounded corners, sthlmRed, fill=sthlmRed!10] (-2.2,-1) rectangle (12,1);
\node[fill=white,text=black,draw=none,scale=.7] at (-.75,1.35) {$p^{u,z}_{x'}$};
\node[fill=white,text=black,draw=none,scale=.7] at (-.75,-1.35) {$p^{u,z}_{x'}$};
\node[fill=white,text=black,draw=none,scale=.7] at (3.5,-1.35) {$p^{u,z}_{x'}$};
\node[fill=white,text=black,draw=none,scale=.7] at (3.5,1.35) {$p^{u,z}_{x'}$};
\node[fill=white,text=black,draw=none,scale=.7] at (7.5,-1.35) {$p^{u,z}_{x'}$};
\node[fill=white,text=black,draw=none,scale=.7] at (7.5,1.35) {$p^{u,z}_{x'}$};
\node[initial,state,place] (S0) {State};
\node[state,place] (S1) [right of=S0] {State};
\node[state,place] (S2) [ right of=S1] {State};
\node (S3) [ right of=S2] {};
\node[state,green placesplit, scale=.61] (O0) [below of=S0,node distance=3.75cm] { Observation$^1$ \nodepart{lower} Reward$^1$};
\node[state,green placesplit, scale=.61] (O1) [below of=S1,node distance=3.75cm] { Observation$^1$ \nodepart{lower} Reward$^1$};
\node[state,green placesplit, scale=.61] (O2) [below of=S2,node distance=3.75cm] { Observation$^1$ \nodepart{lower} Reward$^1$};
\node[state,green placesplit, scale=.61] (O3) [above of=S0,node distance=3.75cm] { Observation$^2$ \nodepart{lower} Reward$^2$};
\node[state,green placesplit, scale=.61] (O4) [above of=S1,node distance=3.75cm] { Observation$^2$ \nodepart{lower} Reward$^2$};
\node[state,green placesplit, scale=.61] (O5) [above of=S2,node distance=3.75cm] { Observation$^2$ \nodepart{lower} Reward$^2$};
\node[state,red place] (A0) [below right of=O0,node distance=2cm] {\footnotesize Action$^1$};
\node[state,red place] (A1) [right of=A0] {\footnotesize Action$^1$};
\node (A2) [right of=A1] {};
\node[state,red place] (A3) [above right of=O3,node distance=2cm] {\footnotesize Action$^2$};
\node[state,red place] (A4) [right of=A3] {\footnotesize Action$^2$};
\node (A5) [right of=A4] {};
\node (Time) at (-1,-5) {Time};
\node (T0) [below of=A0,node distance=1.25cm] {$t$};
\node (T1) [below of=A1,node distance=1.25cm] {$t$+$1$};
\node (T2) [below of=A2,node distance=1.25cm] {$t$+$2$};
\node (N0) at (2.8,4.15) {};
\node (N1) at (2.8,-4.75) {};
\draw[-,dotted] (N0)-- (N1);
\draw[-,dotted] (-2,-1.1)--(12,-1.1);
\draw[-,dotted] (-2,1.1)--(12,1.1);
\node (N2) at (6.8,4.15) {};
\node (N3) at (6.8,-4.75) {};
\draw[-,dotted] (N2)--(N3);
\node[fill=sthlmRed!10,text=black,draw=none,scale=.7] at (11,-.75) {Hidden};
\node[fill=white,text=black,draw=none,scale=.7] at (10.75,-1.35) {Agent $1$'s viewpoint};
\node[fill=white,text=black,draw=none,scale=.7] at (10.75,1.35) {Agent $2$'s viewpoint};
\path (S0) edge node[midway, fill=sthlmRed!10,text=black,draw=none,scale=.7, above=-5pt] {$p^u_{x,x'}$} (S1)
edge node {} (O0)
edge node {} (O3)
(S1) edge node[midway, fill=sthlmRed!10,text=black,draw=none,scale=.7, above=-5pt] {$p^u_{x,x'}$} (S2)
edge node {} (O1)
edge node {} (O4)
(S2) edge node[midway, fill=sthlmRed!10,text=black,draw=none,scale=.7, above=-5pt] {$p^u_{x,x'}$} (S3)
edge node {} (O2)
edge node {} (O5)
(A0) edge [out=90, in=-180] node {} (O1)
edge [out=90, in=-155] node {} (S1)
(A1) edge [out=90, in=-180] node {} (O2)
edge [out=90, in=-155] node {} (S2)
(A3) edge [out=-90, in=-180] node {} (O4)
edge [out=-90, in=-205] node {} (S1)
(A4) edge [out=-90, in=-180] node {} (O5)
edge [out=-90, in=-205] node {} (S2);
\end{tikzpicture}
\caption{A graphical model of a two-agent, partially observable stochastic game.}
\label{graphical:model:zs:posg}
\end{figure}
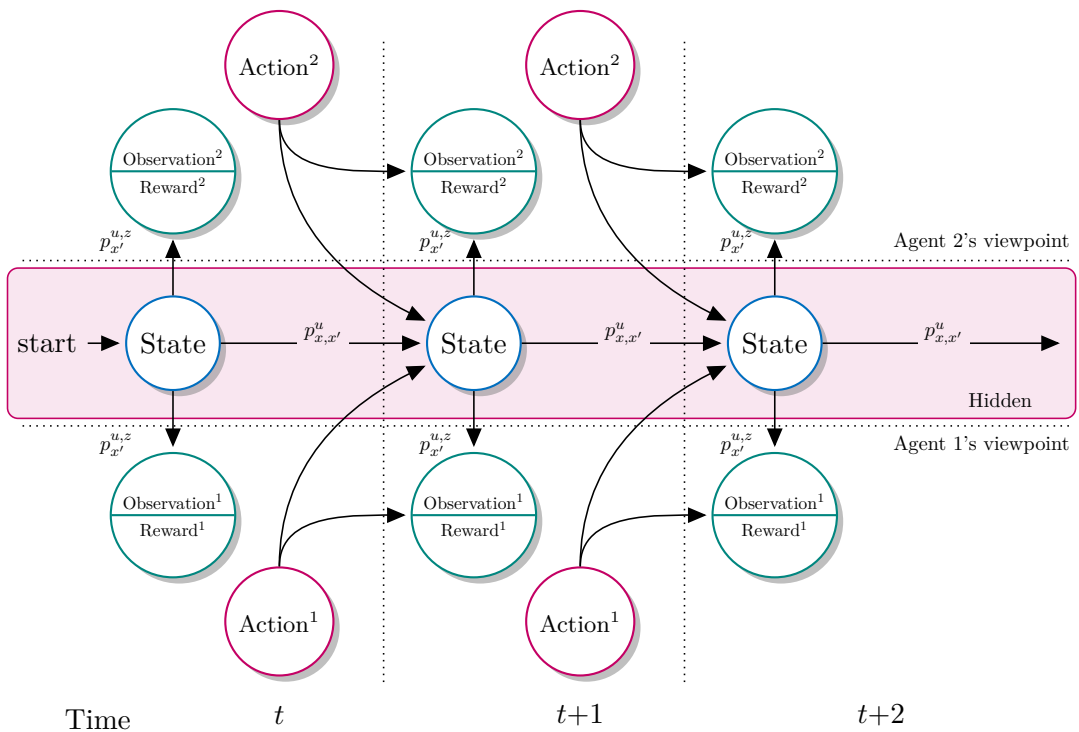

In a POSG, schematically depicted in Figure \ref{graphical:model:zs:posg}, agents begin with uncertain knowledge about the game's initial belief state, which is determined by the initial belief state $s_0$, \ie a distribution over hidden states. At each time step, the game is in a particular state $x$, and each agent $i$ can choose an action $u^i$ from their action space. The game then responds by transitioning to a new state $x'$, based on the actions of all agents $u$, and a transition probability distribution $(p^u_{x,x'})_{x,x'\in\mathcal{X},u\in\mathcal{U}}$. Additionally, each agent $i$ receives an immediate reward $r^i_{xu}$ and individual observation $z^i$, based on joint observation $z$ and an observation probability distribution $(p^{u,z}_{x'})_{x'\in\mathcal{X},u\in\mathcal{U},z\in\mathcal{Z}}$, providing them with potentially incomplete and noisy knowledge about the game's current state $x'$. This process continues until the horizon is exhausted.

\begin{definition}
A $n$-agent, partially observable stochastic game (POSG) is given by a tuple $M \doteq (\mathcal{X},\mathcal{U},\mathcal{Z},p,r)$. The set $\mathcal{X}$ represents a finite collection of hidden states, denoted as $x$. The finite set of joint actions is denoted as $\mathcal{U} = \mathcal{U}^{1:}$, and the actions taken by the agents are denoted as $u=u^{1:n}$. The finite set of observations is denoted as $\mathcal{Z} = \mathcal{Z}^{1:}$, and the observations made by the agents are denoted as $z=z^{1:n}$. The transition function $p_{x,x'}^{u,z}$ is given by $p^u_{x,x'} \cdot p^{u,z}_{x'}$, and specifies the probability of the process state being $x'$ and having observation $z$ after taking action $u$ in hidden state $x$. The reward function $r_{x,u}$ specifies the reward vector $r^{1:n}_{x,u}$ received by the agents after they take joint action $u$ in state $x$.
\end{definition}

\paragraph{Assumptions.} Throughout the paper, we assume that the process described by $M$ operates under the following conditions: the observation space of each agent $i$ is made up of a public subspace and a private subspace, denoted as $\mathcal{Z}^i = \mathcal{W} \times \tilde{\mathcal{Z}}^i$, \cf \citep{rethinking}. The public observations are accessible to all agents. Additionally, rewards are bounded on both sides, with a positive scalar $c$ such that $\|r^{\cdot}_{\cdot,\cdot}\|_\infty \leq c$. The planning horizon, referred to as $\ell$, is finite. However, if an infinite-horizon solution is required, any $\ell$-horizon solution where $\ell = \lceil \log_\gamma{({(1-\gamma)\epsilon}/{c})} \rceil$ can provide an $\epsilon$-close approximation to any infinite-horizon solution. This finding holds for a discount factor $\gamma\in [0, 1)$ and scalar $\epsilon>0$.

\begin{figure}[!ht]
\centering
\begin{tikzpicture}[scale=1.25]
\node (q_1) at (-3, -1.5) {$\varsigma_0^i$};
\node (q_11) at (-4.5, -3) {$\varsigma_1^{i,1}$};
\node (q_111) at (-5.25,-4.5) {};
\node (q_112) at (-3.75,-4.5) {};
\node (q_113) at (-4.2,-4) {};
\node (q_114) at (-4.8,-4) {};
\node (q_12) at (-1.5, -3) {$\varsigma_1^{i,|\mathcal{U}^i||\mathcal{Z}^i|}$};
\node (q_121) at (-2.25,-4.5) {};
\node (q_122) at (-0.75,-4.5) {};
\node (q_123) at (-1.2,-4) {};
\node (q_124) at (-1.8,-4) {};
\node (q_13) at (-3.3, -2.5) {};
\node (q_14) at (-2.7, -2.5) {};
\node[sthlmRed] at (-3,-5) {$\cdots$};
\node[sthlmRed] at (-3,-3) {$\cdots$};
\draw[-, very thin] (-5.25,-1.25) -- (-5.5,-1.25) -- node[fill=white, rotate=90]{\emph{depth} $\ell$} (-5.5,-5) -- (-5.25,-5);
\draw[-, very thin] (-5,-1.25) -- (-5,-1) -- node[fill=white]{\emph{branching} $|\mathcal{U}^i| |\mathcal{Z}^i|$} (-1,-1) -- (-1,-1.25);
\path[-, draw=black, thick, sthlmRed] (q_1) edge (q_11)
edge (q_12)
edge (q_13)
edge (q_14)
(q_11) edge (q_111)
edge (q_112)
edge (q_113)
edge (q_114)
(q_12) edge (q_121)
edge (q_122)
edge (q_123)
edge (q_124);
\end{tikzpicture}
\caption{A $\ell$-step policy tree captures a sequence of $\ell$ time steps, each of which can be conditioned on the outcome of previous decisions. Each node is labelled with the decision made if it is reached.}
\label{fig:policies}
\end{figure}
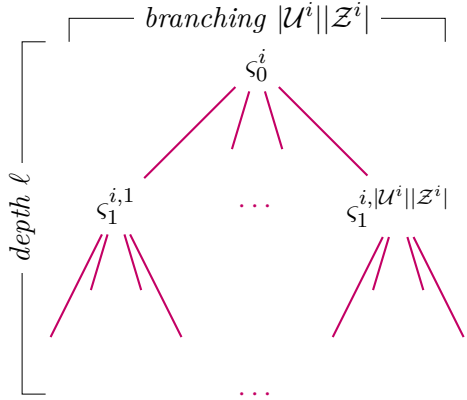

The goal of solving game $M$ is to identify a joint policy that satisfies a specific criterion based on the problem at hand. This involves determining a series of individual decision rules for each agent, denoted by ${a}^{1:n}_{0:\ell-1}$ or simply ${a}_{0:}$. Under a joint policy, every time step $t$, each agent $i$ follows a unique private decision rule ${a}^i_t\colon o^i_t \mapsto \varsigma^i_t$, where $\varsigma^i_t\in \mathcal{P}(\mathcal{U}^i)$, which is dependent on the $t$-step private history $o^i_t \doteq (u^i_{0:t-1},z^i_{1:t})$, with the first private history being $o^i_0 \doteq \emptyset$, as shown in Figure \ref{fig:policies}. Every joint policy $a_{0:}$ will have a performance index induced by state-value functions $\upsilon^{a_{0:}}_{M,\gamma,0:}$ that are defined as follows.

\begin{definition}
\label{def:state:value:fct:m}
For every agent $i$, the $t$-step state-value function $\upsilon^{i,a_{t:}}_{M,\gamma,t}\colon \mathcal{X}\times \mathcal{O}_t \to \mathbb{R}$ under the joint policy $a_{t:}$ is given by: for any arbitrary hidden state $x_t$ and joint history $o_t$,
\begin{align*}
\upsilon^{i,a_{t:}}_{M,\gamma,t}(x_t,o_t) &\doteq\textstyle \mathbb{E}\{\sum_{\tau=t}^{\ell-1} \gamma^{\tau-t} \cdot r^i_{X_\tau,U_\tau} | x_t,o_t,a_{t:}\},
\end{align*}
with boundary condition $\upsilon^{\cdot}_{M,\gamma,\ell}(\cdot) \doteq 0$.
\end{definition}

One can extend Definition \ref{def:state:value:fct:m} by abuse of notation to hold over posterior probability distributions over hidden states and histories as follows.

\begin{definition}
For every agent $i$, the $t$-step state-value function $\upsilon^{i,a_{t:}}_{M,\gamma,t}\colon \mathcal{P}(\mathcal{X}\times \mathcal{O}_t) \to \mathbb{R}$ under joint policy $a_{t:}$ is given by: for any distribution $s_t$,
\begin{align*}
\upsilon^{i,a_{t:}}_{M,\gamma,t}(s_t) &\doteq\textstyle \mathbb{E}\{\sum_{\tau=t}^{\ell-1} \gamma^{\tau-t} \cdot r^i_{X_\tau,U_\tau} | s_t,a_{t:}\},
\end{align*}
with boundary condition $\upsilon^{i,\cdot}_{M,\gamma,\ell}(\cdot) \doteq 0$.
\end{definition}

The performance index for joint policy $a_{0:}$ at initial belief state $s_0$ is given by $\upsilon^{a_{0:}}_{M,\gamma,0}(s_0)$. As discussed below, defining the \citeauthor{bellman}'s equations underlying state-value functions under a fixed joint policy will prove useful.

\begin{restatable}{lem}{lembellmanequationm}
\label{lem:bellman:equation:m}
The $t$-step state-value function $\upsilon^{i,a_{t:}}_{M,\gamma,t}\colon \mathcal{X}\times \mathcal{O}_t \to \mathbb{R}$ satisfies a recursion formula, known as \citeauthor{bellman}'s equations, \ie for any arbitrary hidden state $x_t$ and joint history $o_t$,
\begin{align}
\upsilon^{i,a_{t:}}_{M,\gamma,t} (x_t,o_t)&= \textstyle \sum_{u_t\in \mathcal{U}}~ a_t(u_t|o_t)\cdot q^{i,a_{t+1:}}_{M,\gamma,t}(x_t,o_t,u_t) \nonumber\\
q^{i,a_{t+1:}}_{M,\gamma,t} (x_t,o_t,u_t) &\doteq \textstyle r^i_{x_tu_t} + \gamma\sum_{x_{t+1}\in \mathcal{X}}\sum_{z\in \mathcal{Z}} p_{x_t,x_{t+1}}^{u_t,z_{t+1}} \cdot \upsilon^{i,a_{t+1:}}_{M,\gamma,t+1}(x_{t+1},(o_t,u_t,z_{t+1})),\label{eq:lem:bellman:equation:m}
\end{align}
with boundary condition $\upsilon^{i,\cdot}_{M,\gamma,\ell}(\cdot) = q^{i,\cdot}_{M,\gamma,\ell}(\cdot) \doteq 0$.
\end{restatable}

\begin{proof}
The proof proceeds by induction. It is obviously true at the last time step $t=\ell-1$, \ie for every hidden state $x_t$, and joint history $o_t$,
\begin{align*}
\upsilon^{a_{\ell-1}}_{M,\gamma, \ell-1}(x_{\ell-1},o_{\ell-1}) &\doteq \mathbb{E}\{ r_{x_{\ell-1},U_{\ell-1}}| x_{\ell-1},o_{\ell-1}, a_{\ell-1} \}.
\end{align*}
Expanding the expectation over joint actions results in the following expression:
\begin{align*}
\upsilon^{a_{\ell-1}}_{M,\gamma, \ell-1}(x_{\ell-1},o_{\ell-1}) &= { \textstyle\sum_{u_{\ell-1}\in \mathcal{U}}~ a_{\ell-1}(u_{\ell-1}|o_{\ell-1})\cdot q^{\cdot}_{M,\gamma,\ell-1}(x_{\ell-1},o_{\ell-1}, u_{\ell-1})},
\end{align*}
where $q^{\cdot}_{M,\gamma,\ell-1}(x_{\ell-1},o_{\ell-1}, u_{\ell-1}) \doteq r_{x_{\ell-1}, u_{\ell-1}}$.
Suppose the induction hypothesis holds from time step $t+1$ onward. The recursive scheme holds for time step $t$, for hidden state $x_t$ and joint history $o_t$,
\begin{align*}
\upsilon^{a_{t:}}_{M,\gamma, t}(x_t,o_t)
&\doteq\textstyle { \mathbb{E}\left\{ \sum_{\tau=t}^{\ell-1}~\gamma^{\tau-t}\cdot r_{X_\tau,U_\tau} | x_t, o_t, a_{t:} \right\}}.
\end{align*}
Splitting the immediate reward and the future rewards leads to the following expression:
\begin{align*}
\upsilon^{a_{t:}}_{M,\gamma, t}(x_t,o_t)&=\textstyle {\mathbb{E}\left\{ r_{x_t,U_t} + \gamma \sum_{\tau=t+1}^{\ell-1}~\gamma^{\tau-t-1}\cdot r_{X_\tau,U_\tau} | x_t , o_t , a_{t:} \right\}}.
\end{align*}
Expanding the expectation of the future rewards over all histories at time step $t+1$ yields:
\begin{align*}
\upsilon^{a_{t:}}_{M,\gamma, t}(x_t,o_t)&=\textstyle \mathbb{E}\left\{r_{x_t,U_t} + \gamma \mathbb{E}{\left\{ \sum_{\tau=t+1}^{\ell-1}~\gamma^{\tau-t-1}\cdot r_{X_\tau,U_\tau}| X_{t+1}, (o_t,U_t,Z_{t+1}), a_{t+1:} \right\}} | x_t , o_t, a_t\right\}.
\end{align*}
The inner expectation is the state-value function of hidden states and histories at time step $t+1$, by Definition \ref{def:state:value:fct:m} of $\upsilon^{a_{t+1:}}_{M,\gamma, t+1}$, \ie
\begin{align*}
\upsilon^{a_{t:}}_{M,\gamma, t}(x_t,o_t)&=\textstyle \mathbb{E}\{ r_{x_t,U_t} + \gamma \upsilon^{a_{t+1:}}_{M,\gamma, t+1}(X_{t+1},(o_t,U_t,Z_{t+1})) | x_t, o_t, a_t \}.
\end{align*}
The application of definition of $q^{a_{t+1:}}_{M,\gamma, t}$ in \eqref{eq:lem:bellman:equation:m} produces the following expression:
\begin{align*}
\upsilon^{a_{t:}}_{M,\gamma, t}(x_t,o_t)&=\textstyle \mathbb{E}\{q^{a_{t+1:}}_{M,\gamma, t}(x_t,o_t,U_t) | x_t, o_t, a_t \}.
\end{align*}
Finally, the expansion of the expectation over joint actions produces the target expression:
\begin{align*}
\upsilon^{a_{t:}}_{M,\gamma, t}(x_t,o_t)&=\textstyle \sum_{u_t\in \mathcal{U}} a_t(u_t|o_t)\cdot q^{a_{t+1:}}_{M,\gamma, t}(x_t,o_t,u_t).
\end{align*}
Hence, the proof holds for any arbitrary time step $t$. Which ends the proof.
\end{proof}

\subsection{Solution Concepts}
After learning how to calculate the state- and action-value functions for a fixed joint policy, we also need to determine how to identify a suitable policy for each agent based on state-value functions. Unfortunately, no universal definition of good policies can be applied to all multi-agent sequential decision-making problems. Instead, this paper proposes solution concepts that refine the notion of a best-response policy, including but not limited to Nash Equilibrium, optimal joint policy, or strong Stackelberg joint policy. To better understand the concept of best response, consider the scenario where an agent knows how other agents will act. In this case, the multi-agent sequential decision-making problem becomes simple, and the agent is left with a single-agent problem of selecting the utility-maximizing policy based on the other agents' fixed policies, \ie a best-response policy.

\begin{definition}
A best-response policy of agent $i$ to joint policy $a^{\neg i}_{0:}\in \mathcal{A}_{0:}^{\neg i}$ of the other agents is a policy $a^i_{0:}\in \mathcal{A}_{0:}^i$ satisfying $\upsilon^{i,(a^i_{0:}, a^{\neg i}_{0:})}_{M,\gamma,0}(s_0) \geq \upsilon^{i,( \underline{a}^i_{0:}, a^{\neg i}_{0:})}_{M,\gamma,0}(s_0)$, for all $ \underline{a}^i_{0:}\in \mathcal{A}_{0:}^i$. Set $\mathcal{A}_{0:}^i(a^{\neg i}_{0:})$ of best-responses \wrt $a^{\neg i}_{0:}$ is made of all the possible randomized policies and is thus infinite except in degenerate cases where a unique (deterministic) best-response, \ie
\begin{align*}
\mathcal{A}_{0:}^i(a^{\neg i}_{0:})&
\doteq
\{
a^i_{0:} \in \mathcal{A}_{0:}^i |
\upsilon^{i,( a^i_{0:}, a^{\neg i}_{0:})}_{M,\gamma,0}(s_0) \geq \upsilon^{i,(\underline{a}^i_{0:}, a^{\neg i}_{0:})}_{M,\gamma,0}(s_0), \forall \underline{a}^i_{0:} \in \mathcal{A}_{0:}^i
\}.
\end{align*}
\end{definition}

The properties of best-response policies are numerous and significant. Firstly, they yield the same performance index for a specific joint policy $a^{\neg i}_{0:}$ of the other agents. Consequently, deterministic (or pure) policies are just as effective as randomized policies. Additionally, if a randomized policy is a best-response, every deterministic policy within its support must also be a best response. Perfect recall ensures that these properties hold, regardless of whether the best-response is deterministic, mixed, or stochastic.

Up to this point, it is assumed that an agent knows other agents' policies. However, an agent may not know how other agents will act. Nevertheless, it is possible to utilize the concept of best-response to define all solution concepts of interest in this paper. A notable example of such concepts is the most central notion in non-cooperative game theory, namely the Nash Equilibrium. A formal definition follows.

\begin{definition}
A joint policy $a_{0:}\in \mathcal{A}_{0:}$ is a Nash Equilibrium if the policy of any agent is a best response to those of the other agents, \ie $a^i_{0:}\in \mathcal{A}_{0:}^i(a^{\neg i}_{0:})$ for every agent $i$.
\end{definition}

A Nash Equilibrium exists in any game with a finite set of agents and a finite deterministic policy space per agent. A Nash equilibrium refers to a stable joint policy where no agent is willing to deviate from its policy, provided they know the policies other agents will follow. Multiple Nash equilibria can exist, all of which are self-enforcing if agents adhere to such behaviour. Each agent's best interest is to stick with its policy. However, it is important to note that these equilibria may not be optimal for the agents, meaning they may not provide optimal values for any agent and can often offer significantly different values. All solution concepts considered onward are refinements of Nash Equilibria. Optimization problems determining the best-response policies are slave games for a master game. To solve a master game, it is necessary to address a multitude of slave games, which may potentially be infinite in number. This is due to numerous Nash equilibria involving mutually beneficial best-response policies.

In this paper, we will be placing a particular emphasis on three significant subclasses of POSGs: zero-sum POSG (\textit{zs-}POSG), common-payoff POSG (\textit{dec-}POMDP), and Stackelberg POSG (\textit{st-}POSG). Our primary objective is to exhibit the underlying structure of their optimal value functions.

\begin{definition}
A two-player POSG $M$ is said to be a \textit{zs-}POSG if the sum of the rewards of both players equals zero, \ie $r^1 = - r^2$. 
\end{definition}

Optimally solving \textit{zs-}POSG $M$ aims at finding a Nash equilibrium. Interestingly, every Nash equilibrium of a \textit{zs-}POSG $M$ will yield the same performance index at the initial belief state $s_0$: $\upsilon_{M,\gamma,0}^{i,*}(s_0) \doteq\textstyle \max_{a^i_{0:}\in \mathcal{A}_{0:}^{i}} \min_{a^{\neg i}_{0:}\in \mathcal{A}_{0:}^{\neg i}} \upsilon_{M,\gamma,0}^{i,a_{0:}}(s_0)$.

\begin{definition}
A $n$-player POSG $M$ is a \textit{dec}-POMDP if all players have the same rewards, \ie $r^1 = \ldots = r^n$.
\end{definition}

Optimally solving \textit{dec-}POMDP $M$ aims at finding an \emph{optimal joint policy}, \ie a Nash equilibrium whose performance index at the initial belief state $s_0$ is: $\upsilon_{M,\gamma,0}^{*}(s_0) \doteq\textstyle \max_{a^{1:}_{0:}\in \mathcal{A}_{0:}^{1:}} \upsilon_{M,\gamma,0}^{a_{0:}}(s_0)$.

\begin{definition}
A two-player POSG $M$ is a \textit{st-}POSG if (the leader) player $1$ publicly announces its policy to (the follower) player $2$ beforehand. 
\end{definition}

Optimally solving \textit{st-}POSG $M$ aims at finding a \emph{strong Stackelberg equilibrium}, \ie a Nash equilibrium whose performance index is: $\upsilon_{M,\gamma,0}^{1,*}(s_0) \doteq\textstyle \max_{a^1_{0:}\in \mathcal{A}_{0:}^1} \max_{a^2_{0:}\in \mathcal{A}_{0:}^2(a^1_{0:})} \upsilon_{M,\gamma,0}^{1,a_{0:}}(s_0)$. It is worth noting that solving a POSG optimally, and in particular solving \textit{zs-}POSG, \textit{dec-}POMDP, and \textit{st-}POSG optimally, aims at finding a Nash equilibrium.

\subsection{Slave Games As POMDPs}
In this subsection, we will formally introduce the concept of slave games as a mean to determine a best-response policy for an agent given a master game expressed as a POSG alongside the joint policy that the other agents follow. We will also delve into the expression of slave games as POMDPs and their corresponding solution concepts, including \citeauthor{bellman}'s optimality equations. To aid comprehension, we will utilize a particular scenario of the tiger problem as an example, depicted in Figure \ref{fig:agent:environment:slave}.

\begin{figure}[!ht]
\centering
\begin{tikzpicture}[scale=1.35]
\def \agentCentral{(0,-3) ellipse (1cm and 0.9cm);
\node[inner sep=0pt] (russell) at (0,-3.25) {\includegraphics[width=.06125\textwidth]{figures/TrustedThirdParty.png}};}
\def \agentA{(-.5,3) ellipse (1cm and 0.9cm);
\node[inner sep=0pt] (calvin) at (-.5,2.7) {\includegraphics[width=.03725\textwidth]{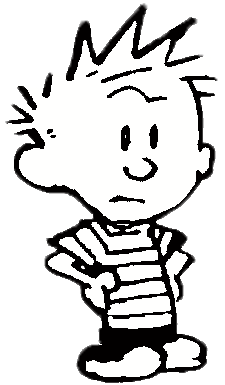}};}
\def \agentD{(-.5,0) ellipse (1cm and 0.9cm);
\node[inner sep=0pt] (susie) at (-.5,-.25) {\includegraphics[width=.03725\textwidth]{figures/Susie.png}};}
\draw[rounded corners, fill=white] (3.5,-2.5) rectangle +(3.5,5.5);
\node (agentCentral) at (-2,.5) {};
\node (agentA) at (-.25,2.25) {};
\node (agentD) at (-.25,-1) {};
\begin{scope}[scale=.5]
\def \agentA{(-1.5,4.5) ellipse (1cm and 0.9cm);
\node[inner sep=0pt] (calvin) at (-1.5,4.2) {\includegraphics[width=.075\textwidth]{figures/Calvin_BW.png}};}
\def \agentD{(-1.5,-2) ellipse (1cm and 0.9cm);
\node[inner sep=0pt] (susie) at (-1.5,-2.25) {\includegraphics[width=.075\textwidth]{figures/Susie.png}}; }
\draw[draw=white] \agentA;
\node[draw=sthlmYellow, fill=yellow!20, rectangle callout, text centered, text width=3cm, text=gray, scale=.7] at (-3,6) {\textcolor{gray}{My policy is \textcolor{sthlmBlue}{$a^{\neg i}_{0:}$}}};
\draw[draw=white] \agentD;
\node[draw=sthlmYellow, fill=yellow!20, rectangle callout, text centered, text width=3cm, text=gray, scale=.7] at (-2.85,-0.25) {\textcolor{gray}{I know \textcolor{sthlmBlue}{$a^{\neg i}_{0:}$} is Calvin's policy}};
\node[inner sep=0pt] (hobbes) at (11.6,3) {\includegraphics[width=.075\textwidth]{figures/Hobbes.png}};
\node[inner sep=0pt] (door1) at (9,3) {\includegraphics[width=.075\textwidth]{figures/Door.png}};
\node[inner sep=0pt] (door2) at (9,-2) {\includegraphics[width=.075\textwidth]{figures/Door.png}};
\node[inner sep=0pt] (treasure) at (11.75,-3) {\includegraphics[width=.075\textwidth]{figures/Treasure.jpg}};
\end{scope}
\draw[->,-latex] (.5,-1) -- node[draw=none, fill=white,scale=.75,rotate=28] {decision 1} (2,-.25);
\draw[->,-latex] (2.75,-.25) -- node[draw=none, fill=white,scale=.75,rotate=28] {perception 1} (1.25,-1);
\draw[->,-latex] (.5,2) -- node[draw=none, fill=white,scale=.75,rotate=-28] {decision 2} (2,1.25);
\draw[->,-latex] (2.75,1.25) -- node[draw=none, fill=white,scale=.75,rotate=-28] {perception 2} (1.25,2);
\node[scale=1] at (-.75,-3) {\sc agents};
\node[scale=1] at (5.25,-3) {\sc environment};
\end{tikzpicture}
\caption{Illustration of agents and their environment for the tiger problem where Calvin acts according to a public and fixed policy. }
\label{fig:agent:environment:slave}
\end{figure}

\begin{example}
In this particular scenario of the tiger problem, Calvin follows a predetermined policy $a_{0:}^{\neg i}$, while Susie is aware of this. Nevertheless, Susie must still identify the correct door to claim the treasure as quickly as possible, \cf Figure \ref{fig:agent:environment:slave}.
\end{example}

\begin{definition}
Consider the master game $M$. A slave game $M(a^{\neg i}_{0:})$ is characterized by the tuple $(\mathcal{X}, \mathcal{O}, \mathcal{U}^i,\mathcal{Z}^i, p, r^i)$. Here, $\mathcal{X}$ represents a finite space of hidden states, $\mathcal{O}$ denotes the space of hidden joint histories, $\mathcal{U}^i$ is the finite action space, and $\mathcal{Z}^i$ represents the finite observation space. All other aspects of $M(a^{\neg i}_{0:})$ are consistent with those in $M$.
\end{definition}

Let us consider master game $M$. Within $M$, there exists a slave game denoted by $M(a^{\neg i}_{0:})$. Agent $i$ is part of a POMDP, as depicted in Figure \ref{graphical:model:slave:problem}. At each time step, the state of the world is determined by the underlying state in $M$ and the joint history of all agents. However, agent $i$ can only access its private history, not the underlying state nor the histories of the other agents. Regardless, agent $i$ can still utilize their history to make informed decisions, which results in an immediate reward signal and a new observation. This process repeats in a new underlying state and joint history at each successive time step until the game concludes.

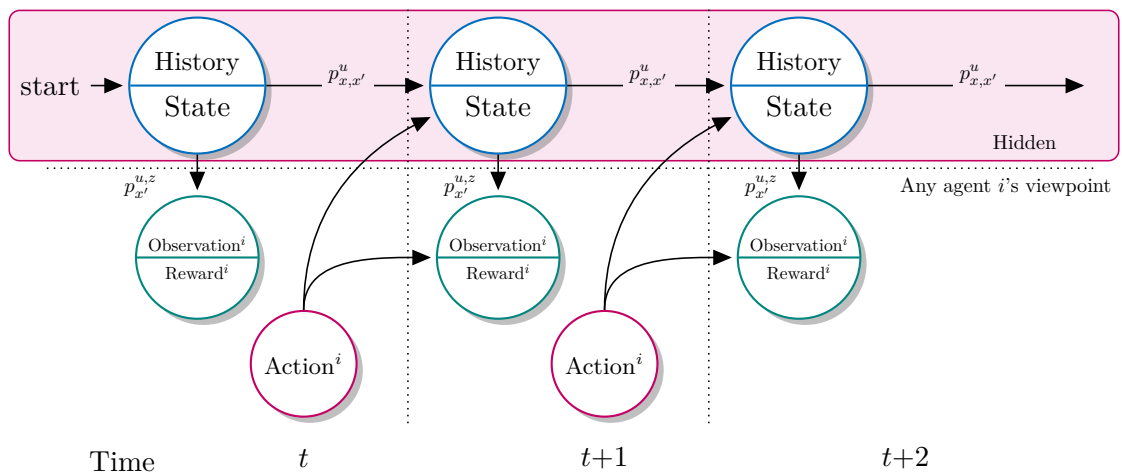
\begin{figure}[!ht]
\centering
\begin{tikzpicture}[->,>=triangle 45,shorten >=2pt,auto,node distance=4cm,semithick]
\tikzstyle{every state}=[draw=black,text=black,,inner color= white,outer color= white,draw= black,text=black, drop shadow]
\tikzstyle{place}=[circle,thick,draw=sthlmBlue,fill=blue!20,minimum size=6mm]
\tikzstyle{placesplit}=[circle split,thick,draw=sthlmBlue,fill=sthlmBlue!20,minimum size=6mm]
\tikzstyle{red place}=[place,draw=sthlmRed,fill=sthlmRed!20]
\tikzstyle{green place}=[place,draw=sthlmGreen,fill=sthlmGreen!20]
\tikzstyle{green placesplit}=[placesplit,draw=sthlmGreen,fill=sthlmGreen!20]
\draw[rounded corners, sthlmRed, fill=sthlmRed!10] (-2.5,-1) rectangle (12.25,1);
\node[fill=white,text=black,draw=none,scale=.7] at (-.75,-1.35) {$p^{u,z}_{x'}$};
\node[fill=white,text=black,draw=none,scale=.7] at (3.5,-1.35) {$p^{u,z}_{x'}$};
\node[fill=white,text=black,draw=none,scale=.7] at (7.5,-1.35) {$p^{u,z}_{x'}$};
\node[initial,state,placesplit] (S0) {\small History \nodepart{lower} State};
\node[state,placesplit] (S1) [right of=S0] {\small History \nodepart{lower} State};
\node[state,placesplit] (S2) [ right of=S1] {\small History \nodepart{lower} State};
\node (S3) [ right of=S2] {};
\node[state,green placesplit, scale=.61] (O0) [below of=S0,node distance=3.75cm] { Observation$^i$ \nodepart{lower} Reward$^i$};
\node[state,green placesplit, scale=.61] (O1) [below of=S1,node distance=3.75cm] { Observation$^i$ \nodepart{lower} Reward$^i$};
\node[state,green placesplit, scale=.61] (O2) [below of=S2,node distance=3.75cm] { Observation$^i$ \nodepart{lower} Reward$^i$};
\node[state,red place] (A0) [below right of=O0,node distance=2cm] {\footnotesize Action$^i$};
\node[state,red place] (A1) [right of=A0] {\footnotesize Action$^i$};
\node (A2) [right of=A1] {};
\node (Time) at (-1,-5) {Time};
\node (T0) [below of=A0,node distance=1.25cm] {$t$};
\node (T1) [below of=A1,node distance=1.25cm] {$t$+$1$};
\node (T2) [below of=A2,node distance=1.25cm] {$t$+$2$};
\node (N0) at (2.8,1.15) {};
\node (N1) at (2.8,-4.75) {};
\draw[-,dotted] (N0)-- (N1);
\draw[-,dotted] (-2,-1.1)--(12,-1.1);
\node (N2) at (6.8,1.15) {};
\node (N3) at (6.8,-4.75) {};
\draw[-,dotted] (N2)--(N3);
\node[fill=sthlmRed!10,text=black,draw=none,scale=.7] at (11,-.75) {Hidden};
\node[fill=white,text=black,draw=none,scale=.7] at (10.75,-1.35) {Any agent $i$'s viewpoint};
\path (S0) edge node[midway, fill=sthlmRed!10,text=black,draw=none,scale=.7, above=-5pt] {$p^{u}_{x,x'}$} (S1)
edge node {} (O0)
(S1) edge node[midway, fill=sthlmRed!10,text=black,draw=none,scale=.7, above=-5pt] {$p^{u}_{x,x'}$} (S2)
edge node {} (O1)
(S2) edge node[midway, fill=sthlmRed!10,text=black,draw=none,scale=.7, above=-5pt] {$p^{u}_{x,x'}$} (S3)
edge node {} (O2)
(A0) edge [out=90, in=-180] node {} (O1)
edge [out=90, in=-155] node {} (S1)
(A1) edge [out=90, in=-180] node {} (O2)
edge [out=90, in=-155] node {} (S2);
\end{tikzpicture}
\caption{The graphical model of a slave game as a partially observable Markov decision process from any agent $i$'s viewpoint. Blue circles represent the states of the problem, including the state of the underlying environment and the joint histories. Green circles denote both private observations and rewards agent $i$ received.Finally, red circles are private actions that agent $i$ took.}
\label{graphical:model:slave:problem}
\end{figure}

The goal of solving the slave game $M(a^{\neg i}_{0:})$ is to find a best-response policy $a^i_{0:}$ for agent $i$. A best-response policy will maximize the expected $\gamma$-discounted cumulative reward, starting from the initial belief state $s_0$, providing the other agents act according to $a^{\neg i}_{0:}$. Every best-response policy will have the same performance index at the initial belief state $s_0$.

\begin{definition}
\label{def:best:response:vf:optimal}
Consider a slave game $M(a^{\neg i}_{0:})$ \wrt the master game $M$. The $t$-step optimal value function $\upsilon_{M(a^{\neg i}_{0:}),\gamma,t}^i\colon \mathcal{O}^i_t \to \mathbb{R}$ of $M(a^{\neg i}_{0:})$ is given by: for any arbitrary private history $o^i_t$,
\begin{align*}
\upsilon_{M(a^{\neg i}_{0:}),\gamma,t}^i(o^i_t) &\doteq\textstyle \max_{a^i_{t:} \in \mathcal{A}^i_{t:}}~ \mathbb{E}\{ \textstyle{\sum_{\tau=t}^{\ell-1}}~\gamma^{\tau-t} \cdot r^i_{X_\tau,U_\tau} | o^i_t, a^i_{t:},a^{\neg i}_{0:} \}.
\end{align*} 
with boundary condition $\upsilon_{M(a^{\neg i}_{0:}),\gamma,\ell}^i(\cdot) \doteq 0$.
\end{definition}

One can extend Definition \ref{def:best:response:vf:optimal} by abuse of notation to hold over posterior probability distributions over private histories of agent $i$, namely marginal occupancy states, as follows.

\begin{definition}
\label{def:best:response:vf:optimal:extended}
Consider a slave game $M(a^{\neg i}_{0:})$ \wrt the master game $M$. The $t$-step optimal value function $\upsilon_{M(a^{\neg i}_{0:}),\gamma,t}^i\colon\mathcal{P}( \mathcal{O}^i_t )\to \mathbb{R}$ of any best-response policy is given by: for any arbitrary marginal occupancy state $s^{\mathtt{m},i}_t$,
\begin{align*}
\upsilon_{M(a^{\neg i}_{0:}),\gamma,t}^i(s^{\mathtt{m},i}_t) &\doteq\textstyle \max_{a^i_{t:} \in \mathcal{A}^i_{t:}}~ \mathbb{E}\{ \textstyle{\sum_{\tau=t}^{\ell-1}}~\gamma^{\tau-t} \cdot r^i_{X_\tau,U_\tau} | s^{\mathtt{m},i}_t, a^i_{t:},a^{\neg i}_{0:} \}.
\end{align*}
with boundary condition $\upsilon_{M(a^{\neg i}_{0:}),\gamma,\ell}^i(\cdot) \doteq 0$.
\end{definition}

When dealing with a slave game formulated as a POMDP, it is important to note that the agent only has access to partial information about the game state. Therefore, relying on the state to reason in such a process is not feasible. Enumerating and evaluating all $\ell$-step policy alternatives is computationally expensive. Instead, the best-response policy can be computed based on the private history of that agent.

\begin{restatable}{lem}{lembellmanhistorypomdp}
\label{lem:bellman:history:pomdp}
The $t$-step optimal state-value function $\upsilon_{M(a^{\neg i}_{0:}),\gamma,t}^i\colon \mathcal{O}^i_t\to \mathbb{R}$ of a slave game $M(a^{\neg i}_{0:})$ satisfies a recursion formula, known as \citeauthor{bellman}'s optimality equations, \ie for any arbitrary private history $o^i_t$, we have:
\begin{align*}
\upsilon_{M(a^{\neg i}_{0:}),\gamma,t}^i(o^i_t) &= \textstyle \max_{u^i_t\in \mathcal{U}^i}~q_{M(a^{\neg i}_{0:}),\gamma,t}^i(o^i_t,u^i_t),\\
q_{M(a^{\neg i}_{0:}),\gamma,t}^i(o^i_t,u^i_t) &\doteq \mathbb{E}\{ r^i_{X_t,U_t} + \gamma \upsilon_{M(a^{\neg i}_{0:}),\gamma,t+1}^i(\langle o^i_t,u^i_t,Z^i_{t+1}\rangle) | o^i_t,u^i_t,a^{\neg i}_{0:}\}
\end{align*}
with boundary condition $\upsilon_{M(a^{\neg i}_{0:}),\gamma,\ell}^i(\cdot) = q_{M(a^{\neg i}_{0:}),\gamma,t}^i(\cdot,\cdot) \doteq 0$.
\end{restatable}

\begin{proof}
The proof starts by defining the best-response policy from time step $t$ onward, providing joint policy of the other agents $a^{\neg i}_{0:}\in \mathcal{A}^{\neg i}_{0:}$, \cf Definition \ref{def:best:response:vf:optimal}:
\begin{align*}
\upsilon_{M(a^{\neg i}_{0:}),\gamma,t}^i(o^i_t) &\doteq\textstyle \max_{a^i_{t:} \in \mathcal{A}^i_{t:}}~ \mathbb{E}\{ \textstyle{\sum_{\tau=t}^{\ell-1}}~\gamma^{\tau-t} \cdot r^i_{X_\tau,U_\tau} | o^i_t, a^i_{t:},a^{\neg i}_{0:} \}.
\end{align*}
Let $\upsilon_{M(a^{\neg i}_{0:}),\gamma,t}^{i,a^i_{t:}}\colon o^i_t \mapsto \mathbb{E}\{ \textstyle{\sum_{\tau=t}^{\ell-1}}~\gamma^{\tau-t} \cdot r^i_{X_\tau,U_\tau} | o^i_t, a^i_{t:},a^{\neg i}_{0:} \}$ be a $t$-step value function under policy $a^i_{t:}$ in the slave game $M(a^{\neg i}_{0:})$. Re-arranging terms by decoupling the immediate reward from the future rewards leads to the following expression:
\begin{align*}
\upsilon_{M(a^{\neg i}_{0:}),\gamma,t}^i(o^i_t)
&= \textstyle \max_{(u^i_t,a^{\neg i}_{t+1:}) \in \mathcal{A}^i_{t:}}~ \mathbb{E}\{ r^i_{X_t,U_t} + \gamma \upsilon^{i,a^i_{t+1:}}_{M(a^{\neg i}_{0:}),\gamma,t+1}(\langle o^i_t,u^i_t,Z^i_{t+1}\rangle) | o^i_t, (u^i_t,a^i_{t+1:},a^{\neg i}_{0:}) \}.
\end{align*}
It is worth noticing that the policy $(u^i_t,a^i_{t+1:})$ rooted at history $o^i_t$ can be written recursively as a tree $\langle u^i_t, (a^{i,z^i}_{t+1:})_{z^i\in \mathcal{Z}^i}\rangle$ with initial action being $u^i_t$ and the tree after seeing private observation $z^i$ being $a^{i,z^i}_{t+1:}$. Consequently, replacing $(u^i_t,a^i_{t+1:})$ by $\langle u^i_t, (a^{i,z^i}_{t+1:})_{z^i\in \mathcal{Z}^i}\rangle$ results in the following expression:
\begin{align*}
\upsilon_{M(a^{\neg i}_{0:}),\gamma,t}^i(o^i_t)
&= \textstyle \max_{\langle u^i_t, (a^{i,z^i}_{t+1:})_{z^i\in \mathcal{Z}^i}\rangle \in \mathcal{A}^i_{t:}}~ \mathbb{E}\{ r^i_{X_t,U_t} + \gamma \upsilon^{i,a^{i,Z^i_{t+1}}_{t+1:}}_{M(a^{\neg i}_{0:}),\gamma,t+1}(o^i_t,u^i_t,Z^i_{t+1}) | o^i_t, (u^i_t,a^i_{t+1:},a^{\neg i}_{0:}) \}.
\end{align*}
One can split the $\max$ operator into two parts, the first one over private actions $\max_{u^i_t\in \mathcal{U}^i}$ and the other one over the future decision rules $\max_{(a^{i,z^i}_{t+1:})_{z^i\in \mathcal{Z}^i}\in \mathcal{A}^i_{t+1:}}$, as follows:
\begin{align*}
\upsilon_{M(a^{\neg i}_{0:}),\gamma,t}^i(o^i_t) &= \textstyle \max_{u^i_t \in \mathcal{U}^i} \mathbb{E}\{ r^i_{X_t,U_t} + \gamma \max_{a^{i,Z_{t+1}^i}_{t+1:} \in \mathcal{A}^i_{t+1:}}~ \upsilon^{i,a^{i,Z_{t+1}^i}_{t+1:}}_{M(a^{\neg i}_{0:}),\gamma,t+1}(\langle o^i_t,u^i_t,Z^i_{t+1}\rangle) | o^i_t, (u^i_t,a^{\neg i}_{0:}) \}.
\end{align*}
By Definition \ref{def:best:response:vf:optimal}, we know that the following holds:
\begin{align*}
\upsilon^i_{M(a^{\neg i}_{0:}),\gamma,t+1}((o^i_t,u^i_t,z^i_{t+1}))& \doteq\textstyle \max_{a^{i,z_{t+1}^i}_{t+1:} \in \mathcal{A}^i_{t+1:}}~ \upsilon^{i,a^{i,z_{t+1}^i}_{t+1:}}_{M(a^{\neg i}_{0:}),\gamma,t+1}(\langle o^i_t,u^i_t,z^i_{t+1}\rangle),
\end{align*}
which leads us to the following expression:
\begin{align*}
\upsilon_{M(a^{\neg i}_{0:}),\gamma,t}^i(o^i_t) &= \textstyle \max_{u^i_t \in \mathcal{U}^i} \mathbb{E}\{ r^i_{X_t,U_t} + \gamma \upsilon^i_{M(a^{\neg i}_{0:}),\gamma,t+1}(\langle o^i_t,u^i_t,Z^i_{t+1}\rangle) | o^i_t, (u^i_t,a^{\neg i}_{0:}) \}.
\end{align*}
Which ends the proof.
\end{proof}

The tabular nature of state- and action-value functions in a slave game results in linearity over marginal occupancy states.

\begin{restatable}{cor}{corlinearoptimalslave}
\label{cor:linear:optimal:slave}
The $t$-step optimal state-value function $\upsilon_{M(a^{\neg i}_{0:}),\gamma,t}^i\colon\mathcal{P}( \mathcal{O}^i_t )\to \mathbb{R}$ is linear over posterior probability distributions of private histories, \ie for any arbitrary marginal occupancy state $s^{\mathtt{m},i}_t$,
\begin{align*}
\upsilon_{M(a^{\neg i}_{0:}),\gamma,t}^i(s^{\mathtt{m},i}_t) &= \textstyle \sum_{o^i_t\in \mathcal{O}_t^i}~s^{\mathtt{m},i}_t(o^i_t)\cdot \upsilon_{M(a^{\neg i}_{0:}),\gamma,t}^i(o^i_t).
\end{align*}
Similarly, the $t$-step optimal action-value function $q_{M(a^{\neg i}_{0:}),\gamma,t}^i\colon \mathcal{P}(\mathcal{O}^i_t ) \times \mathcal{A}^i \to \mathbb{R}$ is linear over posterior probability distributions of private histories, \ie for any arbitrary marginal occupancy state $s^{\mathtt{m},i}_t$ and private decision rule $a^i_t$,
\begin{align*}
q_{M(a^{\neg i}_{0:}),\gamma,t}^i(s^{\mathtt{m},i}_t,a^i_t) &=\textstyle \sum_{o^i_t\in \mathcal{O}_t^i}~s^{\mathtt{m},i}_t(o^i_t)\sum_{u^i_t\in \mathcal{U}^i}~a^i_t(u^i_t| o^i_t) \cdot q_{M(a^{\neg i}_{0:}),\gamma,t}^i(o^i_t,u^i_t).
\end{align*}
with boundary condition $\upsilon_{M(a^{\neg i}_{0:}),\gamma,\ell}^i(\cdot) = q_{M(a^{\neg i}_{0:}),\gamma,\ell}^i(\cdot,\cdot) \doteq 0$.
\end{restatable}

\begin{proof}
The proof follows directly from that in Lemma \ref{lem:bellman:history:pomdp}.
\end{proof}

The optimal state- and action-value functions in slave games are shown to be linear across probability distributions of private histories in Corollary \ref{cor:linear:optimal:slave}. However, it is important to note that this finding cannot be directly applied to the optimal value function of the master game. This is because the master game requires optimal value function formulations for all the slave games, encompassing all agents' histories. Efforts have been made to apply Corollary \ref{cor:linear:optimal:slave} in showing the uniform continuity properties of optimal value functions of master games. A notable instance is incorporating a weak convexity property for the optimal value functions of \textit{zs-}POSGs by \citet{structure}. To gain a more comprehensive understanding of this, it is essential to define some basilar concepts.

\begin{definition}
Let $s_t\in \mathcal{P}(\mathcal{X}\times \mathcal{O}_t)$ be a posterior probability distribution over hidden states and joint histories, namely an occupancy state, and define the corresponding factors:
\begin{itemize}
\item marginal occupancy state $s_t^{\mathtt{m},i}$, \ie for any private history $o^i_t$, $$s_t^{\mathtt{m},i}(o^i_t) = \textstyle \sum_{x_t\in \mathcal{X}}\sum_{o^{\neg i}_t\in \mathcal{O}^{\neg i}_t} s_t(x_t, o^i_t,o^{\neg i}_t),$$ 
\item conditional occupancy state $s_t^{\mathtt{c},i}$, \ie for any hidden state $x_t$ and joint history $\langle o^i_t,o^{\neg i}_t\rangle$, $$s_t^{\mathtt{c},i}(x_t,o^{\neg i}_t|o^i_t) = s_t(x_t,o^i_t,o^{\neg i}_t)/ s_t^{\mathtt{m},i}(o^i_t).$$
\end{itemize}
We shall use short-hand notation $s_t \doteq s_t^{\mathtt{m},i}\odot s_t^{\mathtt{c},i}$ to represent occupancy state $s_t$ given by: for any arbitrary hidden state $x_t$ and joint history $o_t$, $s_t(x_t,o_t) = s_t^{\mathtt{m},i}(o^i_t)\cdot s_t^{\mathtt{c},i}(x_t,o^{\neg i}_t|o^i_t)$.
\end{definition}

Intuitively, the $t$-step conditional occupancy state $s_t^{\mathtt{c},i}$ is a summary of the suffix joint policy $a^{\neg i}_{:t-1}$ together with the initial belief $s_0$. Consequently, a slave game $M(a_{0:}^{\neg i})$ can be expressed as a sequence of slave subgames, one slave subgame $M(s_t^{\mathtt{c},i}, a^{\neg i}_{t:})$ for each time step $t$. So, solving slave game $M(a_{0:}^{\neg i})$ relies on solving sequence of slave subgames $\langle M(a_{0:}^{\neg i}), M(s_{1}^{\mathtt{c},i}, a^{\neg i}_{1:}), \ldots, M(s_{\ell-1}^{\mathtt{c},i}, a^{\neg i}_{\ell-1:}) \rangle$ through \citeauthor{bellman}'s principle of optimality. We are now ready to state the convexity properties of interest.

\begin{restatable}{thm}{thmconvexhistory}[\citeauthor{structure}]
\label{thm:convex:history}
Consider the \emph{zs}-POSG $M$. The $t$-step optimal state-value function $\upsilon_{M,\gamma,t}^{i,*}\colon \mathcal{P}(\mathcal{X}\times\mathcal{O}_t) \to \mathbb{R}$ is convex over marginal occupancy states for a fixed conditional occupancy state, \ie for any occupancy state $s_t \doteq s_t^{\mathtt{m},\neg i}\odot s_t^{\mathtt{c},\neg i}$,

\begin{align*}
\upsilon_{M,\gamma,t}^{i,*}(s_t) &= \textstyle \max_{a^{i}_{t:}\in \mathcal{A}_{t:}^{i}} \upsilon^i_{M(s_t^{\mathtt{c},\neg i},a^{i}_{t:}),\gamma,t}(s^{\mathtt{m},\neg i}_t),
\end{align*}
where $\upsilon^i_{M(s_t^{\mathtt{c},\neg i},a^{i}_{t:}),\gamma,t} \doteq -\upsilon^{\neg i}_{M(s_t^{\mathtt{c},\neg i},a^{i}_{t:}),\gamma,t}$ and $M(s_t^{\mathtt{c},\neg i},a^{i}_{t:})$ describes the slave subgame given the conditional occupancy state $s_t^{\mathtt{c},\neg i}$ and private decision rule $a^{i}_{t:}$.
\end{restatable}

Theorem \ref{thm:convex:history} stipulates that the optimal value functions for \textit{zs-}POSGs are convex functions of marginal occupancy states, given a fixed conditional occupancy state. However, this property alone does not facilitate value generalization across different occupancy states. To address this challenge, \citet{zerosum} integrated Lipschitz continuity and the convexity property in Theorem \ref{thm:convex:history}, thus allowing for the generalization across occupancy states.

\begin{restatable}{thm}{thmconvexlipschitzhistory}[\citeauthor{zerosum}]
\label{thm:convex:lipschitz:history}
Consider a \emph{zs}-POSG $M$. The $t$-step optimal state-value function $\upsilon_{M,\gamma,t}^{i,*}\colon \mathcal{P}(\mathcal{X}\times \mathcal{O}_t)\to \mathbb{R}$ is $\kappa_t$-Lipschitz continuous over occupancy states. There exists a collection $\mathcal{V}^i_{t}$ of linear functions across marginal occupancy states such that for any occupancy state $s_t \doteq s_t^{\mathtt{m},\neg i}\odot s_t^{\mathtt{c},\neg i}$, and vector $\upsilon^{\neg i}_{M(\tilde{s}_t^{\mathtt{c},\neg i},a^{ i}_{t:}),\gamma,t} \in \mathcal{V}^i_t$ associated to conditional occupancy state $\tilde{s}_t^{\mathtt{c},\neg i}$, the following holds:
\begin{align*}
\upsilon_{M,\gamma,t}^{i,*}(s_t) &\leq \upsilon^i_{M(\tilde{s}_t^{\mathtt{c},\neg i},a^{i}_{t:}),\gamma,t}(s^{\mathtt{m},\neg i}_t) + \kappa_t \|s_t - s_t^{\mathtt{m},\neg i}\odot \tilde{s}_t^{\mathtt{c},\neg i}\|,
\end{align*}
where $\upsilon^i_{M(\tilde{s}_t^{\mathtt{c},\neg i},a^{i}_{t:}),\gamma,t} \doteq -\upsilon^{\neg i}_{M(\tilde{s}_t^{\mathtt{c},\neg i},a^{i}_{t:}),\gamma,t}$, and $M( s_t^{\mathtt{c},\neg i},a^{i}_{t:})$ describes the slave subgame given conditional occupancy state $s_t^{\mathtt{c},\neg i}$ and private decision rule $a^{i}_{t:}$, and $\kappa_t \doteq \frac{1-\gamma^{\ell-t}}{1-\gamma} c$ is the Lipschitz constant at time step $t$.
\end{restatable}

Using Theorem \ref{thm:convex:lipschitz:history} allows the generalization of values across various occupancy states by harmonizing local convexity and Lipschitz continuity properties. Nonetheless, it is worth noting that the Lipschitz constant can result in imprecise approximations, which may ultimately impede algorithmic efficiency.

\section{Problem Reformulations}
\label{sec:reformulations}
The existing literature on master and slave games primarily focuses on an agent-centric perspective. However, this approach poses significant challenges in terms of knowledge transfer from Markov games to POSGs. In particular, the agent-centric perspective is not always geared to ease optimal decision-making. As such, adopting a more generalized approach to knowledge transfer is necessary to address this limitation. This section examines the dynamics of master and slave games from the perspective of a trusted third party acting on behalf of all involved agents simultaneously. To accomplish this goal, we shall introduce two reformulations: plan-time Markov games, which use the history available at planning time to the trusted third party about master or slave games, and occupancy Markov games, which leverage sufficient statistics to represent the history available to the trusted third party. We shall demonstrate that both reformulations maintain the ability to solve the original master or slave games optimally.

\begin{figure}[!ht]
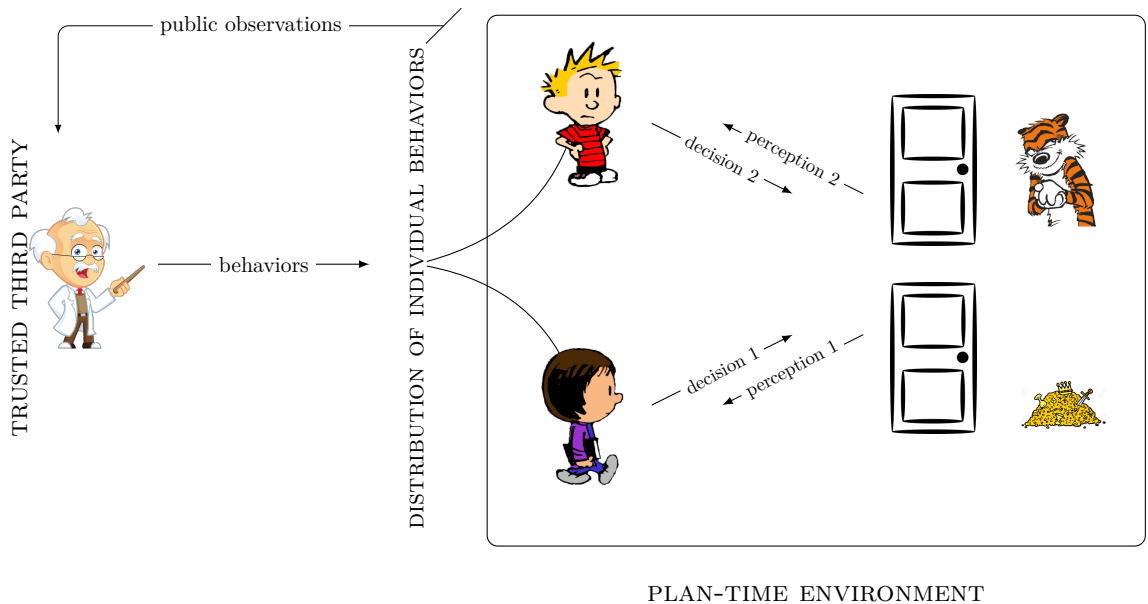

\centering
\begin{tikzpicture}[scale=1.25]
\draw[very thin, color=black] (-1.9,2.85) -- (-1.525, 3.225);
\draw[->, -latex,rounded corners,very thin, color=black] (-1.7125,3.0375) --node[scale=.75, fill=white] {public observations} (-5.8125, 3.0375) -- (-5.8125, 1.8875);
\draw[rounded corners, fill=white] (-1.25,-2.5) rectangle +(7,5.65);
\node (agentCentral) at (-2,.5) {};
\node (agentA) at (-.25,2.25) {};
\node (agentD) at (-.25,-1) {};
\path[->, -latex] (agentCentral) edge [bend right] (agentA)
edge [bend left] (agentD);
\begin{scope}[scale=.5]
\def \agentCentral{(-11,1) ellipse (1cm and 0.9cm);
\node[scale=1, rotate=90] at (-12.45,0.5) {\textcolor{black}{\sc trusted third party}};
\node[inner sep=0pt] (russell) at (-11,.75) {\includegraphics[width=.125\textwidth]{figures/TrustedThirdParty.png}};}
\def \agentA{(-.5,4.5) ellipse (1cm and 0.9cm);
\node[inner sep=0pt] (calvin) at (-.5,4.2) {\includegraphics[width=.075\textwidth]{figures/Calvin_Color.png}};}
\def \agentD{(-.5,-2) ellipse (1cm and 0.9cm);
\node[inner sep=0pt] (susie) at (-.5,-2.25) {\includegraphics[width=.075\textwidth]{figures/Susie.png}};}
\draw[draw=white] \agentCentral;
\node (prof) at (-11,2.3) {};
\draw[draw=white] \agentA;
\draw[draw=white] \agentD;
\node[inner sep=0pt] (hobbes) at (9.6,3) {\includegraphics[width=.075\textwidth]{figures/Hobbes.png}};
\node[inner sep=0pt] (door1) at (7,3) {\includegraphics[width=.075\textwidth]{figures/Door.png}};
\node[inner sep=0pt] (door2) at (7,-1) {\includegraphics[width=.075\textwidth]{figures/Door.png}};
\node[inner sep=0pt] (treasure) at (9.75,-2) {\includegraphics[width=.075\textwidth]{figures/Treasure.jpg}};
\end{scope}
\draw[->,-latex, color=black] (-4.75,.5) -- node[draw=none, fill=white, scale=.75] {behaviors} (-2.5,.5);
\node[scale=1] at (2.25,-3) {\sc plan-time environment};
\draw[->,-latex] (.5,-1) -- node[draw=none, fill=white,scale=.65,rotate=28] {decision 1} (2,-.25);
\draw[->,-latex] (2.75,-.25) -- node[draw=none, fill=white,scale=.65,rotate=28] {perception 1} (1.25,-1);
\draw[->,-latex] (.5,2) -- node[draw=none, fill=white,scale=.65,rotate=-28] {decision 2} (2,1.25);
\draw[->,-latex] (2.75,1.25) -- node[draw=none, fill=white,scale=.65,rotate=-28] {perception 2} (1.25,2);
\node[scale=.9, rotate=90] at (-2,.25) {\sc distribution of individual behaviors};
\end{tikzpicture}
\caption{A plan-time environment for the tiger problem from the perspective of a trusted third party. While each agent knows its decision and perception, which may differ from one agent to the other, the trusted third party has only access to public observations about the plan-time environment.}
\label{fig:multi:single:agent:environment}
\end{figure}

\subsection{Master Game Reformulations}
In this subsection, we will discuss agents delegating their decision-making authority to a trusted third party, also known as a central planner. Within this setting, we shall distinguish between the planning and execution phases. During the planning phase, the central planner is given the power to make decisions on behalf of the agents by using publicly available information about the environment, agents, decisions, and perceptions. It is important to note that the central planner may have information at the planning time that is typically unknown to all agents at the execution stage. At the planning time, the central planner provides the agents with decision rules to follow at each time step, but each agent may only have access to their private decision rules. However, it is important to note that the central planner may not have access to agent histories or decisions at the execution stage. From the central planner viewpoint, the plan-time environment includes the original environment and the agents interacting with it, as shown in Figure \ref{fig:multi:single:agent:environment}. The central planner remains external to the plan-time environment.

Games can be viewed as interactions between agents and their environment or, alternatively, as interactions between a central planner and a plan-time environment, resulting in what is known as a \emph{plan-time Markov game}. However, this alternate perspective does not impact how agents behave in the original game. Both approaches to the game can yield similar solutions under mild conditions. For example, if agents disclose all of their confidential information at every stage, both views are equivalent. Nonetheless, agents may possess divergent and conflicting information about the game at each stage. In plan-time Markov games, a central planner executes a joint decision rule on behalf of all agents and receives public observations available to all agents involved in the game. Figure \ref{fig:ck:game:nf} illustrates the process. The initial belief state, history of public observations, and joint policy are integral components of plan-time histories.

\begin{figure}[!ht]
\centering
\begin{tikzpicture}[->,-latex,auto,node distance=3.75cm,semithick, square/.style={regular polygon,regular polygon sides=4}]
\tikzstyle{every state}=[draw=black,text=black,inner color= white,outer color= white,draw= black,text=black, drop shadow]
\tikzstyle{place}=[thick,draw=sthlmBlue,fill=blue!20,minimum size=12mm, opacity=.5]
\tikzstyle{red place}=[square,place,draw=sthlmRed,fill=sthlmLightRed,minimum size=18mm]
\tikzstyle{green place}=[diamond,place,draw=sthlmGreen,fill=sthlmLightGreen,minimum size=15mm]
\node[fill=white, scale=.75] (T) at (-3.5,-3.75) {};
\node[fill=white, scale=.75] (T0) at (0,-3.75) {$0$};
\node[fill=white, scale=.75] (T1) [right of=T0,node distance=3.75cm, fill=white] {$t$};
\node[fill=white, scale=.75] (T2) [right of=T1,node distance=3.75cm] {$t+1$};
\node[fill=white, scale=.75] (T3) [right of=T2,fill=white, node distance=3.75cm] {$t+2$};
\node[fill=white, scale=.75] (T4) [right of=T3,fill=white, node distance=2.5cm] {$\ldots$};
\draw[->,-latex,dashed,very thin, color=black,anchor=mid] (T3) -- (T4);
\draw[->,-latex, very thin, color=black, anchor=mid] (T2) -- (T3);
\draw[->,-latex, very thin, color=black, anchor=mid] (T1) -- (T2);
\draw[->,-latex,dashed,very thin, color=black,anchor=mid] (T0) -- (T1);
\draw[->,-latex,rounded corners,very thin, color=black,anchor=mid] (T) node[fill=white, scale=.75]{Time} -- (T0);
\node[state,place, scale=.75] (S0) {$y_0$};
\node[state,place, scale=.75] (S1) [right of=S0] {$y_t$};
\node[state,place, scale=.75] (S2) [ right of=S1] {$y_{t+1}$};
\node[state,place, scale=.75] (S3) [ right of=S2] {$y_{t+2}$};
\node[, scale=.75] (S4) [ right of=S3,node distance=2.5cm] {$\cdots$};
\path[very thin] (S0) edge[dashed] node[midway,text=black,draw=none,scale=.7, above=-5pt] {} (S1)
(S1) edge node[midway,text=black,draw=none,scale=.7, above=-5pt] {} (S2)
(S2) edge node[midway,text=black,draw=none,scale=.7, above=-5pt] {} (S3)
(S3) edge[dashed] (S4);
\node[, scale=.75] (A_SLACK) [below of=S0,node distance=2.6cm] {};
\node[state,red place, scale=.65] (A0) [below right of=A_SLACK,node distance=2cm] {$a_{0}$};
\node[state,red place, scale=.65] (A1) [right of=A0,node distance=4.25cm] {$a_t$};
\node[state,red place, scale=.65] (A2) [right of=A1,node distance=4.25cm] {$a_{t+1}$};
\node[, scale=.75] (A3) [right of=A2,node distance=5.5cm] {};
\path[very thin] (A0) edge [out=90, in=-155, dashed] node[scale=.75] {$w_t$} (S1)
(A1) edge [out=90, in=-155] node[scale=.75] {$w_{t+1}$} (S2)
(A2) edge [out=90, in=-155] node[scale=.75] {$w_{t+2}$ } (S3);
\node[state,green place, scale=.75] (O0) [below of=S0,node distance=2cm] {$r^{1:n}_0$};
\node[state,green place, scale=.75] (O1) [below of=S1,node distance=2cm] {$r^{1:n}_t$};
\node[state,green place, scale=.75] (O2) [below of=S2,node distance= 2cm] {$r^{1:n}_{t+1}$};
\node[state,green place, scale=.75] (O3) [below of=S3,node distance= 2cm] {$r^{1:n}_{t+2}$};
\node[, scale=.75] (O4) [ right of=O3,node distance=2.5cm] {$\cdots$};
\node[inner sep=0pt] (prof) at (-3.35,-1.5) {\includegraphics[width=.1\textwidth]{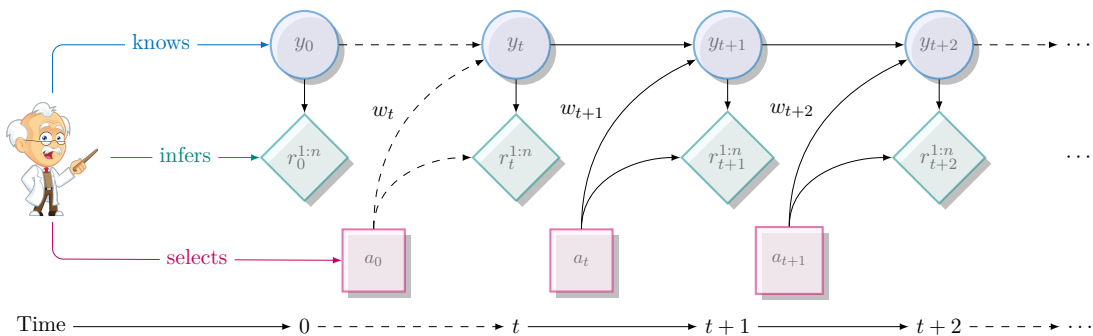}};
\draw[->,-latex,rounded corners,very thin, color=sthlmBlue,anchor=mid] (prof) -- (-3.35,0) --node[draw=none, fill=white, scale=.75] {knows} (S0);
\draw[->,-latex,rounded corners,very thin, color=sthlmGreen,anchor=mid] (prof) --node[draw=none, fill=white, scale=.75] {infers} (O0);
\draw[->,-latex,rounded corners,very thin, color=sthlmRed,anchor=mid] (prof) -- (-3.35,-2.85) --node[draw=none, fill=white, scale=.75] {selects} (A0);
\path[very thin] (S0) edge node {} (O0)
(S1) edge node {} (O1)
(A0) edge[out=90, in=-180, dashed] node {} (O1)
(S2) edge node {} (O2)
(A1) edge[out=90, in=-180] node {} (O2)
(S3) edge node {} (O3)
(A2) edge[out=90, in=-180] node {} (O3);
\end{tikzpicture}
\caption{Influence diagram for a plan-time Markov game.}
\label{fig:ck:game:nf}
\end{figure}

\begin{definition}
A plan-time history for the central planner at time step $t$ includes the initial belief state and the joint policy of all agents up to time step $t-1$, and the history of public observation up to time step $t$, \ie $y_t \doteq \langle s_0,a_{0:t-1}, w_{1:t} \rangle$, with boundary condition $y_0 \doteq s_0$.
\end{definition}

The plan-time histories describe a simultaneous-moves stochastic Markov game, \ie a plan-time Markov game. During the execution phase, the agents do not have access to the plan-time history, the exception being public observations. Consequently, since the central planner makes decisions on behalf of the agents, these decisions cannot rely on the plan-time history. To maintain consistency with the original master game, the plan-time Markov game should not be interpreted outside the context in which it was introduced.

\begin{definition}
\label{def:ckmg}
A $n$-agent, simultaneous-move, plan-time Markov game \wrt~ $M$ is given by a tuple $M' \doteq \langle \mathcal{Y}, \mathcal{A}, \mathcal{W}, \rho_{M'}, \omega_{M'}, r_{M'}\rangle$, where $\mathcal{Y}$ is the space of all plan-time histories; $\mathcal{A}$ is the space of actions describing joint decision rules; $\mathcal{W}$ is the space of public observations; $\rho_{M'}\colon \mathcal{Y} \times \mathcal{A} \times \mathcal{W} \to \mathcal{Y}$ is the deterministic transition rule, where $\rho_{M'}(y_t,a_t,w_{t+1}) \doteq (y_t,a_t,w_{t+1})$; $\omega_{M'}\colon \mathcal{Y} \times \mathcal{A} \times \mathcal{W}\to [0,1]$ is the public observation model, where $\omega_{M'}(w_{t+1}|y_t,a_t) \doteq \Pr\{w_{t+1}|y_t,a_t\}$; and $r_{M'}\colon \mathcal{Y} \times \mathcal{A} \to \mathbb{R}^n$ describes the reward function---\ie $r^i_{M'}(y_t,a_t) \doteq {\textstyle\sum_{x_t\in \mathcal{X}}\sum_{o_t\in \mathcal{O}_t}}~ \Pr\{x_t,o_t|y_t\} \sum_{u_t\in \mathcal{U}} r^i_{x_tu_t} \cdot a_t(u_t|o_t)$.
\end{definition}

Definition \ref{def:ckmg} provides a central-planner viewpoint about the master game $M$. It differs from that of \citet{delayed,sufficient,continuous} by allowing the explicit utilization of public observations. Doing so shall later prove useful to connect our general results to narrower settings. Please note that $M'$ is a Markov game in which the agents cannot observe the states. This means that the solution for this game involves a sequence of actions conditional on public observations, where the state of the game has no impact on the actions taken. The main objective of solving game $M'$ is to find a sequence of actions $a^{1:}_{0:}$ conditional on public observations, also known as a joint policy $a_{0:}$. This policy is determined by a performance index represented by state-value functions. To better understand the relationship between games $M$ and $M'$, we introduce the state-value function under a joint policy $a_{0:}$ when expressed in the plan-time game $M'$.

\begin{definition}
The $t$-step state-value function $\upsilon^{i,a_{t:}}_{M',\gamma,t}\colon \mathcal{Y}_t\to \mathbb{R}$ of agent $i$ under joint policy $a_{t:}$ is given by: for any arbitrary plan-time history $y_t$,
\begin{align*}
\upsilon^{i,a_{t:}}_{M',\gamma,t}(y_t) &=\textstyle \mathbb{E}\{ \sum_{\tau=t}^{\ell-1} \gamma^{\tau-t}\cdot r^i_{M'}(Y_\tau,a_\tau) | y_t,a_{t:}\},
\end{align*}
with boundary condition $\upsilon^{i,\cdot}_{M',\gamma,\ell}(\cdot) \doteq 0$.
\end{definition}

The optimal joint policy for the original master game $M$ can be found by searching for an optimal joint policy in the surrogate master game $M'$.

\begin{restatable}{pro}{prosufficiencyckh}
\label{pro:sufficiency:ckh}
In a master game $M$ and the corresponding plan-time Markov game $M'$, if a joint policy $a_{0:}$ has the optimal performance index in $M'$ at the initial belief state $s_0$, it will also have the same optimal performance index in $M$ at the initial belief state $s_0$.
\end{restatable}

\begin{proof}
First, the solution of interest for both games lies in joint policies. Second, regardless of solution concepts of interest, the state-value functions under a joint policy characterize a good solution. Let $a_{0:}$ be a Nash equilibrium of the plan-time game $M'$, \ie for any agent $i$, $a^i_{0:} \in \mathcal{A}^i_{0:}(a^{\neg i}_{0:})$,
\begin{align*}
\mathcal{A}^i_{0:}(a^{\neg i}_{0:})&\doteq \{a^i_{0:} | a^i_{0:} \in \mathcal{A}^i_{0:} \colon \upsilon^{i,(a^i_{0:},a^{\neg i}_{0:})}_{M',\gamma,0}(y_0) \geq \upsilon^{i,(\underline{a}^i_{0:},a^{\neg i}_{0:})}_{M',\gamma,0}(y_0), \forall \underline{a}^i_{0:} \in \mathcal{A}^i_{0:}\}.
\end{align*}
Let us show that it is also a Nash equilibrium for the master game $M$. To do so, it only suffices to show that $\upsilon^{a_{0:}}_{M',\gamma,0}(y_0) = \upsilon^{a_{0:}}_{M,\gamma,0}(s_0)$ for any arbitrary joint policy $a_{0:}$.
\begin{align}
\upsilon^{a_{0:}}_{M',\gamma,0}(y_0) &\doteq\textstyle \mathbb{E}\{ \sum_{t=0}^{\ell-1} \gamma^t\cdot r_{M'}(Y_t,a_t)|y_0,a_{0:}\}.
\label{eq:pro:sufficiency:ckh:0}
\end{align}
The injection of definition of $r_{M'}$, \cf Definition \ref{def:ckmg}, into \eqref{eq:pro:sufficiency:ckh:0}, results in the following:
\begin{align*}
\upsilon^{a_{0:}}_{M',\gamma,0}(y_0)
&= \textstyle \sum_{t=0}^{\ell-1} \gamma^t\mathbb{E}\left\{ \sum_{x_t\in \mathcal{X}}\sum_{o_t\in \mathcal{O}_t} \Pr\{x_t,o_t| y_t\} \sum_{u_t\in \mathcal{U}} r_{x_t,u_t} \cdot a_t(u_t|o_t)|y_0,a_{0:}\right\}.
\end{align*}
Replace the expression into the bracket with an expectation leads us to the following
\begin{align*}
\upsilon^{a_{0:}}_{M',\gamma,0}(y_0) &= \textstyle \mathbb{E}\left\{ \sum_{t=0}^{\ell-1} \gamma^t \mathbb{E}_{(x_t,o_t)\sim \Pr\{\cdot| Y_t\}, u_t\sim a_t(\cdot|o_t)}\{ r_{x_tu_t} \}|y_0,a_{0:}\right\}.
\end{align*}
Exploiting the linearity of the expectation gives us the following expression
\begin{align*}
\upsilon^{a_{0:}}_{M',\gamma,0}(y_0) &= \textstyle \mathbb{E}_{(x_{0:\ell-1},o_{0:\ell-1},u_{0:\ell-1})\sim \Pr\{\cdot|y_0,a_{0:}\}}\{ \sum_{t=0}^{\ell-1} \gamma^t \cdot r_{x_tu_t} \}\\
&= \textstyle\mathbb{E}\{ \sum_{t=0}^{\ell-1} \gamma^t \cdot r_{X_t,U_t} | y_0,a_{0:}\}\\
&=\textstyle \mathbb{E}\{ \sum_{t=0}^{\ell-1} \gamma^t \cdot r_{X_t,U_t} | s_0,a_{0:}\},&\text{(since $y_0 \doteq s_0$)}\\
&\doteq \upsilon^{a_{0:}}_{M,\gamma,0}(s_0).
\end{align*}
The expression presented above demonstrates that for any arbitrary joint policy $a_{0:}$, both $M'$ and $M$ deliver the same evaluation, that is to say, $\upsilon^{a_{0:}}_{M',\gamma,0}(y_0) = \upsilon^{a_{0:}}_{M,\gamma,0}(s_0)$. As the state-value functions characterize a Nash equilibrium, the fact that $M'$ and $M$ provide the same evaluation for a fixed joint policy indicates that if a joint policy $a_{0:}$ is a Nash equilibrium for $M'$, it is also a Nash equilibrium for $M$. This argument can similarly apply to refinements of Nash equilibria, such as subgame perfect Nash equilibria or optimal joint policies. This concludes the proof.
\end{proof}

According to Proposition \ref{pro:sufficiency:ckh}, solving the plan-time Markov game $M'$ is the same as solving the master game $M$. However, keeping track of the plan-time histories in $M'$ can be hard as they constantly grow with each time step. To simplify things, we can summarize history using the concept of occupancy state, \cf Figure \ref{fig:occupancy:state}.

\begin{definition}
\label{def:occupancy:state}
A $t$-step occupancy state $s_t$ is a posterior probability distribution over hidden states and joint histories conditional upon a $t$-step plan-time history $y_t$, \ie for any arbitrary hidden state $x_t$, joint history $o_t$,
\begin{align*}
s_t(x_t,o_t) &\doteq \Pr\{x_t,o_t |y_t\}.
\end{align*}
\end{definition}

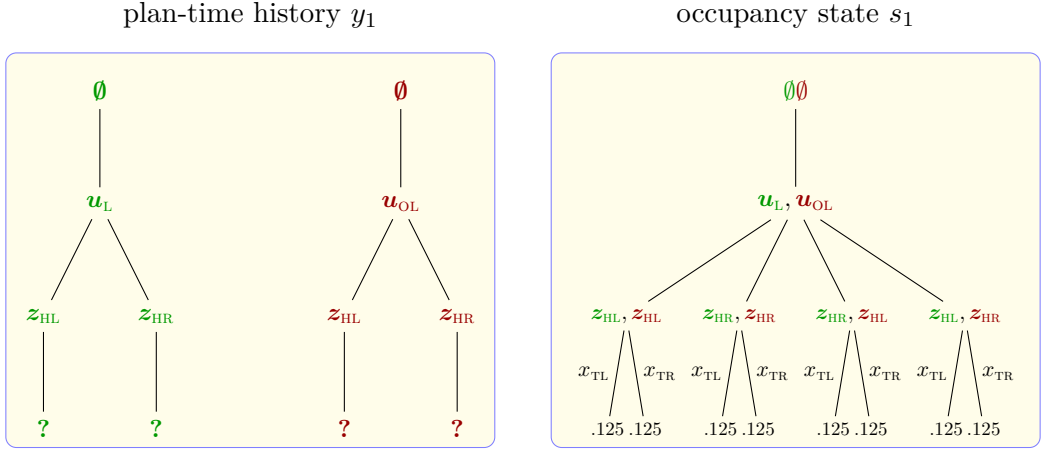
\begin{figure}[!htbp]
\centering
\begin{tikzpicture}[scale=1]
\draw[fill=yellow!20, opacity=.5, rounded corners, draw=blue] (-5.5,-4.75) rectangle (1,.5);
\draw[fill=yellow!20, opacity=.5, rounded corners, draw=blue] (1.75,-4.75) rectangle (8.25,.5);
\node[scale=1] at (-2.25,1) {plan-time history $y_1$};
\node[scale=1] at (5,1) {occupancy state $s_1$};
[sibling distance=2cm,-, thick]
\node[scale=.85] at (-4.25,0) {$\textcolor{green!60!black}{\boldsymbol{\emptyset}}$}
child {node[scale=.85] {$\textcolor{green!60!black}{\boldsymbol{u_\textsc{l}}}$}
child {node[scale=.85] {$\textcolor{green!60!black}{\boldsymbol{z_\textsc{hl}}}$}
child {node[scale=.85] {$\textcolor{green!60!black}{\boldsymbol{?}}$}}}
child {node[scale=.85] { $\textcolor{green!60!black}{\boldsymbol{z_\textsc{hr}}}$}
child {node[scale=.85] {$\textcolor{green!60!black}{\boldsymbol{?}}$}}}};
\node[scale=.85] at (-.25,0) {$\textcolor{red!60!black}{\boldsymbol{\emptyset}}$}
child {node[scale=.85] {$\textcolor{red!60!black}{\boldsymbol{u_\textsc{ol}}}$}
child {node[scale=.85] {$\textcolor{red!60!black}{\boldsymbol{z_\textsc{hl}}}$}
child {node[scale=.85] {$\textcolor{red!60!black}{\boldsymbol{?}}$}}}
child {node[scale=.85] { $\textcolor{red!60!black}{\boldsymbol{z_\textsc{hr}}}$}
child {node[scale=.85] {$\textcolor{red!60!black}{\boldsymbol{?}}$}}}};
\node[scale=.85] (n1) at (5,0) {$\textcolor{green!60!black}{\emptyset}\textcolor{red!60!black}{\emptyset}$}
child {node[scale=.85] (n2) {$\textcolor{green!60!black}{\boldsymbol{u_\textsc{l}}},\textcolor{red!60!black}{\boldsymbol{u_\textsc{ol}}}$}
child {node[scale=.75] {$\textcolor{green!60!black}{\boldsymbol{z_\textsc{hl}}},\textcolor{red!60!black}{\boldsymbol{z_\textsc{hl}}}$}
[sibling distance=.5cm,-]
child {node[scale=.65] {.125} edge from parent node[left, scale=.75] {$x_\textsc{tl}$}}
child {node[scale=.65] {.125} edge from parent node[right, scale=.75] {$x_\textsc{tr}$}}}
child {node[scale=.75] { $\textcolor{green!60!black}{\boldsymbol{z_\textsc{hr}}},\textcolor{red!60!black}{\boldsymbol{z_\textsc{hr}}}$}
[sibling distance=.5cm,-]
child {node[scale=.65] {.125} edge from parent node[left, scale=.75] {$x_\textsc{tl}$}}
child {node[scale=.65] {.125} edge from parent node[right, scale=.75] {$x_\textsc{tr}$}}}
child {node[scale=.75] { $\textcolor{green!60!black}{\boldsymbol{z_\textsc{hr}}},\textcolor{red!60!black}{\boldsymbol{z_\textsc{hl}}}$}
[sibling distance=.5cm,-]
child {node[scale=.65] {.125} edge from parent node[left, scale=.75] {$x_\textsc{tl}$}}
child {node[scale=.65] {.125} edge from parent node[right, scale=.75] {$x_\textsc{tr}$}}}
child {node[scale=.75] { $\textcolor{green!60!black}{\boldsymbol{z_\textsc{hl}}},\textcolor{red!60!black}{\boldsymbol{z_\textsc{hr}}}$}
[sibling distance=.5cm,-]
child {node[scale=.65] {.125} edge from parent node[left, scale=.75] {$x_\textsc{tl}$}}
child {node[scale=.65] (n4) {.125} edge from parent node[right, scale=.75] {$x_\textsc{tr}$}}}};
\end{tikzpicture}
\caption{Illustration of a plan-time history $y_1$ and its corresponding occupancy state $s_1$ in case of no public observations.}
\label{fig:occupancy:state}
\end{figure}

The probability distributions over hidden states and joint histories lay in the space $\mathcal{P}(\mathcal{X}\times \mathcal{O})$, which includes the space of reachable occupancy states $\mathcal{S}$. This convex space indicates that combining occupancy states by blending joint policies can produce a new occupancy state. These occupancy states describe a Markov game called the occupancy-state Markov game. In this game, the central planner simultaneously makes decisions for all agents and receives public observations as feedback, as shown in Figure \ref{fig:occupancy:markov:game:nf}.

\begin{definition}
\label{def:omg}
A $n$-agent, simultaneous-move, occupancy-state Markov game \wrt~ $M'$ (resp. $M$) is given by a tuple $M'' \doteq \langle \mathcal{S}, \mathcal{A}, \mathcal{W}, \rho_{M''}, \omega_{M''}, p_{M''}, r_{M''}\rangle$, where $\mathcal{S}$ is the occupancy-state space, occupancy states being conditional probability distributions over hidden states and joint histories; $\mathcal{A}$ is the space of actions describing joint decision rules; $\mathcal{W}$ is the space of public observations available to all agents involved in the game. Mappings $\rho_{M''}$, $\omega_{M''}$, $p_{M''}$, and $r_{M''}$ are described as follows:
\begin{itemize}
\item The deterministic transition rule $\rho_{M''}\colon \mathcal{S} \times \mathcal{A} \times \mathcal{W} \to \mathcal{S}$ prescribes the next occupancy state $s_{t+1} \doteq \rho_{M''}(s_t,a_t,w_{t+1})$ upon taking decision rule $a_t$ in occupancy state $s_t$ and receiving public observation $w_{t+1}$---\ie for any arbitrary hidden state $x_{t+1}$, joint history $o_t$, joint action $u_t$, and joint (private) observation $\tilde{z}_{t+1}$,
\begin{align*}
s_{t+1}(x_{t+1}, (o_t,u_t, \tilde{z}_{t+1},w_{t+1}))~\propto~ {\textstyle \sum_{x_t\in \mathcal{X}}}~ s_t(x_t,o_t) \cdot p_{x_t,x_{t+1}}^{u_t,\langle \tilde{z}_{t+1},w_{t+1}\rangle} \cdot a_t(u_t|o_t).
\end{align*}
\item The observation probability model $\omega_{M''}\colon \mathcal{S} \times \mathcal{A} \times \mathcal{W} \to [0,1]$ prescribes the probability of receiving public observation $w_{t+1}$ upon taking decision rule $a_t$ in occupancy state $s_t$---\ie
\begin{align*}
\omega_{M''}(w_{t+1}|s_t,a_t) &\doteq \sum_{x_t\in \mathcal{X}}\sum_{o_t\in \mathcal{O}_t}\sum_{u_t\in \mathcal{U}} \sum_{x_{t+1}\in \mathcal{X}} \sum_{ \tilde{z}_{t+1}\in \tilde{\mathcal{Z}}}
p_{x_t,x_{t+1}}^{u_t,\langle \tilde{z}_{t+1},w_{t+1}\rangle} \cdot a_t(u_t|o_t) \cdot s_t(x_t,o_t).
\end{align*}
\item The transition probability model $p_{M''} \colon \mathcal{S} \times \mathcal{A} \times \mathcal{S} \to [0,1]$ defines the probability of reaching next occupancy state $s_{t+1}$ upon taking decision rule $a_t$ in occupancy state $s_t$---\ie
\begin{align*}
p_{M''}(s_t,a_t,s_{t+1}) &\doteq {\textstyle \sum_{w\in \mathcal{W}}}~ \omega_{M''}(w_{t+1}|s_t,a_t)\cdot \delta_{\rho_{M''}(s_t,a_t,w_{t+1})}(s_{t+1}).
\end{align*}
\item The reward model $r_{M''}\colon \mathcal{S} \times \mathcal{A} \to \mathbb{R}^n$ describes the linear reward function upon taking decision rule $a_t$ in occupancy state $s_t$---\ie
\begin{align*}
r^i_{M''}(s_t,a_t) \doteq \textstyle\sum_{x_t\in \mathcal{X}}\sum_{o_t\in \mathcal{O}_t} s_t(x_t, o_t)\sum_{u_t\in \mathcal{U}} a_t(u_t|o_t) \cdot r^i_{x_t,u_t}.
\end{align*}
\end{itemize}
\end{definition}

\begin{figure}[!ht]
\centering
\begin{tikzpicture}[->,-latex,auto,node distance=3.75cm,semithick, square/.style={regular polygon,regular polygon sides=4}]
\tikzstyle{every state}=[draw=black,text=black,inner color= white,outer color= white,draw= black,text=black, drop shadow]
\tikzstyle{place}=[thick,draw=sthlmBlue,fill=blue!20,minimum size=12mm, opacity=.5]
\tikzstyle{red place}=[square,place,draw=sthlmRed,fill=sthlmLightRed,minimum size=18mm]
\tikzstyle{green place}=[diamond,place,draw=sthlmGreen,fill=sthlmLightGreen,minimum size=15mm]
\node[fill=white, scale=.75] (T) at (-3.5,-3.75) {};
\node[fill=white, scale=.75] (T0) at (0,-3.75) {$0$};
\node[fill=white, scale=.75] (T1) [right of=T0,node distance=3.75cm, fill=white] {$t$};
\node[fill=white, scale=.75] (T2) [right of=T1,node distance=3.75cm] {$t+1$};
\node[fill=white, scale=.75] (T3) [right of=T2,fill=white, node distance=3.75cm] {$t+2$};
\node[fill=white, scale=.75] (T4) [right of=T3,fill=white, node distance=2.5cm] {$\ldots$};
\draw[->,-latex,dashed,very thin, color=black,anchor=mid] (T3) -- (T4);
\draw[->,-latex, very thin, color=black, anchor=mid] (T2) -- (T3);
\draw[->,-latex, very thin, color=black, anchor=mid] (T1) -- (T2);
\draw[->,-latex,dashed,very thin, color=black,anchor=mid] (T0) -- (T1);
\draw[->,-latex,rounded corners,very thin, color=black,anchor=mid] (T) node[fill=white, scale=.75]{Time} -- (T0);
\node[state,place, scale=.75] (S0) {$s_0$};
\node[state,place, scale=.75] (S1) [right of=S0] {$s_t$};
\node[state,place, scale=.75] (S2) [ right of=S1] {$s_{t+1}$};
\node[state,place, scale=.75] (S3) [ right of=S2] {$s_{t+2}$};
\node[, scale=.75] (S4) [ right of=S3,node distance=2.5cm] {$\cdots$};
\path[very thin] (S0) edge[dashed] node[midway,text=black,draw=none,scale=.7, above=-5pt] {} (S1)
(S1) edge node[midway,text=black,draw=none,scale=.7, above=-5pt] {} (S2)
(S2) edge node[midway,text=black,draw=none,scale=.7, above=-5pt] {} (S3)
(S3) edge[dashed] (S4);
\node[, scale=.75] (A_SLACK) [below of=S0,node distance=2.6cm] {};
\node[state,red place, scale=.65] (A0) [below right of=A_SLACK,node distance=2cm] {$a_{0}$};
\node[state,red place, scale=.65] (A1) [right of=A0,node distance=4.25cm] {$a_t$};
\node[state,red place, scale=.65] (A2) [right of=A1,node distance=4.25cm] {$a_{t+1}$};
\node[, scale=.75] (A3) [right of=A2,node distance=5.5cm] {};
\path[very thin] (A0) edge [out=90, in=-155, dashed] node[scale=.75] {$w_t$} (S1)
(A1) edge [out=90, in=-155] node[scale=.75] {$w_{t+1}$} (S2)
(A2) edge [out=90, in=-155] node[scale=.75] {$w_{t+2}$ } (S3);
\node[state,green place, scale=.75] (O0) [below of=S0,node distance=2cm] {$r^{1:n}_0$};
\node[state,green place, scale=.75] (O1) [below of=S1,node distance=2cm] {$r^{1:n}_t$};
\node[state,green place, scale=.75] (O2) [below of=S2,node distance= 2cm] {$r^{1:n}_{t+1}$};
\node[state,green place, scale=.75] (O3) [below of=S3,node distance= 2cm] {$r^{1:n}_{t+2}$};
\node[, scale=.75] (O4) [ right of=O3,node distance=2.5cm] {$\cdots$};
\node[inner sep=0pt] (prof) at (-3.35,-1.5) {\includegraphics[width=.1\textwidth]{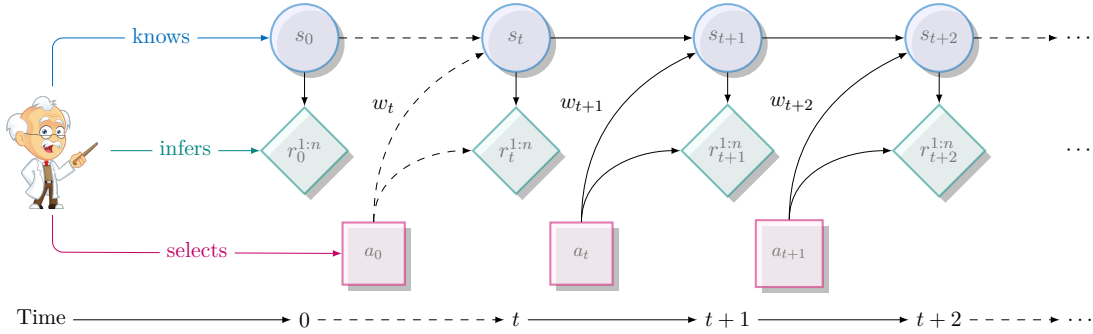}};
\draw[->,-latex,rounded corners,very thin, color=sthlmBlue,anchor=mid] (prof) -- (-3.35,0) --node[draw=none, fill=white, scale=.75] {knows} (S0);
\draw[->,-latex,rounded corners,very thin, color=sthlmGreen,anchor=mid] (prof) --node[draw=none, fill=white, scale=.75] {infers} (O0);
\draw[->,-latex,rounded corners,very thin, color=sthlmRed,anchor=mid] (prof) -- (-3.35,-2.85) --node[draw=none, fill=white, scale=.75] {selects} (A0);
\path[very thin] (S0) edge node {} (O0)
(S1) edge node {} (O1)
(A0) edge[out=90, in=-180, dashed] node {} (O1)
(S2) edge node {} (O2)
(A1) edge[out=90, in=-180] node {} (O2)
(S3) edge node {} (O3)
(A2) edge[out=90, in=-180] node {} (O3);
\end{tikzpicture}
\caption{Influence diagram for an occupancy-state Markov game.}
\label{fig:occupancy:markov:game:nf}
\end{figure}

The occupancy-state Markov game falls under the category of Markov games. Since the public observations have been made explicit, the reformulation leads to a stochastic rather than a deterministic process, as shown in \citet{continuous}. The aim of solving $M''$ is to find a joint policy similar to the ones used in both $M$ and $M'$. However, in this game, occupancy states replace plan-time histories, which can lead to overlooking crucial information, such as the joint policy. Doing so may not be the best approach for optimal decision-making. The substitution of occupancy states to plan-time histories requires guarantees that the former are sufficient statistics of the latter concerning three aspects: (i) predicting the current expected reward, (ii) predicting the next public observations, and (iii) estimating the next occupancy state. With these properties in hand, one can formally define a Markov game. Formal proofs of these properties follow.

\begin{restatable}{lem}{lemrewardpredictionsoc}
\label{lem:reward:prediction:soc}
The $t$-step occupancy state, defined as $s_t \doteq \Pr\{X_t,O_t|y_t\}$, is a sufficient statistic of the plan-time history $y_t$ for an accurate prediction of the expected reward that will be received immediately after executing the decision rule $a_t$, \ie $r_{M'}(y_t,a_t) = r_{M''}(s_t,a_t)$.
\end{restatable}

\begin{proof}
The proof starts with the definition of the expected reward conditional on plan-time history $y_t$ and decision rule $a_t$, \ie
\begin{align*}
r_{M'}(y_t,a_t) &\doteq \mathbb{E}\{ r_{X_t,U_t} | y_t, a_t \}.
\end{align*}
The expansion over hidden states, joint histories and joint actions leads to the following expression:
\begin{align*}
r_{M'}(y_t,a_t) &= \textstyle\sum_{x_t\in \mathcal{X}}\sum_{o_t\in \mathcal{O}_t}\sum_{u_t\in \mathcal{U}} \Pr\{x_t, o_t, u_t| y_t, a_t \} \cdot r_{x_t,u_t}.
\end{align*}
The decomposition of the joint probability as the product of conditional probabilities gives the following expression:
\begin{align*}
r_{M'}(y_t,a_t) &= \textstyle\sum_{x_t\in \mathcal{X}}\sum_{o_t\in \mathcal{O}_t}\sum_{u_t\in \mathcal{U}} \Pr\{x_t| y_t, a_t, o_t, u_t \}\Pr\{u_t| y_t, a_t, o_t \} \cdot r_{x_t,u_t}\cdot\Pr\{o_t| y_t \}.
\end{align*}
By removing terms that do not influence conditional probabilities, we obtain the following expression:
\begin{align*}
r_{M'}(y_t,a_t) &= \textstyle\sum_{x_t\in \mathcal{X}}\sum_{o_t\in \mathcal{O}_t}\sum_{u_t\in \mathcal{U}} \Pr\{x_t| y_t, o_t \} \Pr\{u_t| a_t, o_t \} \cdot r_{x_t,u_t}\cdot \Pr\{o_t| y_t \}.
\end{align*}
Finally re-arranging terms, noticing that $s_t(x_t, o_t) \doteq \Pr\{x,o|y_t\} =\Pr\{x_t| y_t, o_t \}\cdot \Pr\{o_t| y_t \}$ and $a_t(u_t|o_t) \doteq \Pr\{u_t| a_t, o_t \}$, we have that
\begin{align*}
r_{M'}(y_t,a_t) &= \textstyle\sum_{x_t\in \mathcal{X}}\sum_{o_t\in \mathcal{O}_t}\sum_{u_t\in \mathcal{U}} a_t(u_t|o_t) \cdot r_{x_t,u_t} \cdot s_t(x_t, o_t).
\end{align*}
By the definition of $r_{M''}(s_t,a_t)$, from Definition \ref{def:omg}, the following holds:
\begin{align*}
r_{M'}(y_t,a_t) &=r_{M''}(s_t,a_t).
\end{align*}
Which ends the proof.
\end{proof}

\begin{restatable}{lem}{lemobservationpredictionsoc}
\label{lem:observation:prediction:soc}
The $t$-step occupancy state, defined as $s_t \doteq \Pr\{X_t,O_t|y_t\}$, is a sufficient statistic of the plan-time history $y_t$ for predicting of the next public observation $w_{t+1}$ following the executing of the decision rule $a_t$, \ie $\omega_{M'}(w_{t+1}|y_t,a_t) = \omega_{M''}(w_{t+1}|s_t,a_t)$.
\end{restatable}

\begin{proof}
The proof starts with the definition of the probability of receiving public observation $w_{t+1}$ upon taking decision rule $a_t$ in plan-time history $y_t$, \ie
\begin{align*}
\omega_{M'}(w_{t+1} | y_t, a_t) &\doteq \Pr\{w_{t+1} | y_t, a_t \}.
\end{align*}
The expansion of the conditional probability $\Pr\{w_{t+1} | y_t, a_t \}$ over hidden states $x_t, x_{t+1}$, joint histories $o_t$, and joint actions $u_t$ and joint private observations $\tilde{z}_{t+1}$, result in the following expression:
\begin{align*}
\omega_{M'}(w_{t+1} | y_t, a_t) &= \sum_{x_t\in \mathcal{X}}\sum_{o_t\in \mathcal{O}_t}\sum_{u_t\in \mathcal{U}} \sum_{ x_{t+1}\in \mathcal{X}} \sum_{ \tilde{z}_{t+1}\in \tilde{\mathcal{Z}}} \Pr\{x_t, o_t, u_t, x_{t+1}, \tilde{z}_{t+1}, w_{t+1} | y_t, a_t \}.
\end{align*}
The decomposition of the joint probability as a product of conditional probabilities yields the following result:
\begin{align*}
\omega_{M'}(w_{t+1} | y_t, a_t) &= \sum_{x_t\in \mathcal{X}}\sum_{o_t\in \mathcal{O}_t}\sum_{u_t\in \mathcal{U}} \sum_{ x_{t+1}\in \mathcal{X}} \sum_{ \tilde{z}_{t+1}\in \tilde{\mathcal{Z}}}
\Pr\{x_{t+1}, \tilde{z}_{t+1}, w_{t+1} | y_t, a_t, x_t, o_t, u_t \} \\
& \quad \cdot \Pr\{x_t | y_t, a_t, u_t, o_t \} \cdot \Pr\{u_t | y_t, a_t, o_t \} \cdot
\Pr\{o_t | y_t, a_t \}.
\end{align*}
The first factor is the game's dynamics; the second factor is the probability of the hidden states given plan-time history and joint history; the third factor is the probability of joint actions given decision rule and joint history, and the last factor is the probability of joint histories given the plan-time history, \ie
\begin{align*}
\omega_{M'}(w_{t+1} | y_t, a_t) &= \sum_{x_t\in \mathcal{X}}\sum_{o_t\in \mathcal{O}_t}\sum_{u_t\in \mathcal{U}} \sum_{ x_{t+1}\in \mathcal{X}} \sum_{ \tilde{z}_{t+1}\in \tilde{\mathcal{Z}}}
\Pr\{x_{t+1}, \tilde{z}_{t+1}, w_{t+1} | x_t, u_t \} \\
& \quad \cdot \Pr\{x_t | y_t, o_t \} \cdot \Pr\{u_t | a_t, o_t \} \cdot
\Pr\{o_t | y_t \}.
\end{align*}
Re-arranging and re-naming terms yield the following expression
\begin{align*}
\omega_{M'}(w_{t+1} | y_t, a_t) &= \sum_{x_t\in \mathcal{X}}\sum_{o_t\in \mathcal{O}_t}\sum_{u_t\in \mathcal{U}} \sum_{ x_{t+1}\in \mathcal{X}} \sum_{ \tilde{z}_{t+1}\in \tilde{\mathcal{Z}}}
p_{x_t,x_{t+1}}^{u_t,\langle \tilde{z}_{t+1}, w_{t+1}\rangle} \cdot a_t(u_t | o_t ) \cdot s_t(x_t ,o_t ).
\end{align*}
Consequently, the following equality holds by Definition \ref{def:omg} for $\omega_{M''}(w_{t+1} |s_t, a_t)$:
\begin{align*}
\omega_{M'}(w_{t+1} | y_t, a_t) &= \omega_{M''}(w_{t+1} |s_t, a_t).
\end{align*}
Which ends the proof.
\end{proof}

\begin{restatable}{lem}{lemsocpredictionsoc}
\label{lem:soc:prediction:soc}
The $t$-step occupancy state, defined as $s_t \doteq \Pr\{X_t,O_t|y_t\}$, is a sufficient statistic of the plan-time history $y_t$ for predicting the subsequent occupancy state following the executing of the decision rule $a_t$ and the reception of public observation $w_{t+1}$, \ie $\rho_{M'}(y_t,a_t,w_{t+1}) = \rho_{M''}(s_t,a_t,w_{t+1})$.
\end{restatable}

\begin{proof}
In demonstrating the sufficiency of the occupancy states, we also derive the rule for updating occupancy states from one time step to the next. For the process defined by $M'$, the only information that the central planner obtains during the plan time $t$ is the fact that a particular decision rule $a_t$ has been followed, and public observation $w_{t+1}$ has been received at the end of the plan time. If $(y_t, a_t,w_{t+1})$ denotes the total data available to the central planner at plan time $t+1$, then we can write:
\begin{align*}
y_{t+1} &\doteq (y_t,a_t,w_{t+1}),
\end{align*}
$y_{t+1}$ represents the plan-time history before time step $t+1$ plus the additional information that a particular decision rule was recorded and public observation occurred. By Definition \ref{def:occupancy:state}, occupancy states are written as follows: for every hidden state $x_{t+1}$, joint history $o_t$, joint action $u_t$, and joint observation $\langle \tilde{z}_{t+1}, w_{t+1}\rangle$,
\begin{align}
s_{t+1}(x_{t+1}, (o_t,u_t,\langle \tilde{z}_{t+1}, w_{t+1}\rangle)) &\doteq \Pr\{ x_{t+1}, (o_t,u_t,\langle \tilde{z}_{t+1}, w_{t+1}\rangle) | y_{t+1}\}.
\label{eq:lem:soc:prediction:soc:1}
\end{align}
The application of Bayes' rule and the substitution of $y_{t+1}$ by $(y_t,a_t,w_{t+1})$ in \eqref{eq:lem:soc:prediction:soc:1} yield
\begin{align}
s_{t+1}(x_{t+1}, (o_t,u_t,\langle \tilde{z}_{t+1}, w_{t+1}\rangle)) &= \frac{ \Pr\{ x_{t+1}, o_t,u_t, \tilde{z}_{t+1}, w_{t+1}|y_t,a_t\}}{\Pr\{w_{t+1}|y_t,a_t\}}.
\label{eq:lem:soc:prediction:soc:2}
\end{align}
From Lemma \ref{lem:observation:prediction:soc}, we know that the denominator of \eqref{eq:lem:soc:prediction:soc:2} is given $\omega_{M'}(w_{t+1} | y_t, a_t) = \omega_{M''}(w_{t+1} |s_t, a_t)$. The expansion of the joint probability in the numerator of \eqref{eq:lem:soc:prediction:soc:2} as a product of conditional probabilities yields
\begin{align*}
s_{t+1}(x_{t+1}, (o_t,u_t,\langle \tilde{z}_{t+1}, w_{t+1}\rangle)) &~\propto \sum_{x_t\in \mathcal{X}} \Pr\{ x_t, x_{t+1}, o_t,u_t, \tilde{z}_{t+1}, w_{t+1}|y_t,a_t\}.
\end{align*}
The expansion of the joint probability as a product of conditional probabilities gives us the following expression:
\begin{align*}
s_{t+1}(x_{t+1}, (o_t,u_t,\langle \tilde{z}_{t+1}, w_{t+1}\rangle)) &~\propto \sum_{x_t\in \mathcal{X}} \Pr\{ x_{t+1}, \tilde{z}_{t+1}, w_{t+1}|x_t, o_t,u_t, y_t,a_t\} \nonumber\\
&\quad \cdot \Pr\{x_t|y_t,a_t, o_t,u_t \} \cdot
\Pr\{u_t|y_t,a_t, o_t \} \cdot
\Pr\{o_t|y_t,a_t \}
\end{align*}
Re-arranging terms
\begin{align*}
s_{t+1}(x_{t+1}, (o_t,u_t,\langle \tilde{z}_{t+1}, w_{t+1}\rangle)) &~\propto \sum_{x_t\in \mathcal{X}} p_{x_t,x_{t+1}}^{u_t, \tilde{z}_{t+1}, w_{t+1}} \cdot a_t(u_t| o_t) \cdot s_t(x_t,o_t).
\end{align*}
Consequently, the following equality holds by Definition \ref{def:omg} for $\rho_{M''}(s_t, a_t, w_{t+1})$:
\begin{align*}
s_{t+1} &= \rho_{M''}(s_t, a_t, w_{t+1}).
\end{align*}
Which ends the proof.
\end{proof}

We introduce state-value functions for a fixed joint policy under $M''$ to determine if occupancy states are sufficient. These state-value functions compute the expected cumulative $\gamma$-discounted reward from the initial belief state $s_0$ onward. They are represented as $\upsilon_{M'',\gamma,0}^{a_{0:}}(s_0) \doteq \mathbb{E}\{ \textstyle{\sum_{t=0}^{\ell-1}}~\gamma^t \cdot r_{M''}(S_t,a_t) | s_0, a_{0:} \}$.

\begin{definition}
\label{def:state:value:fct:m:os}
The $t$-step state-value function $\upsilon^{i,a_{t:}}_{M'',\gamma,t}\colon \mathcal{S}_t \to \mathbb{R}$ under the joint policy $a_{t:}$ is given by: for any occupancy state $s_t$,
\begin{align*}
\upsilon^{i,a_{t:}}_{M'',\gamma,t}(s_t) &\doteq\textstyle \mathbb{E} \left\{\sum_{\tau=t}^{\ell-1} \gamma^{\tau-t} \cdot r^i_{M''}(S_\tau,a_\tau) | s_t, a_{t:} \right\},
\end{align*}
with boundary condition $\upsilon^{i,\cdot}_{M,\gamma,\ell}(\cdot) \doteq 0$.
\end{definition}

According to the following lemma, the state-value functions satisfy \citeauthor{bellman}'s equations for occupancy states at every time step.

\begin{restatable}{lem}{lembellmanstatevaluefctmos}
The $t$-step state-value function $\upsilon^{i,a_{t:}}_{M'',\gamma,t}\colon \mathcal{S}_t\to \mathbb{R}$ under the joint policy $a_{0:}$ satisfies \citeauthor{bellman}'s equations: for any occupancy state $s_t$,
\begin{align*}
\upsilon^{i,a_{t:}}_{M'',\gamma,t}(s_t) &= q^{i,a_{t+1:}}_{M'',\gamma,t}(s_t,a_t),\\
q^{i,a_{t+1:}}_{M'',\gamma,t}(s_t,a_t) &\doteq \textstyle r^i_{M''}(s_t,a_t) + \gamma \sum_{s_{t+1}\in \mathcal{S}_{t+1}} p_{M''}(s_t,a_t,s_{t+1})\cdot \upsilon^{i,a_{t+1:}}_{M'',\gamma,t+1}(s_{t+1}),
\end{align*}
with boundary condition $\upsilon^{i,\cdot}_{M'',\gamma,\ell}(\cdot) \doteq 0$.
\end{restatable}

\begin{proof}
The proof follows directly from the definition of the state-value functions under joint policy $a_{0:}$ over the occupancy states: for time step $t$, occupancy state $s_t$,
\begin{align*}
\upsilon^{a_{t:}}_{M'',\gamma,t}(s_t) &\doteq\textstyle \mathbb{E}\left\{ \sum_{\tau=t}^{\ell-1} \gamma^{\tau-t} \cdot r_{M''}(S_\tau,a_\tau) | s_t, a_{t:} \right\},
\end{align*}
Exploiting the fact that $r_{M''}(s_t,a_t)$ is not a random variable, we obtain the following expression:
\begin{align*}
\upsilon^{a_{t:}}_{M'',\gamma,t}(s_t)
&=\textstyle r_{M''}(s_t,a_t) + \gamma \mathbb{E}\left\{ \sum_{\tau=t+1}^{\ell-1} \gamma^{\tau-t-1} \cdot r_{M''}(S_\tau,a_\tau) | s_t, a_{t:} \right\}.
\end{align*}
The expansion of the expectation over all next occupancy states (through public observations) gives the following expression:
\begin{align*}
\upsilon^{a_{t:}}_{M'',\gamma,t}(s_t)
&=\textstyle r_{M''}(s_t,a_t) + \gamma \mathbb{E}\left\{\mathbb{E}\{ \sum_{\tau=t+1}^{\ell-1} \gamma^{\tau-t-1} \cdot r_{M''}(S_\tau,a_\tau) | s_{t+1}, a_{t+1:} \}|s_t, a_t\right\}.
\end{align*}
The injection of Definition \ref{def:state:value:fct:m:os} of $\upsilon^{a_{t+1:}}_{M'',\gamma,t+1}(s_{t+1})$ instead of the inner expectation gives the following expression:
\begin{align*}
\upsilon^{a_{t:}}_{M'',\gamma,t}(s_t)
&=\textstyle r_{M''}(s_t,a_t) + \gamma \mathbb{E}_{s_{t+1}\sim p_{M''}(s_t,a_t,\cdot)}\left\{\upsilon^{a_{t+1:}}_{M'',\gamma,t+1}(s_{t+1})|s_t, a_t\right\}.
\end{align*}
By definition of the expectation, we know that the expression becomes:
\begin{align*}
&\overset{\text{Definition \ref{def:state:value:fct:m:os}}}{=} \textstyle r_{M''}(s_t,a_t) + \gamma \sum_{s_{t+1}\in \mathcal{S}_{t+1}}p_{M''}(s_t,a_t,s_{t+1}) \cdot \upsilon^{a_{t+1:}}_{M'',\gamma,t+1}(s_{t+1})\\
&\doteq q^{a_{t+1:}}_{M'',\gamma,t}(s_t,a_t).
\end{align*}
Which ends the proof.
\end{proof}

It is important to recognize that state- and action-value functions in the master game $M$ under a fixed joint policy, which have a tabular nature, inherently display uniform continuity characteristics. The following explanation highlights the simple yet significant linear properties associated with these functions, providing a fixed joint policy. It is worth noting that these uniform continuity properties apply to occupancy states.

\begin{restatable}{lem}{lemlinearityjointpolicy}
\label{lem:linearity:joint:policy}
The $t$-step state-value function $\upsilon^{i,a_{t:}}_{M'',\gamma,t}\colon \mathcal{S}_t \to \mathbb{R}$ under joint policy $a_{t:}$ is linear over occupancy states: for any occupancy state $s_t$,
\begin{align*}
\upsilon^{i,a_{t:}}_{M'',\gamma,t}(s_t) &= \textstyle\sum_{x_t\in \mathcal{X},o_t\in \mathcal{O}_t} s_t(x_t,o_t)\cdot \upsilon^{i,a_{t:}}_{M,\gamma,t}(x_t,o_t).
\end{align*}
Similarly, the $t$-step action-value function $q^{i,a_{t+1:}}_{M'',t}\colon \mathcal{S}_t\times \mathcal{A}_t \to \mathbb{R}$ under the joint policy $a_{t+1:}$ is linear over occupancy states and joint decision rules: for any occupancy state $s_t\in \mathcal{S}_t$ and joint decision rule $a_t$
\begin{align*}
q^{i,a_{t+1:}}_{M'',\gamma,t}(s_t,a_t) &= \textstyle\sum_{x_t\in \mathcal{X},o_t\in \mathcal{O}_t} s_t(x_t,o_t)\sum_{u_t\in \mathcal{U}}a_t(u_t|o_t)\cdot q^{i,a_{t+1:}}_{M,\gamma,t}(x_t,o_t,u_t).
\end{align*}
\end{restatable}

\begin{proof}
The proof starts with the definition of the state- and action-value functions for a fixed joint policy at time step $t$ and occupancy state $s_t$,
\begin{align*}
\upsilon^{a_{t:}}_{M'',\gamma,t}(s_t) &\doteq \textstyle \mathbb{E}\{\sum_{\tau=t}^{\ell-1} \gamma^{\tau-t}\cdot r_{M''}(S_\tau,a_\tau) | s_t, a_{t:}\}.
\end{align*}
It is worth noticing that the dynamics of occupancy states are stochastic, which makes it possible to rewrite $\upsilon^{a_{t:}}_{M'',\gamma,t}(s_t)$ as follows:
\begin{align*}
\upsilon^{a_{t:}}_{M'',\gamma,t}(s_t) &= \textstyle q^{a_{t+1:}}_{M'',\gamma,t}(s_t,a_t)\\
q^{a_{t+1:}}_{M'',\gamma,t}(s_t,a_t) &=\textstyle r_{M''}(s_t,a_t) + \gamma \sum_{s_{t+1}\in \mathcal{S}_{t+1}} p_{M''}(s_t,a_t,s_{t+1})\cdot \upsilon^{a_{t+1:}}_{M'',\gamma,t+1}(s_{t+1}).
\end{align*}
The property naturally holds when the horizon is exhausted, \ie $\upsilon^{\cdot}_{M'',\gamma,\ell}(\cdot) = q^{\cdot}_{M'',\gamma,\ell}(\cdot,\cdot) \doteq 0$. Suppose it also holds at time step $t+1$ onward, \ie
\begin{align*}
\upsilon^{a_{t+1:}}_{M'',\gamma,t+1}(s_{t+1}) &=\textstyle \sum_{x_t\in \mathcal{X}}\sum_{o_t\in \mathcal{O}_t} s_{t+1}(x_t,o_t)\cdot \upsilon^{a_{t+1:}}_{M,\gamma,t+1}(x_t,o_t)\\
q^{a_{t+2:}}_{M'',\gamma,t+1}(s_{t+1},a_{t+1}) &=\textstyle\sum_{x_t\in \mathcal{X}}\sum_{o_t\in \mathcal{O}_t} s_{t+1}(x_t,o_t)\sum_{u_t\in \mathcal{U}} a_{t+1}(u_t|o_t)\cdot q^{a_{t+2:}}_{M,\gamma,t+1}(x_t,o_t,u_t).
\end{align*}
Then, it follows that
\begin{align*}
\upsilon^{a_{t:}}_{M'',\gamma,t}(s_t) &= \textstyle q^{a_{t+1:}}_{M,\gamma,t}(s_t,a_t)\\
q^{a_{t+1:}}_{M'',\gamma,t}(s_t,a_t)&=\textstyle \sum_{x_t\in \mathcal{X}}\sum_{o_t\in \mathcal{O}_t}s_t(x_t,o_t)\sum_{u_t\in \mathcal{U}} a_t(u_t|o_t) \cdot q^{a_{t+1:}}_{M,\gamma,t}(x_t,o_t,u_t)\\
q^{a_{t+1:}}_{M,\gamma,t}(x_t,o_t,u_t)&=\textstyle r_{x_t,u_t} + \gamma \sum_{x_{t+1}\in \mathcal{X},z_{t+1}\in \mathcal{Z}}p_{x_tx_{t+1}}^{u_tz_{t+1}}\cdot \upsilon^{a_{t+1:}}_{M,\gamma,t+1}(x_{t+1},(o_t,u_t,z_{t+1})) .
\end{align*}
Re-arranging terms then yield the target results, \ie
\begin{align*}
\upsilon^{a_{t:}}_{M'',\gamma,t}(s_t) &= \textstyle \sum_{x_t\in \mathcal{X}}\sum_{o_t\in \mathcal{O}_t} s_t(x_t,o_t)\cdot \upsilon^{a_{t:}}_{M,\gamma,t}(x_t,o_t)\\
q^{a_{t+1:}}_{M'',\gamma,t}(s_t,a_t)&=\textstyle\sum_{x_t\in \mathcal{X}}\sum_{o_t\in \mathcal{O}_t} s_t(x_t,o_t)\sum_{u_t\in \mathcal{U}} a_t(u_t|o_t)\cdot q^{a_{t+1:}}_{M,\gamma,t}(x_t,o_t,u_t).
\end{align*}
Hence, the property holds for any arbitrary time step. Which ends the proof.
\end{proof}

Lemma \ref{lem:linearity:joint:policy} highlights that when a fixed joint policy is utilized, the state- and action-value functions can be elegantly expressed as linear functions of occupancy states across all types of master games that are represented as partially observable stochastic games. This linearity property is a fundamental aspect that contributes greatly to proving the convexity property of optimal value functions for specific master games, including the common-reward partially observable stochastic games family (\emph{dec}-POMDP).

The conditions for the sufficiency of occupancy states in optimally solving master games as occupancy-state Markov games are presented in the following proposition.

\begin{restatable}{pro}{prosufficiencyoccupancystate}
\label{pro:sufficiency:occupancy:state}
Let $M$ be a master game, $M'$ be the corresponding plan-time Markov game, and $M''$ the associated occupancy-state Markov game. If we let $a_{0:}$ be a solution of $M''$, then $a_{0:}$ is also a solution of $M$ and $M'$.
\end{restatable}
\begin{proof}
Proposition \ref{pro:sufficiency:ckh} demonstrates that a solution for $M'$ is also valid for $M$. The only thing left to prove is that if $a_{0:}$ is a solution for $M''$, then it is also a solution for $M'$. Initially, we need to show that the state-value function for any particular time step $t$ under a fixed joint policy is dependent on the plan-time history $y_t$ only through the corresponding occupancy state $s_t$, which is defined as $(\Pr\{x_t,o_t| y_t\})_{x_t\in \mathcal{X},o_t\in \mathcal{O}_t}$. For any arbitrary fixed joint policy $a_{t:}$,
\begin{align*}
\upsilon^{a_{t:}}_{M',\gamma,t}(y_t) &\doteq \textstyle \mathbb{E}\{\sum_{\tau=t}^{\ell-1} \gamma^{\tau-t} \cdot r_{M'}(Y_\tau,a_\tau)|y_t,a_{t:}\}.
\end{align*}
The application of Lemma \ref{lem:reward:prediction:soc}, Lemma \ref{lem:observation:prediction:soc}, and Lemma \ref{lem:soc:prediction:soc}, lead to the following expression:
\begin{align*}
\upsilon^{a_{t:}}_{M',\gamma,t}(y_t) &=\textstyle \mathbb{E}\{\sum_{\tau=t}^{\ell-1} \gamma^{\tau-t} \cdot r_{M''}(S_\tau,a_\tau)|s_t,a_{t:} \},\\
&\doteq \upsilon^{a_{t:}}_{M'',\gamma,t}(s_t).
\end{align*}
Which ends the proof.
\end{proof}

Please note that Propositions \ref{pro:sufficiency:ckh} and \ref{pro:sufficiency:occupancy:state} do not provide a practical solution to master games. Rather, they show that finding a solution to either reformulated version will result in a valid solution for the original master game. It is crucial to remember that the solution for these reformulated games does not depend on their states, as they are non-observable. Therefore, it is not advisable to use state-dependent solutions when solving these reformulations, as is often done in fully observable settings.

\begin{figure}[!ht]
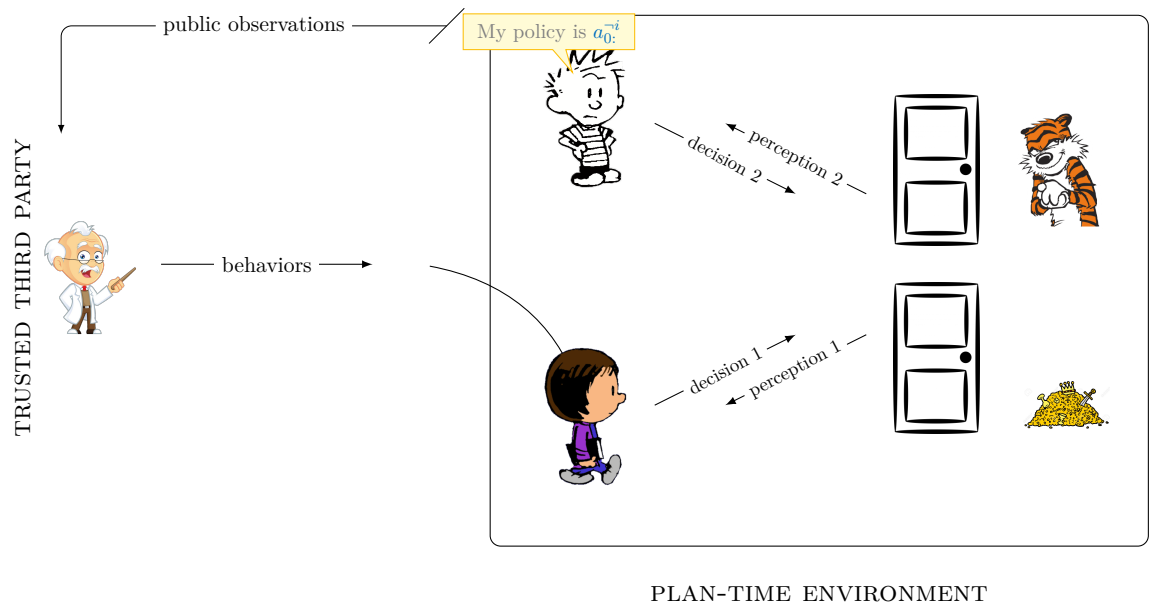

\centering
\begin{tikzpicture}[scale=1.25]
\draw[very thin, color=black] (-1.9,2.85) -- (-1.525, 3.225);
\draw[->, -latex,rounded corners,very thin, color=black] (-1.7125,3.0375) --node[scale=.75, fill=white] {public observations} (-5.8125, 3.0375) -- (-5.8125, 1.8875);
\draw[rounded corners, fill=white] (-1.25,-2.5) rectangle +(7,5.65);
\node (agentCentral) at (-2,.5) {};
\node (agentA) at (-.25,2.25) {};
\node (agentD) at (-.25,-1) {};
\path[->, -latex] (agentCentral)
edge [bend left] (agentD);
\begin{scope}[scale=.5]
\def \agentCentral{(-11,1) ellipse (1cm and 0.9cm);
\node[scale=1, rotate=90] at (-12.45,0.5) {\textcolor{black}{\sc trusted third party}};
\node[inner sep=0pt] (russell) at (-11,.75) {\includegraphics[width=.1\textwidth]{figures/TrustedThirdParty.png}};}
\def \agentA{(-.5,4.5) ellipse (1cm and 0.9cm);
\node[inner sep=0pt] (calvin) at (-.5,4.2) {\includegraphics[width=.075\textwidth]{figures/Calvin_BW.png}};}
\def \agentD{(-.5,-2) ellipse (1cm and 0.9cm);
\node[inner sep=0pt] (susie) at (-.5,-2.25) {\includegraphics[width=.075\textwidth]{figures/Susie.png}};}
\draw[draw=white] \agentCentral;
\node (prof) at (-11,2.3) {};
\draw[draw=white] \agentA;
\draw[draw=white] \agentD;
\node[draw=sthlmYellow, fill=yellow!20, rectangle callout, text centered, text width=3cm, text=gray, scale=.7] at (-1.25,5.95) {\textcolor{gray}{My policy is \textcolor{sthlmBlue}{$a^{\neg i}_{0:}$}}};
\node[inner sep=0pt] (hobbes) at (9.6,3) {\includegraphics[width=.075\textwidth]{figures/Hobbes.png}};
\node[inner sep=0pt] (door1) at (7,3) {\includegraphics[width=.075\textwidth]{figures/Door.png}};
\node[inner sep=0pt] (door2) at (7,-1) {\includegraphics[width=.075\textwidth]{figures/Door.png}};
\node[inner sep=0pt] (treasure) at (9.75,-2) {\includegraphics[width=.075\textwidth]{figures/Treasure.jpg}};
\end{scope}
\draw[->,-latex, color=black] (-4.75,.5) -- node[draw=none, fill=white, scale=.75] {behaviors} (-2.5,.5);
\node[scale=1] at (2.25,-3) {\sc plan-time environment};
\draw[->,-latex] (.5,-1) -- node[draw=none, fill=white,scale=.65,rotate=28] {decision 1} (2,-.25);
\draw[->,-latex] (2.75,-.25) -- node[draw=none, fill=white,scale=.65,rotate=28] {perception 1} (1.25,-1);
\draw[->,-latex] (.5,2) -- node[draw=none, fill=white,scale=.65,rotate=-28] {decision 2} (2,1.25);
\draw[->,-latex] (2.75,1.25) -- node[draw=none, fill=white,scale=.65,rotate=-28] {perception 2} (1.25,2);
\end{tikzpicture}
\caption{A plan-time environment for the tiger problem from the perspective of the trusted third party, now reasoning only on behalf of Suzie given that Calvin's policy is plan-time. He knows Suzie's decisions and perceptions, but when it comes to Calvin, he has only access to public observations, including Calvin's policy and the plan-time environment.}
\label{fig:multi:single:agent:environment:private}
\end{figure}

\subsection{Slave Game Reformulations}
This subsection presents a new perspective on the slave games, viewed from the point of view of a trusted third party. As depicted in Figure \ref{fig:multi:single:agent:environment:private}, the slave games were originally formulated as a single-agent decision-making problem, assuming that other agents had a fixed joint policy. However, this approach led to weak uniform continuity properties, as discussed in Theorems \ref{thm:convex:history} and \ref{thm:convex:lipschitz:history}, due to the use of different viewpoints to describe master and slave games, making it difficult to blend their value functions. This subsection aims to provide a unified perspective of both master and slave games from the point of view of a trusted third party, making it easier to identify properties of the slave games that can be applied to master games.

The central planner can think for one agent when considering a trusted third-party viewpoint. The other agents have a predetermined collective strategy, and this information is known to the central planner and the agent for whom it makes decisions. The central planner then gives that agent a set of rules to follow at every time step. In this scenario, the central planner can receive information about the action taken and the observations made by the specific agent it decides for. The private plan-time environment, viewed from the trusted third party's perspective, includes the original game environment and all interacting agents. However, unlike the master game's plan-time environment, the other agents have a predetermined joint policy, and the central planner can only see the private history of the agent they are reasoning for, not the private histories of the other agents. This reformulation turns the slave games into the private plan-time Markov game. Still, it does not affect how the single agent, the central planner, is reasoning to act in the original environment. Additionally, we will prove that both games, although not equivalent, can lead to equivalent solutions under mild conditions. The central planner collects data at the end of each time step, known as the private plan-time history. This includes the initial belief state, the private history of the specific agent taken into account, and the collective policy of all other agents.

\begin{definition}
For an agent $i$, a private plan-time history at time step $t$ includes the initial belief state, its individual history and the joint policy of all other agents except agent $i$, up to time step $t-1$, \ie $y_t^i \doteq \langle s_0,o^i_t, a^{\neg i}_{:t-1}\rangle$, with boundary condition $y_0^i \doteq s_0$.
\end{definition}

Private plan-time histories describe a special process called private plan-time Markov game that is illustrated in Figure \ref{fig:private:ck:game:nf}. A formal description of this game follows.

\begin{definition}
\label{def:private:plantime:slave:game}
Consider the master game $M$ and its surrogate $M'$. Let $a^{\neg i}_{0:}$ be a joint policy of all agents except agent $i$. Private plan-time Markov game $M'(a^{\neg i}_{0:})$ is characterized by the tuple $\langle \mathcal{Y}^i, \mathcal{U}^i, \mathcal{Z}^i, \rho^i_{M'}, \omega^i_{M'}, r^i_{M'}\rangle$. Here, $\mathcal{Y}^i$ represents a space of private plan-time histories, $\mathcal{U}^i$ is the finite action space, and $\mathcal{Z}^i$ is the finite observation space. Mappings $\rho^i_{M'}$, $\omega^i_{M'}$, and $r^i_{M'}$ are formally defined as follows:
\begin{itemize}
\item Mapping $\rho^i_{M'}\colon \mathcal{Y}^i\times \mathcal{U}^i\times \mathcal{Z}^i \to \mathcal{Y}^i$ describes a transition rule over private plan-time histories, \ie for any private plan-time history $y^i_t$, action $u_t^i$, and observation $z^i_{t+1}$, the next private plan-time history is given by: $\rho^i_{M'}(y_t^i, u_t^i, z_{t+1}^i)\doteq (y_t^i,a^{\neg i}_t, u^i_t,z^i_{t+1})$.
\item Mapping $\omega^i_{M'}\colon \mathcal{Y}^i\times \mathcal{U}^i\times \mathcal{Z}^i \to [0,1]$ is the observation rule, \ie for any private plan-time history $y^i_t$, action $u_t^i$, and observation $z^i_{t+1}$, the observation probability is given by: $\omega^i_{M'}(z_{t+1}^i|y_t^i, u_t^i) \doteq \Pr\{z^i_{t+1}|y_t^i, u_t^i,a^{\neg i}_t\}$.
\item Mapping $r^i_{M'}\colon \mathcal{Y}^i\times \mathcal{U}^i \to \mathbb{R}$ defines the reward model, \ie for any private plan-time history $y^i_t$, action $u_t^i$, the reward is given by: $r^i_{M'}( y_t^i, u_t^i)\doteq \mathbb{E}\{r^i_{X_t,U_t}|y_t^i, u_t^i,a^{\neg i}_t\}$.
\end{itemize}
\end{definition}

\begin{figure}[!ht]
\centering
\begin{tikzpicture}[->,-latex,auto,node distance=3.75cm,semithick, square/.style={regular polygon,regular polygon sides=4}]
\tikzstyle{every state}=[draw=black,text=black,inner color= white,outer color= white,draw= black,text=black, drop shadow]
\tikzstyle{place}=[thick,draw=sthlmBlue,fill=blue!20,minimum size=12mm, opacity=.5]
\tikzstyle{red place}=[square,place,draw=sthlmRed,fill=sthlmLightRed,minimum size=18mm]
\tikzstyle{green place}=[diamond,place,draw=sthlmGreen,fill=sthlmLightGreen,minimum size=15mm]
\node[fill=white, scale=.75] (T) at (-3.5,-3.75) {};
\node[fill=white, scale=.75] (T0) at (0,-3.75) {$0$};
\node[fill=white, scale=.75] (T1) [right of=T0,node distance=3.75cm, fill=white] {$t$};
\node[fill=white, scale=.75] (T2) [right of=T1,node distance=3.75cm] {$t+1$};
\node[fill=white, scale=.75] (T3) [right of=T2,fill=white, node distance=3.75cm] {$t+2$};
\node[fill=white, scale=.75] (T4) [right of=T3,fill=white, node distance=2.5cm] {$\ldots$};
\draw[->,-latex,dashed,very thin, color=black,anchor=mid] (T3) -- (T4);
\draw[->,-latex, very thin, color=black, anchor=mid] (T2) -- (T3);
\draw[->,-latex, very thin, color=black, anchor=mid] (T1) -- (T2);
\draw[->,-latex,dashed,very thin, color=black,anchor=mid] (T0) -- (T1);
\draw[->,-latex,rounded corners,very thin, color=black,anchor=mid] (T) node[fill=white, scale=.75]{Time} -- (T0);
\node[state,place, scale=.75] (S0) {$y^i_0$};
\node[state,place, scale=.75] (S1) [right of=S0] {$y^i_t$};
\node[state,place, scale=.75] (S2) [ right of=S1] {$y^i_{t+1}$};
\node[state,place, scale=.75] (S3) [ right of=S2] {$y^i_{t+2}$};
\node[, scale=.75] (S4) [ right of=S3,node distance=2.5cm] {$\cdots$};
\path[very thin] (S0) edge[dashed] node[midway,text=black,draw=none,scale=.7, above=-5pt] {} (S1)
(S1) edge node[midway,text=black,draw=none,scale=.7, above=-5pt] {} (S2)
(S2) edge node[midway,text=black,draw=none,scale=.7, above=-5pt] {} (S3)
(S3) edge[dashed] (S4);
\node[, scale=.75] (A_SLACK) [below of=S0,node distance=2.6cm] {};
\node[state,red place, scale=.65] (A0) [below right of=A_SLACK,node distance=2cm] {$u^i_{0}$};
\node[state,red place, scale=.65] (A1) [right of=A0,node distance=4.25cm] {$u^i_t$};
\node[state,red place, scale=.65] (A2) [right of=A1,node distance=4.25cm] {$u^i_{t+1}$};
\node[, scale=.75] (A3) [right of=A2,node distance=5.5cm] {};
\path[very thin] (A0) edge [out=90, in=-155, dashed] node[scale=.75] {$z^i_t$} (S1)
(A1) edge [out=90, in=-155] node[scale=.75] {$z^i_{t+1}$} (S2)
(A2) edge [out=90, in=-155] node[scale=.75] {$z^i_{t+2}$} (S3);
\node[state,green place, scale=.75] (O0) [below of=S0,node distance=2cm] {$r^i_0$};
\node[state,green place, scale=.75] (O1) [below of=S1,node distance=2cm] {$r^i_t$};
\node[state,green place, scale=.75] (O2) [below of=S2,node distance= 2cm] {$r^i_{t+1}$};
\node[state,green place, scale=.75] (O3) [below of=S3,node distance= 2cm] {$r^i_{t+2}$};
\node[, scale=.75] (O4) [ right of=O3,node distance=2.5cm] {$\cdots$};
\node[inner sep=0pt] (prof) at (-3.35,-1.5) {\includegraphics[width=.1\textwidth]{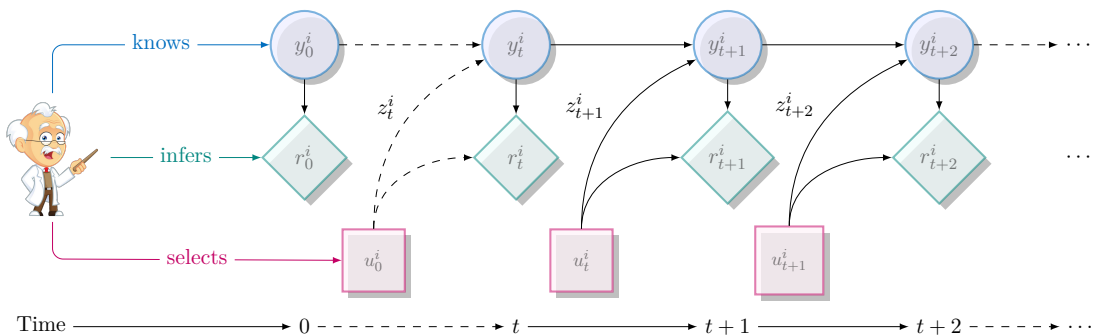}};
\draw[->,-latex,rounded corners,very thin, color=sthlmBlue,anchor=mid] (prof) -- (-3.35,0) --node[draw=none, fill=white, scale=.75] {knows} (S0);
\draw[->,-latex,rounded corners,very thin, color=sthlmGreen,anchor=mid] (prof) --node[draw=none, fill=white, scale=.75] {infers} (O0);
\draw[->,-latex,rounded corners,very thin, color=sthlmRed,anchor=mid] (prof) -- (-3.35,-2.85) --node[draw=none, fill=white, scale=.75] {selects} (A0);
\path[very thin] (S0) edge node {} (O0)
(S1) edge node {} (O1)
(A0) edge[out=90, in=-180, dashed] node {} (O1)
(S2) edge node {} (O2)
(A1) edge[out=90, in=-180] node {} (O2)
(S3) edge node {} (O3)
(A2) edge[out=90, in=-180] node {} (O3);
\end{tikzpicture}
\caption{Influence diagram for a private plan-time Markov game.}
\label{fig:private:ck:game:nf}
\end{figure}

Upon careful examination of private plan-time Markov games, it becomes apparent that they can be constructed using slave games. This removes the need for the explicit creation of the former, allowing for the generation of its components as required. Slave games are linked to partially observable Markov decision processes, while private plan-time Markov games have a more direct connection. However, it is important to note that private plan-time Markov games were introduced as surrogate games for the original slave games, and their objectives remain the same. Solving a private plan-time Markov game aims to determine the agent's best-response policy, denoted as $a^i_{0:}$. Interestingly, a one-to-one mapping exists between private plan-time histories in these games and private histories in slave games, \ie $$\varphi^i_{M'}\colon \left\{ \begin{array}{lcl} \mathcal{Y}^i &\to& \mathcal{O}^i \\ y^i_t\doteq (s_0,o^i_t,a^{\neg i}_{:t-1}) &\mapsto& o^i_t\end{array} \right.$$ This partially explains why best-response policies achieve identical performance indices in both settings. It is worth noticing that because a private plan-time Markov game $M'(a^{\neg i}_{0:})$ is a Markov game, one can actually optimally solve it by breaking it into a sequence of subgames $\langle M'(a^{\neg i}_{0:}), M'(a^{\neg i}_{1:}), M'(a^{\neg i}_{2:}), \ldots, M'(a^{\neg i}_{\ell-1:})\rangle$ solved recursively.

\begin{definition}
\label{def:state:value:fct:pckmg}
Consider a private plan-time Markov game $M'(a^{\neg i}_{0:})$ \wrt the slave game $M(a^{\neg i}_{0:})$. The $t$-step optimal value function $\upsilon^i_{M'(a^{\neg i}_{t:}),\gamma, t}\colon \mathcal{Y}^i_t \to \mathbb{R}$ of a private plan-time Markov subgame $M'(a^{\neg i}_{t:})$ is given by: for any private plan-time history $y^i_t$,
\begin{align*}
\upsilon^i_{M'(a^{\neg i}_{t:}),\gamma, t}(y_t^i) &\doteq\textstyle \max_{a^i_{t:}\in \mathcal{A}^i_{t:}} \mathbb{E}\{\sum_{\tau=t}^{\ell-1}\gamma^{\tau-t}\cdot r^i_{X_\tau U_\tau}| y^i_t, a_{t:}\}.
\end{align*}
\end{definition}

The optimal state-value functions of private plan-time Markov game $M'(a^{\neg i}_{0:})$ exhibit a special property: they are also solutions of the \citeauthor{bellman}'s optimality equations.

\begin{restatable}{lem}{lembellmanoptimalityeqn}
The $t$-step optimal state-value function $\upsilon^i_{M'(a^{\neg i}_{t:}),\gamma, t}\colon \mathcal{Y}^i_t\to \mathbb{R}$ is the solution of \citeauthor{bellman}'s optimality equations, \ie for any arbitrary private plan-time history $y^i_t$,
\begin{align*}
\upsilon^i_{M'(a^{\neg i}_{t:}),\gamma, t}(y_t^i) &= \textstyle \max_{u^i_t\in \mathcal{U}^i}~q^i_{M'(a^{\neg i}_{t:}),\gamma, t}(y_t^i,u^i_t),\\
q^i_{M'(a^{\neg i}_{t:}),\gamma, t}(y_t^i,u^i_t) &\doteq r^i_{M'}(y_t^i,u^i_t) + \gamma \!\!\!\sum_{z^i_{t+1}\in \mathcal{Z}^i}\!\!\! \omega^i_{M'}(z^i_{t+1}|y_t^i,u^i_t)\cdot \upsilon^i_{M'(a^{\neg i}_{t+1:}),\gamma, t+1}(\rho^i_{M'}(y_t^i,u^i_t,z^i_{t+1})),
\end{align*}
with boundary condition $\upsilon^i_{M',\gamma, \ell}(\cdot) = q^i_{M',\gamma, \ell}(\cdot) \doteq 0$.
\end{restatable}

\begin{proof}
The proof follows directly from the definition of optimal state-value functions for private plan-time Markov game $M'(a^{\neg i}_{t:})$, \cf Definition \ref{def:state:value:fct:pckmg}, \ie
\begin{align*}
\upsilon^i_{M'(a^{\neg i}_{t:}),\gamma, t}(y^i_t) &\textstyle \doteq \max_{a^i_{t:} \in \mathcal{A}^i_{t:}}~ \mathbb{E}\{ \textstyle{\sum_{\tau=t}^{\ell-1}}~\gamma^{\tau-t} \cdot r^i_{X_\tau,U_\tau} | y^i_t, a_{t:} \}.
\end{align*}
Let $\upsilon_{M'(a^{\neg i}_{t:}),\gamma, t}^{i,a^i_{t:}}\colon y^i_t \mapsto \mathbb{E}\{ \textstyle{\sum_{\tau=t}^{\ell-1}}~\gamma^{\tau-t} \cdot r^i_{X_\tau,U_\tau} | y^i_t, a_{t:} \}$ be a $t$-step value function under the policy $a^i_{t:}$ in the slave game $M'(a^{\neg i}_{t:})$. Re-arranging terms by decoupling the immediate reward from the future rewards leads successively to the following expressions:
\begin{align*}
\upsilon^i_{M'(a^{\neg i}_{t:}),\gamma, t}(y^i_t)
&= \max_{(u^i_t,a^i_{t+1:}) \in \mathcal{A}^i_{t:}}~ \mathbb{E}\{ r^i_{X_t,U_t} + \gamma \upsilon^i_{M'(a^{\neg i}_{t+1:}),\gamma, t+1}(\rho^i_{M'}(y^i_t,u^i_t,Z^i_{t+1})) | y^i_t, u^i_t,a^{\neg i}_t, a_{t+1:} \},
\end{align*}
where $\upsilon^{i,a_{t+1:}^i}_{M'(a^{\neg i}_{t+1:}),\gamma, t+1}(\rho^i_{M'}(y^i_t,u^i_t,z^i_{t+1})) \doteq \mathbb{E}\{\sum_{\tau=t+1}^{\ell-1}~\gamma^{\tau-t-1} \cdot r^i_{X_\tau,U_\tau} |\rho^i_{M'}(y^i_t,u^i_t,z^i_{t+1}), a_{t+1:}\}$. It is worth noticing that policy $(u^i_t,a^i_{t+1:})$ rooted at history $o^i_t$ can be written recursively as a tree $\langle u^i_t, (a^{i,z^i}_{t+1:})_{z^i\in \mathcal{Z}^i}\rangle$ with initial action being $u^i_t$ and the subtree after seeing private observation $z^i$ being $a^{i,z^i}_{t+1:}$. Consequently, replacing $(u^i_t,a^i_{t+1:})$ by $\langle u^i_t, (a^{i,z^i}_{t+1:})_{z^i\in \mathcal{Z}^i}\rangle$ results in:
\begin{align*}
\upsilon^i_{M'(a^{\neg i}_{t:}),\gamma, t}(y^i_t)
&= \!\!\! \max_{\langle u^i_t, (a^{i,z^i}_{t+1:})_{z^i\in \mathcal{Z}^i}\rangle \in \mathcal{A}^i_{t:}}\!\!\!\mathbb{E}\{ r^i_{X_t,U_t} + \gamma \upsilon^{i,a^{i,z^i_{t+1}}_{t+1:}}_{M'(a^{\neg i}_{t+1:}),\gamma,t+1}(\rho^i_{M'}(y^i_t,u^i_t,z^i_{t+1})) | y^i_t, u^i_t,a^{\neg i}_t,a_{t+1:} \}.
\end{align*}
One can split the $\max$ operator into two parts, the first one over private actions $\max_{u^i_t\in \mathcal{U}^i}$ and the other one over the future decision rules $\max_{(a^{i,z^i}_{t+1:})_{z^i\in \mathcal{Z}^i}\in \mathcal{A}^i_{t+1:}}$, as follows:
\begin{align*}
\upsilon^i_{M'(a^{\neg i}_{t:}),\gamma, t}(y^i_t)
&= \max_{u^i_t \in \mathcal{U}^i} \mathbb{E}\{ r^i_{X_t,U_t} + \gamma \max_{a^{i,z_{t+1}^i}_{t+1:} \in \mathcal{A}^i_{t+1:}}~ \upsilon^{i,a^{i,z_{t+1}^i}_{t+1:}}_{M'(a^{\neg i}_{t+1:}),\gamma,t+1}(\rho^i_{M'}(y^i_t,u^i_t,z^i_{t+1})) | y^i_t, u^i_t,a^{\neg i}_t \}.
\end{align*}
By Definition \ref{def:state:value:fct:pckmg}, we know that the following holds:
\begin{align*}
\upsilon^{i}_{M'(a^{\neg i}_{t+1:}),\gamma,t+1}(\rho^i_{M'}(y^i_t,u^i_t,z^i_{t+1}))& \doteq \max_{a^{i,z_{t+1}^i}_{t+1:} \in \mathcal{A}^i_{t+1:}}~ \upsilon^{i,a^{i,z_{t+1}^i}_{t+1:}}_{M'(a^{\neg i}_{t+1:}),\gamma,t+1}(\rho^i_{M'}(y^i_t,u^i_t,z^i_{t+1})),
\end{align*}
which leads us to the following expression:
\begin{align*}
\upsilon^i_{M'(a^{\neg i}_{t:}),\gamma, t}(y^i_t) &= \max_{u^i_t \in \mathcal{U}^i} \mathbb{E}\{ r^i_{X_t,U_t} + \gamma \upsilon^i_{M'(a^{\neg i}_{t+1:}),\gamma,t+1}(\rho^i_{M'}(y^i_t,u^i_t,Z^i_{t+1})) | y^i_t, u^i_t,a^{\neg i}_t \}.
\end{align*}
By the application of Definition \ref{def:private:plantime:slave:game}, we know that $r^i_{M'}(y_t^i,u^i_t) \doteq \mathbb{E}\{r^i_{X_t,U_t}|y_t^i, u_t^i,a^{\neg i}_t\} $, which results in the following expression:
\begin{align*}
\upsilon^i_{M'(a^{\neg i}_{t:}),\gamma, t}(y^i_t) &= \max_{u^i_t \in U^i} \left\{ r^i_{M'}(y_t^i,u^i_t)+\gamma\mathbb{E}\{ \upsilon^i_{M'(a^{\neg i}_{t+1:}),\gamma,t+1}(\rho^i_{M'}(y^i_t,u^i_t,Z^i_{t+1})) | y^i_t, u^i_t,a^{\neg i}_t \}\right\}.
\end{align*}
The expansion over the private observations of the expectation yields the following expression:
\begin{align*}
\upsilon^i_{M'(a^{\neg i}_{t:}),\gamma, t}(y^i_t) &= \max_{u^i_t\in \mathcal{U}^i}~q^{i}_{M'(a^{\neg i}_{t:}),\gamma, t}(y_t^i,u^i_t),\\
q^{i}_{M'(a^{\neg i}_{t:}),\gamma, t}(y_t^i,u^i_t) &\doteq r^i_{M'}(y_t^i,u^i_t) + \gamma \sum_{z^i_{t+1}\in \mathcal{Z}^i} \omega^i_{M'}(z^i_{t+1}|y_t^i,u^i_t)\cdot \upsilon^{i}_{M'(a^{\neg i}_{t+1:}),\gamma, t+1}(\rho^i_{M'}(y_t^i,u^i_t,z^i_{t+1})).
\end{align*}
Which ends the proof.
\end{proof}

The following proposition explicitly connects the games $M(a^{\neg i}_{0:})$ and $M'(a^{\neg i}_{0:})$, such that solving either game solves the other.

\begin{restatable}{pro}{propequivalentprivateckmg}
\label{prop:equivalent:private:ckmg}
Let $M(a^{\neg i}_{0:})$ be a slave game and $M'(a^{\neg i}_{0:})$ be the corresponding private plan-time Markov game. If we let $a_{0:}^i$ be the best-response policy of $M'(a^{\neg i}_{0:})$, then it is also the best-response policy of $M(a^{\neg i}_{0:})$.
\end{restatable}

\begin{proof}
The proof consists in demonstrating that optimal state-value functions for $M'(a^{\neg i}_{t:})$ and $M(a^{\neg i}_{0:})$ achieve the same performance index on equivalent inputs, \ie for any time step $t$ and private plan-time history $y^i_t$, $\upsilon^{i}_{M'(a^{\neg i}_{t:}),\gamma, t}(y_t^i) = \upsilon^{i}_{M(a^{\neg i}_{0:}),\gamma, t}(\varphi^i_{M'}(y_t^i))$. Consequently, the best-response policy $a_{t:}^i$ solution of \citeauthor{bellman}'s optimality equations for $M'(a^{\neg i}_{t:})$ achieves the best performance index also for $M(a^{\neg i}_{0:})$. Thus, it is also a best-response policy for $M(a^{\neg i}_{0:})$. Starting with the definition of the optimal action-value function for $M(a^{\neg i}_{0:})$ at time step $t$ and private history $\varphi^i_{M'}(y_t^i)$, we have that:
\begin{align*}
q^{i}_{M(a^{\neg i}_{0:}),\gamma, t}(\varphi^i_{M'}(y_t^i),u^i_t) &\doteq \max_{a^i_{t+1:}\in \mathcal{A}^i_{t+1:}} \mathbb{E}\left\{\sum_{\tau=t}^{\ell-1} \gamma^{\tau-t} \cdot r^i_{X_\tau,U_\tau} | \varphi^i_{M'}(y_t^i),u^i_t, a^i_{t+1:},a^{\neg i}_{0:}\right\}.
\end{align*}
There is no loss in replacing $(\varphi^i_{M'}(y_t^i),a^{\neg i}_{:t-1})$ by $y_t^i$, consequently it follows that:
\begin{align*}
q^{i}_{M(a^{\neg i}_{0:}),\gamma, t}(\varphi^i_{M'}(y_t^i),u^i_t) &= \max_{a^i_{t+1:}\in \mathcal{A}^i_{t+1:}} \mathbb{E}\left\{\sum_{\tau=t}^{\ell-1} \gamma^{\tau-t} \cdot r^i_{X_\tau,U_\tau} | y_t^i, u^i_t, a^{\neg i}_t, a_{t+1:}\right\}.
\end{align*}
The following splits the immediate and the future rewards:
\begin{align*}
q^{i}_{M(a^{\neg i}_{0:}),\gamma, t}(\varphi^i_{M'}(y_t^i),u^i_t)&= \max_{a^i_{t+1:}\in \mathcal{A}^i_{t+1:}}\mathbb{E}\left\{r^i_{X_t,U_t}+ \gamma \sum_{\tau=t+1}^{\ell-1} \gamma^{\tau-t-1} \cdot r^i_{X_\tau,U_\tau}| y_t^i, u^i_t, a^{\neg i}_t, a_{t+1:}\right\}.
\end{align*}
Re-arranging terms leads to the following expression:
\begin{align*}
q^{i}_{M(a^{\neg i}_{0:}),\gamma, t}(\varphi^i_{M'}(y_t^i),u^i_t) &= \max_{a^i_{t+1:}\in \mathcal{A}^i_{t+1:}}\mathbb{E}\left\{r^i_{X_t,U_t}|y_t^i,u^i_t,a^{\neg i}_t\right\} + \gamma \mathbb{E}\left\{ \sum_{\tau=t+1}^{\ell-1} \gamma^{\tau-t-1} \cdot r^i_{X_\tau,U_\tau}| y_t^i, u^i_t, a^{\neg i}_t, a_{t+1:}\right\}.
\end{align*}
By the application of Definition \ref{def:private:plantime:slave:game}, we know that $r^i_{M'}(y_t^i,u^i_t) \doteq \mathbb{E}\{r^i_{X_t,U_t}|y_t^i,u^i_t,a^{\neg i}_t\} $, which results in the following expression:
\begin{align*}
q^{i}_{M(a^{\neg i}_{0:}),\gamma, t}(\varphi^i_{M'}(y_t^i),u^i_t) &= r^i_{M'}(y_t^i,u^i_t) + \gamma \max_{a^i_{t+1:}\in \mathcal{A}^i_{t+1:}} \mathbb{E}\left\{\sum_{\tau=t+1}^{\ell-1} \gamma^{\tau-t-1} \cdot r^i_{X_\tau,U_\tau}| y_t^i,(u^i_t, a^i_{t+1:}, a^{\neg i}_{t+1:})\right\}.
\end{align*}
The expansion of all private observations of agent $i$ produces the following expression:
\begin{align*}
q^{i}_{M(a^{\neg i}_{0:}),\gamma, t}(\varphi^i_{M'}(y_t^i),u^i_t) &= r^i_{M'}(y_t^i, u_t^i) + \gamma \mathbb{E}_{z^i_{t+1}\sim\Pr\{\cdot|y_t^i,u^i_t,a^{\neg i}_{0:}\}}\{\upsilon^{i}_{M'(a^{\neg i}_{t+1:}),\gamma, t+1}(\rho^i_{M'}( y_t^i, u_t^i, z_{t+1}^i ))\}
\end{align*}
By the application of Definition \ref{def:private:plantime:slave:game}, we know that $\omega^i_{M'}(z_{t+1}^i|y_t^i, u_t^i) \doteq \Pr\{z^i_{t+1}|y_t^i,u^i_t,a^{\neg i}_t\}$, which results in the following expression:
\begin{align*}
q^{i}_{M(a^{\neg i}_{0:}),\gamma, t}(\varphi^i_{M'}(y_t^i),u^i_t) &=\textstyle r^i_{M'}(y_t^i, u_t^i) + \gamma \sum_{z_{t+1}^i\in \mathcal{Z}^i} \omega^i_{M'}(z_{t+1}^i|y_t^i, u_t^i) \cdot \upsilon^{i}_{M'(a^{\neg i}_{t+1:}),\gamma, t+1}(\rho^i_{M'}( y_t^i, u_t^i, z_{t+1}^i )).
\end{align*}
By the application of the definition of $q^{i}_{M'(a^{\neg i}_{t:}),\gamma, t}(y_t^i, u_t^i)$, it follows that:
\begin{align*}
q^{i}_{M(a^{\neg i}_{0:}),\gamma, t}(\varphi^i_{M'}(y_t^i), u_t^i)&= q^{i}_{M'(a^{\neg i}_{t:}),\gamma, t}(y_t^i, u_t^i).
\end{align*}
The greedy selection of the best-response policy $a^i_{0:}$ \wrt joint policy $a^{\neg i}_{0:}$ of the other agents is given as follows: at time step $t$, and private history $\varphi^i_{M'}(y_t^i)$,
\begin{align*}
a_t^i(\varphi^i_{M'}(y_t^i)) &\in\textstyle \argmax_{u_t^i\in \mathcal{U}^i}~q^{i}_{M(a^{\neg i}_{0:}),\gamma, t}(\varphi^i_{M'}(y_t^i),u^i_t),\\
&\in\textstyle \argmax_{u_t^i\in \mathcal{U}^i}~q^{i}_{M'(a^{\neg i}_{t:}),\gamma, t}(y_t^i, u_t^i).
\end{align*}
Which ends the proof.
\end{proof}

Proposition \ref{prop:equivalent:private:ckmg} demonstrates that a solution for $M'(a^{\neg i}_{0:})$ can also serve as a solution for the original slave game $M(a^{\neg i}_{0:})$. Nonetheless, like other historical data, private plan-time histories grow in size with each time step, making them less feasible to use. Hence, we suggest a substitute statistic, private occupancy states, to tackle this problem. Nevertheless, before we delve deeper into this subject, it is essential to establish the definition of private occupancy states.

\begin{definition}
\label{def:private:occupancy:state}
Each agent $i$ has a private occupancy state at a specific time step $t$, represented as $s^i_t$ (also denoted as $s^{i,o_t^i}_t$ when needed to make the dependence from a given private history $o_t^i$ more explicit). This state is the posterior probability distribution over hidden states and joint histories based on their private plan-time history $y^i_t$, expressed as $\Pr\{X_t,O_t|y^i_t\}$. The initial condition is defined as $s^i_0 = s_0$.
\end{definition}

The concept of hierarchical information-sharing in two-agent games has been studied extensively, with research indicating that the private occupancy state can be expressed in a more concise manner \citep{stochastic,cooperative,onesided,optimally}. If agent 2 has unrestricted access to agent 1's actions and observations at every time step $t$ without incurring any cost, denoted by $(u^1_{t-1},z^1_t) \sqsubseteq z_t^2$, then the private occupancy state at that specific time step $t$ can be represented as a posterior probability distribution over the hidden states, commonly referred to as the belief state. This is denoted as $\Pr\{ x_t|y^2_t\}$, where $x$ belongs to the set $\mathcal{X}$. Furthermore, if agent 2 can observe the underlying state of the game, \ie one-sidedness, then the private occupancy state at a specific time step $t$ is equivalent to that underlying state. The conciseness of private occupancy states is crucial in identifying strong uniform continuity properties for the optimal state-value functions.

The optimal state- and action-value functions of a slave game can be better understood through private occupancy states. Revealing the uniform continuity property of those value functions allows for knowledge transfer between different private occupancy states. While private occupancy states can be more concise for a single slave game, it is important to consider knowledge transfer between multiple slave games. Transfer of knowledge from one slave game to another one is important because to solve a master game, we need to solve multiple slave games. Defining private occupancy states as posterior distributions over hidden states and histories of other agents $\neg i$ or as private histories of agent $i$, as shown by \citet{structure}, is adequate for solving a single slave game optimally. However, using a more concise statistic for a single slave game would make knowledge transfer between slave games more difficult, especially in exhibiting the relationship between the optimal state-value function of the master game and the corresponding slave games. This insight explains why previous attempts to establish convexity properties of the optimal state-value functions of the master game failed. The following section explores how private occupancy states are sufficient for optimally solving any slave game.

\section{Connecting Slave and Master Games}
\label{sec:connections}
This section will discuss an alternative method to solve a slave game by utilizing private occupancy states instead of private plan-time histories. We will analyze if using private occupancy states is enough to solve the reformulated slave games optimally. Furthermore, we will establish a correlation between private occupancy states and occupancy states to merge the state-value functions of master and slave games. We will first define the Markov decision process that private occupancy states induce.

\begin{figure}[!ht]
\centering
\begin{tikzpicture}[->,-latex,auto,node distance=3.75cm,semithick, square/.style={regular polygon,regular polygon sides=4}]
\tikzstyle{every state}=[draw=black,text=black,inner color= white,outer color= white,draw= black,text=black, drop shadow]
\tikzstyle{place}=[thick,draw=sthlmBlue,fill=blue!20,minimum size=12mm, opacity=.5]
\tikzstyle{red place}=[square,place,draw=sthlmRed,fill=sthlmLightRed,minimum size=18mm]
\tikzstyle{green place}=[diamond,place,draw=sthlmGreen,fill=sthlmLightGreen,minimum size=15mm]
\node[fill=white, scale=.75] (T) at (-3.5,-3.75) {};
\node[fill=white, scale=.75] (T0) at (0,-3.75) {$0$};
\node[fill=white, scale=.75] (T1) [right of=T0,node distance=3.75cm, fill=white] {$t$};
\node[fill=white, scale=.75] (T2) [right of=T1,node distance=3.75cm] {$t+1$};
\node[fill=white, scale=.75] (T3) [right of=T2,fill=white, node distance=3.75cm] {$t+2$};
\node[fill=white, scale=.75] (T4) [right of=T3,fill=white, node distance=2.5cm] {$\ldots$};
\draw[->,-latex,dashed,very thin, color=black,anchor=mid] (T3) -- (T4);
\draw[->,-latex, very thin, color=black, anchor=mid] (T2) -- (T3);
\draw[->,-latex, very thin, color=black, anchor=mid] (T1) -- (T2);
\draw[->,-latex,dashed,very thin, color=black,anchor=mid] (T0) -- (T1);
\draw[->,-latex,rounded corners,very thin, color=black,anchor=mid] (T) node[fill=white, scale=.75]{Time} -- (T0);
\node[state,place, scale=.75] (S0) {$s^i_0$};
\node[state,place, scale=.75] (S1) [right of=S0] {$s^i_t$};
\node[state,place, scale=.75] (S2) [ right of=S1] {$s^i_{t+1}$};
\node[state,place, scale=.75] (S3) [ right of=S2] {$s^i_{t+2}$};
\node[, scale=.75] (S4) [ right of=S3,node distance=2.5cm] {$\cdots$};
\path[very thin] (S0) edge[dashed] node[midway,text=black,draw=none,scale=.7, above=-5pt] {} (S1)
(S1) edge node[midway,text=black,draw=none,scale=.7, above=-5pt] {} (S2)
(S2) edge node[midway,text=black,draw=none,scale=.7, above=-5pt] {} (S3)
(S3) edge[dashed] (S4);
\node[, scale=.75] (A_SLACK) [below of=S0,node distance=2.6cm] {};
\node[state,red place, scale=.65] (A0) [below right of=A_SLACK,node distance=2cm] {$u^i_{0}$};
\node[state,red place, scale=.65] (A1) [right of=A0,node distance=4.25cm] {$u^i_t$};
\node[state,red place, scale=.65] (A2) [right of=A1,node distance=4.25cm] {$u^i_{t+1}$};
\node[, scale=.75] (A3) [right of=A2,node distance=5.5cm] {};
\path[very thin] (A0) edge [out=90, in=-155, dashed] node[scale=.75] {$z^i_t$} (S1)
(A1) edge [out=90, in=-155] node[scale=.75] {$z^i_{t+1}$} (S2)
(A2) edge [out=90, in=-155] node[scale=.75] {$z^i_{t+2}$} (S3);
\node[state,green place, scale=.75] (O0) [below of=S0,node distance=2cm] {$r^i_0$};
\node[state,green place, scale=.75] (O1) [below of=S1,node distance=2cm] {$r^i_t$};
\node[state,green place, scale=.75] (O2) [below of=S2,node distance= 2cm] {$r^i_{t+1}$};
\node[state,green place, scale=.75] (O3) [below of=S3,node distance= 2cm] {$r^i_{t+2}$};
\node[, scale=.75] (O4) [ right of=O3,node distance=2.5cm] {$\cdots$};
\node[inner sep=0pt] (prof) at (-3.35,-1.5) {\includegraphics[width=.1\textwidth]{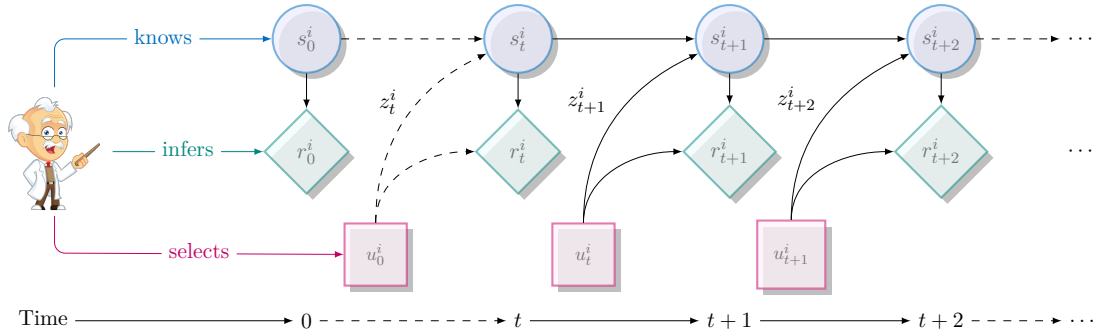}};
\draw[->,-latex,rounded corners,very thin, color=sthlmBlue,anchor=mid] (prof) -- (-3.35,0) --node[draw=none, fill=white, scale=.75] {knows} (S0);
\draw[->,-latex,rounded corners,very thin, color=sthlmGreen,anchor=mid] (prof) --node[draw=none, fill=white, scale=.75] {infers} (O0);
\draw[->,-latex,rounded corners,very thin, color=sthlmRed,anchor=mid] (prof) -- (-3.35,-2.85) --node[draw=none, fill=white, scale=.75] {selects} (A0);
\path[very thin] (S0) edge node {} (O0)
(S1) edge node {} (O1)
(A0) edge[out=90, in=-180, dashed] node {} (O1)
(S2) edge node {} (O2)
(A1) edge[out=90, in=-180] node {} (O2)
(S3) edge node {} (O3)
(A2) edge[out=90, in=-180] node {} (O3);
\end{tikzpicture}
\caption{Influence diagram for a private occupancy-state Markov game.}
\label{fig:ck:game:nf:privateomg}
\end{figure}

\begin{definition}
\label{def:slave:problem:as:mdp}
Let $M$ be a master game, and let $M(a^{\neg i}_{0:})$ (respectively $M'(a^{\neg i}_{0:})$) be its corresponding slave game. A private occupancy-state Markov decision process is given by a tuple $M''(a^{\neg i}_{0:}) \doteq (\mathcal{S}^i,\mathcal{U}^i, \mathcal{Z}^i, \omega^i_{M''}, \rho^i_{M''}, p^i_{M''}, r^i_{M''})$. Here, $\mathcal{S}^i$ is the finite private occupancy-state space, $\mathcal{U}^i$ represents the finite action space, and $\mathcal{Z}^i$ is the finite observation space. Mappings $\omega^i_{M''}$, $\rho^i_{M''}$, $p^i_{M''}$, and $r^i_{M''}$ are formally described as follows.
\begin{itemize}
\item The observation function $\omega^i_{M''}\colon \mathcal{S}^i \times \mathcal{U}^i \times \mathcal{Z}^i \to [0,1]$ gives the probability of private observation $z^i_{t+1}$ upon taking action $u^i_t$ starting in private occupancy state $s^{i,o^i_t}_t$, \ie
\begin{align*}
\omega^i_{M''}(z^i_{t+1}|s^{i,o^i_t}_t,u^i_t) &\doteq\sum_{x_t,x_{t+1}\in \mathcal{X}}\sum_{u_t^{\neg i}\in \mathcal{U}^{\neg i}}\sum_{z_{t+1}^{\neg i}\in \mathcal{Z}^{\neg i}} \sum_{o^{\neg i}_t\in \mathcal{O}^{\neg i}_t} p^{u_tz_{t+1}}_{x_tx_{t+1}}\cdot a^{\neg i}_t(u_t^{\neg i}|o_t^{\neg i} ) \cdot s^{i,o^i_t}_t(x_t, o_t).
\end{align*}
\item The transition rule $\rho^i_{M''}\colon \mathcal{S}^i\times \mathcal{U}^i\times \mathcal{Z}^i \to \mathcal{S}^i$ gives the next private occupancy state upon receiving private observation $z^i_{t+1}$ after taking private action $u^i_t$ starting in private occupancy state $s^i_t$ , \ie for any arbitrary hidden state $x_{t+1}$ and joint history $(o_t,u_t,z^i_{t+1}, z^{\neg i}_{t+1})$,
\begin{align*}
s^i_{t+1}(x_{t+1}, (o_t,u_t,z^i_{t+1}, z^{\neg i}_{t+1})) &~\propto~ \sum_{x_t\in \mathcal{X}} p^{u_t\langle z^i_{t+1}, z^{\neg i}_{t+1} \rangle}_{x_tx_{t+1}} \cdot a^{\neg i}_t( u_t^{\neg i} | o_t^{\neg i} )\cdot s^i_t(x_t, o_t).
\end{align*}
\item The transition probability function $p^i_{M''}\colon \mathcal{S}^i\times \mathcal{U}^i \times \mathcal{S}^i\to [0,1]$ gives the probability of private occupancy state $s^i_{t+1}$ upon taking action $u_t^i$ in private occupancy state $s^i_t$ , \ie
\begin{align*}
p^i_{M''}(s^i_t, u^i_t, s^i_{t+1}) &\doteq \sum_{z_{t+1}^i\in \mathcal{Z}^i} \omega^i_{M''}(z^i_{t+1}|s^i_t,u^i_t) \cdot \delta_{\rho^i_{M''}(s^i_t, u^i_t, z_{t+1}^i)}( s^i_{t+1} ).
\end{align*}
\item The reward model $r^i_{M''}\colon \mathcal{S}^i\times \mathcal{U}^i \to \mathbb{R}$ gives the immediate expected reward upon taking private action $u_t^i$ in private occupancy state $s^i_t$ , \ie
\begin{align*}
r^i_{M''}(s^i_t,u^i_t) &\doteq \sum_{x_t\in \mathcal{X}}\sum_{o_t\in \mathcal{O}_t} s^i_t(x_t,o_t) \sum_{u^{\neg i}_t\in \mathcal{U}^{\neg i}} a^{\neg i}_t(u_t^{\neg i}|o^{\neg i}_t)\cdot r^i_{x_t,u_t}.
\end{align*}
\end{itemize}
\end{definition}

The Markov decision process associated with private occupancy states is called the private occupancy-state Markov decision process, as depicted in Figure \ref{fig:ck:game:nf:privateomg}. Using occupancy states for master game reformulations raises questions similar to those regarding the use of private occupancy states in the context of slave games. \citeauthor{bellman}'s principle of optimality allows us to solve $M''(a^{\neg i}_{0:})$ by solving recursively a sequence of smaller Markov decision processes $\langle M''(a^{\neg i}_{0:}), M''(a^{\neg i}_{1:}), \ldots, M''(a^{\neg i}_{\ell-1:})\rangle$. The reader may wonder if private occupancy states are sufficient statistics for optimally solving slave games through private occupancy-state Markov decision processes. One way to prove this claim is to demonstrate that private occupancy states are sufficient statistics of private plan-time histories for optimally solving $M''(a^{\neg i}_{0:})$, as private plan-time histories are proven to preserve the ability to solve the original slave games optimally, as outlined in Proposition \ref{prop:equivalent:private:ckmg}.

\subsection{Sufficiency of Private Occupancy States}
This subsection illustrates how private occupancy states can act as sufficient statistics for the complete data accessible to the trusted third party who is reasoning on behalf of agent $i$ in any given slave game $M(a^{\neg i}_{0:})$ (or $M'(a^{\neg i}_{0:})$). These statistics are sufficient for optimally solving any slave game. Using this as a basis, we prove that the optimal action- and state-value functions of a slave game $M(a^{\neg i}_{0:})$ (or $M'(a^{\neg i}_{0:})$) depend on private plan-time histories only through private occupancy states.

\begin{restatable}{thm}{thmprivatesufficiencyoccupancystate}
\label{thm:private:sufficiency:occupancy:state}
For each agent $i$, the private occupancy state at time step $t$, defined as $s^i_t \doteq \Pr\{X_t,O_t|y^i_t\}$, is a sufficient statistic of the private plan-time history $y_t^i$ for optimally solving $M'(a^{\neg i}_{0:})$, \ie $$ \textstyle \mathbb{E}\{\sum_{\tau=t}^{\ell-1} \gamma^{\tau-t}\cdot r^i_{M''}(S^i_\tau,U^i_\tau) | s^i_t, a^i_{t:}, a^{\neg i}_{0:} \} = \mathbb{E}\{\sum_{\tau=t}^{\ell-1} \gamma^{\tau-t}\cdot r^i_{M'}(Y^i_\tau,U^i_\tau) | y^i_t, a^i_{t:}, a^{\neg i}_{0:} \} .$$ Also, an optimal solution for $M''(a^{\neg i}_{0:})$ is also a solution for $M(a^{\neg i}_{0:})$ (respectively, $M'(a^{\neg i}_{0:})$).
\end{restatable}

To prove this theorem, we must establish three properties. Firstly, we need to prove that private occupancy states are enough to accurately predict the probability of an observation being received by agent $i$ after taking action $u^i_t$, as stated in Lemma \ref{lem:sufficiency:occupancy:state:observation}. Secondly, we must show that the private occupancy states describe a Markov process, meaning that the current private occupancy state is enough to estimate the next one, as stated in Lemma \ref{lem:sufficiency:private:occupancy:state:markov}. Finally, we must prove that the private occupancy states are sufficient to accurately estimate the immediate reward when taking action $u^i_t$, as stated in Lemma \ref{lem:sufficiency:occupancy:state:reward}. In the following, we will establish these three fundamental properties to demonstrate the sufficiency of private occupancy states of private plan-time histories for optimally solving slave games.

\begin{restatable}{lem}{lemsufficiencyoccupancystateobservation}
\label{lem:sufficiency:occupancy:state:observation}
For each agent $i$, the private occupancy state at time step $t$, defined as $s^i_t \doteq \Pr\{X_t,O_t|y^i_t\}$, is a sufficient statistic of the private plan-time history $y_t^i$ for prediction of the next private observation $z^i_{t+1}$ following the execution of the private action $u^i_t$---\ie $\omega^i_{M'}(z^i_{t+1} | y^i_t,u^i_t) = \omega^i_{M''}(z^i_{t+1} | s^i_t,u^i_t)$.
\end{restatable}

\begin{proof}
For each agent $i$, the probability of private observation $z^i_{t+1}$ at the end of time step $t+1$ can be rewritten as the sum over all possible hidden states $x_t$, and joint histories $o_t^{\neg i}$, and private actions $u_t^{\neg i}$ at the end of time step $t$, as well as all joint observations $z_{t+1}^{\neg i}$, and hidden states $x_{t+1}$ at the end of time step $t+1$, \ie
\begin{align*}
\omega^i_{M'}( z^i_{t+1}|y^i_t,u^i_t )&\doteq \Pr\{ z^i_{t+1}|y^i_t, a^{\neg i}_t, u^i_t \} \\
&=
\sum_{x_t\in \mathcal{X}}\sum_{o_t^{\neg i}\in \mathcal{O}_t^{\neg i}}\sum_{u_t^{\neg i}\in \mathcal{U}^{\neg i}}\sum_{z_{t+1}^{\neg i}\in \mathcal{Z}^{\neg i}} \sum_{x_{t+1}\in \mathcal{X}}
\Pr\{ x_t, o_t^{\neg i}, u_t^{\neg i},z_{t+1}^{\neg i},x_{t+1} | y^i_t, a^{\neg i}_t, u^i_t \}.
\end{align*}
The application of Bayes' rule and the replacement of $y^i_t$ by $(s_0,o^i_t,a_{:t-1}^{\neg i})$ yield
\begin{align}
\omega^i_{M'}( z^i_{t+1}|y^i_t,u^i_t ) &=
\sum_{x_t\in \mathcal{X}}\sum_{o_t^{\neg i}\in \mathcal{O}_t^{\neg i}}\sum_{u_t^{\neg i}\in \mathcal{U}^{\neg i}}\sum_{z_{t+1}^{\neg i}\in \mathcal{Z}^{\neg i}} \sum_{x_{t+1}\in \mathcal{X}}
\frac{\Pr\{ s_0, x_t, o_t, u_t,z_{t+1},x_{t+1}, a_{:t}^{\neg i} \}}{\Pr\{ y^i_t, a^{\neg i}_t, u^i_t \}}.
\label{eq:thm:sufficiency:occupancy:state:factor:8:1}
\end{align}
The expansion of the joint probability in \eqref{eq:thm:sufficiency:occupancy:state:factor:8:1} as a product of conditional probabilities yields
\begin{align}
\omega^i_{M'}( z^i_{t+1}|y^i_t,u^i_t ) &= \sum_{x_t\in \mathcal{X}}\sum_{o_t^{\neg i}\in \mathcal{O}_t^{\neg i}}\sum_{u_t^{\neg i}\in \mathcal{U}^{\neg i}}\sum_{z_{t+1}^{\neg i}\in \mathcal{Z}^{\neg i}} \sum_{x_{t+1}\in \mathcal{X}}\nonumber\\
&\textstyle\qquad\Pr\{ z_{t+1},x_{t+1}|s_0, x_t, o_t, u_t, a_{:t}^{\neg i} \}\cdot \Pr\{u^{\neg i}_t|s_0, x_t, o_t, a_{:t}^{\neg i},u^i_t \} \nonumber\\
&\textstyle\qquad\cdot \Pr\{x_t, o^{\neg i}_t|y^i_t, a_t^{\neg i},u^i_t \}\cdot \Pr\{y^i_t, a_t^{\neg i},u^i_t \}/\Pr\{ y^i_t, a^{\neg i}_t, u^i_t \}.
\label{eq:thm:sufficiency:occupancy:state:factor:8:2}
\end{align}
The first probability in \eqref{eq:thm:sufficiency:occupancy:state:factor:8:2} is simply the dynamics $(p^{u_tz_{t+1}}_{x_tx_{t+1}})$ of $M$, \ie
\begin{align}
\omega^i_{M'}( z^i_{t+1}|y^i_t,u^i_t ) &=\sum_{x_t\in \mathcal{X}}\sum_{u_t^{\neg i}\in \mathcal{U}^{\neg i}}\sum_{z_{t+1}^{\neg i}\in \mathcal{Z}^{\neg i}} \sum_{x_{t+1}\in \mathcal{X}} p^{u_tz_{t+1}}_{x_tx_{t+1}} \nonumber\\
&\textstyle\qquad \sum_{o_t^{\neg i}\in \mathcal{O}_t^{\neg i}} \Pr\{u^{\neg i}_t|s_0, x_t, o_t, a_{:t}^{\neg i},u^i_t \} \cdot \Pr\{x_t, o^{\neg i}_t|y^i_t, a_t^{\neg i},u^i_t \}.
\label{eq:thm:sufficiency:occupancy:state:factor:8:3}
\end{align}
The second probability in \eqref{eq:thm:sufficiency:occupancy:state:factor:8:3} defines the probability of taking joint action $u^{\neg i}_t$ upon experiencing joint history $o^{\neg i}_t$, \ie
\begin{align}
\omega^i_{M'}( z^i_{t+1}|y^i_t,u^i_t ) &=\sum_{x_t\in \mathcal{X}}\sum_{u_t^{\neg i}\in \mathcal{U}^{\neg i}}\sum_{z_{t+1}^{\neg i}\in \mathcal{Z}^{\neg i}} \sum_{x_{t+1}\in \mathcal{X}} p^{u_tz_{t+1}}_{x_tx_{t+1}} \sum_{o_t^{\neg i}\in \mathcal{O}_t^{\neg i}} \nonumber\\
&\textstyle\qquad a^{\neg i}_t(u_t^{\neg i}|o_t^{\neg i} ) \cdot \Pr\{x_t, o^{\neg i}_t|y^i_t, a_t^{\neg i},u^i_t \}.
\label{eq:thm:sufficiency:occupancy:state:factor:8:4}
\end{align}
The third probability in \eqref{eq:thm:sufficiency:occupancy:state:factor:8:4} defines the probability of the previous private occupancy state at hidden state $x_t$ and joint history $o_t$, \ie
\begin{align*}
\omega^i_{M'}( z^i_{t+1}|y^i_t,u^i_t ) &=\sum_{x_t\in \mathcal{X}}\sum_{u_t^{\neg i}\in \mathcal{U}^{\neg i}}\sum_{z_{t+1}^{\neg i}\in \mathcal{Z}^{\neg i}} \sum_{x_{t+1}\in \mathcal{X}} p^{u_tz_{t+1}}_{x_tx_{t+1}} \sum_{o_t^{\neg i}\in \mathcal{O}_t^{\neg i}} a^{\neg i}_t(u_t^{\neg i}|o_t^{\neg i} ) \cdot s^i_t(x_t, o_t)\\
&\doteq \omega^i_{M''}( z^i_{t+1}|s^i_t,u^i_t ) .
\end{align*}
Which ends the proof.
\end{proof}

Based on Lemma \ref{lem:sufficiency:occupancy:state:observation}, it is possible to make precise predictions about a private observation $z^i_{t+1}$ by knowing private occupancy state $s^i_t$, private decision $a^{\neg i}_t$, and private action $u^i_t$. This prediction can be made without knowing the private plan-time history $y^i_t$.

\begin{restatable}{lem}{lemsufficiencyprivateoccupancystatemarkov}
\label{lem:sufficiency:private:occupancy:state:markov}
For each agent $i$, the private occupancy state at time step $t$, defined as $s^i_t \doteq \Pr\{X_t,O_t|y^i_t\}$, is a sufficient statistic of the private plan-time history $y_t^i$ for predicting the subsequent private occupancy state following the execution the private action $u^i_t$ and the reception of the private observation $z^i_{t+1}$, \ie $ \rho^i_{M'}(y^i_t,u^i_t,z^i_{t+1})=\rho^i_{M''}(s^i_t,u^i_t,z^i_{t+1})$.
\end{restatable}

\begin{proof}
In demonstrating the sufficiency of private occupancy states, we also derive the rule for updating private occupancy states from one time step to the next. For the process defined by $M'(a^{\neg i}_{0:})$, the only information that agent $i$ obtains during time step $t$ is the fact that a particular decision rule $a^{\neg i}_t$ has been followed by agents $\neg i$, private action $u_t^i$ has been executed by agent $i$, and private observation $z^i_{t+1}$ has been received at the end of time step $t+1$ by agent $i$, then we can write
\begin{align*}
y^i_{t+1} &\doteq (y^i_t,a^{\neg i}_t,u^i_t, z^i_{t+1}).
\end{align*}
That is, $y^i_{t+1}$ represents the private plan-time history of agent $i$ before time step $t+1$ plus the additional information that a particular decision rule and private action were recorded and private observation occurred. By Definition \ref{def:private:occupancy:state}, private occupancy states are written as follows: for every hidden state $x_{t+1}$, history $o_t^{\neg i}$, action $u_t^{\neg i}$, and observation $z_{t+1}^{\neg i}$,
\begin{align}
s^i_{t+1}(x_{t+1}, (o_t,u_t,z_{t+1})) &\doteq \Pr\{ x_{t+1}, o_t^{\neg i},u_t^{\neg i},z_{t+1}^{\neg i} | y^i_{t+1} \}.
\label{eq:thm:sufficiency:occupancy:state:factor:9:0}
\end{align}
The application of Bayes' rule and the substitution of $y^i_{t+1}$ by $(y^i_t,a^{\neg i}_t,u^i_t, z^i_{t+1})$ in \eqref{eq:thm:sufficiency:occupancy:state:factor:9:0} yield
\begin{align}
s^i_{t+1}(x_{t+1}, (o_t^{\neg i},u_t^{\neg i},z_{t+1}^{\neg i}, o^i_{t+1})) &= \Pr\{ s_0,o_t,u_t,z_{t+1}, x_{t+1}, a^{\neg i}_{:t} \}/\Pr\{ y^i_{t+1} \}.
\label{eq:thm:sufficiency:occupancy:state:factor:9:1}
\end{align}
The expansion of \eqref{eq:thm:sufficiency:occupancy:state:factor:9:1} over all hidden states $x_t$ produces
\begin{align}
s^i_{t+1}(x_{t+1}, (o_t^{\neg i},u_t^{\neg i},z_{t+1}^{\neg i}, o^i_{t+1})) &=\textstyle \sum_{x_t\in \mathcal{X}} \Pr\{ s_0,x_t, o_t,u_t,z_{t+1},x_{t+1}, a^{\neg i}_{:t} \}/\Pr\{ y^i_{t+1} \}.
\label{eq:thm:sufficiency:occupancy:state:factor:9:2}
\end{align}
The expansion of the joint probability in the numerator of \eqref{eq:thm:sufficiency:occupancy:state:factor:9:2} as a product of conditional probabilities yields
\begin{align}
&s^i_{t+1}(x_{t+1}, (o_t^{\neg i},u_t^{\neg i},z_{t+1}^{\neg i}, o^i_{t+1})) =
\textstyle\sum_{x_t\in \mathcal{X}}
\Pr\{z_{t+1},x_{t+1} | s_0,x_t, o_t,u_t,a^{\neg i}_{:t} \}\nonumber \\
&\qquad \cdot \Pr\{ u^{\neg i}_t | s_0,x_t, o_t,a^{\neg i}_{:t},u^i_t \}
\cdot \Pr\{x_t, o^{\neg i}_t | s_0,o^{\neg i}_t, a^{\neg i}_{:t},u^i_t \}
\cdot \frac{\Pr\{ s_0,o^{\neg i}_t, a^{\neg i}_{:t},u^i_t \}}{ \Pr\{ y^i_{t+1} \} } .
\label{eq:thm:sufficiency:occupancy:state:factor:9:3}
\end{align}
The first probability in the numerator of \eqref{eq:thm:sufficiency:occupancy:state:factor:9:3} is independent of private plan-time history $y^i_t$, decision rule $a^{\neg i}_t$ at the end of time step $t$; it is simply the dynamics of $M$, \ie
\begin{align}
&s^i_{t+1}(x_{t+1}, (o_t^{\neg i},u_t^{\neg i},z_{t+1}^{\neg i}, o^i_{t+1})) =
\textstyle\sum_{x_t\in \mathcal{X}}
p^{u_tz_{t+1}}_{x_tx_{t+1}}\cdot \Pr\{ u^{\neg i}_t | s_0,x_t, o_t,a^{\neg i}_{:t},u^i_t \} \nonumber \\
&\qquad
\cdot \Pr\{x_t, o^{\neg i}_t | s_0,o^{\neg i}_t, a^{\neg i}_{:t},u^i_t \}
\cdot \frac{\Pr\{ y^i_t, a^{\neg i}_t,u^i_t \}}{ \Pr\{ y^i_{t+1} \} } .
\label{eq:thm:sufficiency:occupancy:state:factor:9:4}
\end{align}
The second probability in the numerator of \eqref{eq:thm:sufficiency:occupancy:state:factor:9:4} defines the probability of taking joint action $u_t^{\neg i}$ of agents $\neg i$ upon experiencing joint history $o_t^{\neg i}$, \ie
\begin{align}
&\overset{\eqref{eq:thm:sufficiency:occupancy:state:factor:9:0}}{=}
\textstyle\sum_{x_t\in \mathcal{X}}
p^{u_tz_{t+1}}_{x_tx_{t+1}}\cdot a^{\neg i}_t( u^{\neg i}_t | o^{\neg i}_t )
\cdot \Pr\{x_t, o^{\neg i}_t | s_0,o^{\neg i}_t, a^{\neg i}_{:t+1},u^i_t \}
\cdot \frac{\Pr\{ y^i_t,a^{\neg i}_t,u^i_t \}}{ \Pr\{ y^i_{t+1} \} } .
\label{eq:thm:sufficiency:occupancy:state:factor:9:5}
\end{align}
Moreover, the third probability in the numerator of \eqref{eq:thm:sufficiency:occupancy:state:factor:9:5} will be independent of private action $u^i_t$ as well as decision rule $a^{\neg i}_t$, since current decision rule alternatives do not affect current joint history. Hence, it defines the probability of the previous private occupancy state at hidden state $x_t$ and joint history $o^{\neg i}_t$, \ie
\begin{align}
&\overset{\eqref{eq:thm:sufficiency:occupancy:state:factor:9:0}}{=}
\sum_{x_t\in \mathcal{X}}
p^{u_tz_{t+1}}_{x_tx_{t+1}} \cdot a^{\neg i}_t( u^{\neg i}_t | o^{\neg i}_t )
\cdot s^i_t(x_t, o_t)
\cdot \frac{\Pr\{ y^i_t,a^{\neg i}_t,u^i_t \}}{ \Pr\{ y^i_{t+1} \} } .
\label{eq:thm:sufficiency:occupancy:state:factor:9:6}
\end{align}
The last probability in the numerator of \eqref{eq:thm:sufficiency:occupancy:state:factor:9:6} can be simplified by the denominator in \eqref{eq:thm:sufficiency:occupancy:state:factor:9:6}, \ie
\begin{align}
&\overset{\eqref{eq:thm:sufficiency:occupancy:state:factor:9:0}}{=}
\sum_{x_t\in \mathcal{X}}
p^{u_tz_{t+1}}_{x_tx_{t+1}} \cdot a^{\neg i}_t( u^{\neg i}_t | o^{\neg i}_t )
\cdot s^i_t(x_t, o_t)/\Pr\{z^i_{t+1}|y^i_t,a^{\neg i}_t,u^i_t \}.
\label{eq:thm:sufficiency:occupancy:state:factor:9:7}
\end{align}
From Lemma \ref{lem:sufficiency:occupancy:state:observation}, we know the denominator in \eqref{eq:thm:sufficiency:occupancy:state:factor:9:7} depends on private plan-time history $y^i_t$ only through corresponding private occupancy state $s^i_t$, \ie
\begin{align*}
&\overset{\eqref{eq:thm:sufficiency:occupancy:state:factor:9:0}}{=}
\sum_{x_t\in \mathcal{X}}
p^{u_tz_{t+1}}_{x_tx_{t+1}} \cdot a^{\neg i}_t( u^{\neg i}_t | o^{\neg i}_t )
\cdot s^i_t(x_t, o_t)/\Pr\{z^i_{t+1}|s^i_t,a^{\neg i}_t,u^i_t \}.
\end{align*}
where $\omega^i_{M''}(z^i_{t+1} | s^i_t,u^i_t) \doteq \Pr\{ z^i_{t+1} | s^i_t,a^{\neg i}_t,u^i_t \}$ serves as a normalizing constant. The important feature of the update rule of private occupancy states is that the calculation of the private occupancy states after time step $t+1$ requires only $s^i_t$, the private occupancy state at time step $t$; thus, $s^i_{t+1}$ summarizes all the information gained before time step $t$ and represents a sufficient statistic of the complete data $y^i_{t+1}$ available to agent $i$ for predicting the joint history.
\end{proof}

Lemma \ref{lem:sufficiency:private:occupancy:state:markov} essentially states that private occupancy states describe a Markov process.

\begin{restatable}{lem}{lemsufficiencyoccupancystatereward}
\label{lem:sufficiency:occupancy:state:reward}
For each agent $i$, the private occupancy state at time step $t$, defined as $s^i_t \doteq \Pr\{X_t,O_t|y^i_t\}$, is a sufficient statistic of the private plan-time history $y_t^i$ for an accurate prediction of the expected reward that will be received immediately after executing the private action $u^i_t$, \ie $r^i_{M'}(y_t^i, u_t^i)=r^i_{M''}(s_t^i, u_t^i)$.
\end{restatable}

\begin{proof}
The proof starts with the definition of the expected reward conditional on private plan-time history $y^i_t$, the joint decision rule $a^{\neg i}_t$ of agents $\neg i$, and private action $u^i_t$ of agent $i$, \ie
\begin{align}
r^i_{M'}(y_t^i, u_t^i)&\doteq \mathbb{E}\{ r^i_{X_t,U_t} | y^i_t, a^{\neg i}_t, u^i_t\}\nonumber\\
&= \textstyle \sum_{x_t\in \mathcal{X}}\sum_{u^{\neg i}_t\in \mathcal{U}^{\neg i}} r^i_{X_t,U_t} \sum_{o^{\neg i}_t\in \mathcal{O}^{\neg i}_t} \Pr\{ x_t, o^{\neg i}_t, u^{\neg i}_t | y^i_t, a^{\neg i}_t, u^i_t \}.
\label{eq:thm:sufficiency:occupancy:state:factor:1}
\end{align}
The expansion of the joint probability in \eqref{eq:thm:sufficiency:occupancy:state:factor:1} as a product of conditional probabilities produces
\begin{align}
&\overset{\eqref{eq:thm:sufficiency:occupancy:state:factor:1}}{=}
\sum_{x_t\in \mathcal{X}}\sum_{u^{\neg i}_t\in \mathcal{U}^{\neg i}}
r^i_{x_tu_t} \sum_{o^{\neg i}_t\in \mathcal{O}^{\neg i}_t}
\Pr\{ u^{\neg i}_t | o^{\neg i}_t,y^i_t, a^{\neg i}_t, u^i_t \}
\cdot \Pr\{x_t, o^{\neg i}_t | y^i_t, a^{\neg i}_t, u^i_t \}.
\label{eq:thm:sufficiency:occupancy:state:factor:2}
\end{align}
The first probability in \eqref{eq:thm:sufficiency:occupancy:state:factor:2} describes the probability of taking joint action $u^{\neg i}_t$ given joint history $o^{\neg i}_t$ and joint decision rule $a^{\neg i}_t$ of agents $\neg i$, \ie
\begin{align}
&\overset{\eqref{eq:thm:sufficiency:occupancy:state:factor:1}}{=}
\sum_{x_t\in \mathcal{X}}\sum_{u^{\neg i}_t\in \mathcal{U}^{\neg i}}
r^i_{x_tu_t}
\sum_{o^{\neg i}_t\in \mathcal{O}^{\neg i}_t} a^{\neg i}_t(u^{\neg i}_t|o^{\neg i}_t )
\cdot \Pr\{x_t, o^{\neg i}_t | y^i_t, a^{\neg i}_t, u^i_t \}.
\label{eq:thm:sufficiency:occupancy:state:factor:3}
\end{align}
The second probability in \eqref{eq:thm:sufficiency:occupancy:state:factor:3} describes the probability of the hidden state being $x_t$ and the other agents having joint history $o_t^{\not i}$ at the current private occupancy state, \ie
\begin{align*}
&\overset{\eqref{eq:thm:sufficiency:occupancy:state:factor:1}}{=}
\sum_{x_t\in \mathcal{X}}\sum_{u^{\neg i}_t\in \mathcal{U}^{\neg i}}
r^i_{x_tu_t}
\sum_{o^{\neg i}_t\in \mathcal{O}^{\neg i}_t} a^{\neg i}_t(u^{\neg i}_t|o^{\neg i}_t )
\cdot s^i_t(x_t, o_t).
\end{align*}
Which ends the proof.
\end{proof}

The three lemmas, namely Lemma \ref{lem:sufficiency:occupancy:state:observation}, Lemma \ref{lem:sufficiency:private:occupancy:state:markov}, and Lemma \ref{lem:sufficiency:occupancy:state:reward}, establish that private occupancy states are enough to serve as statistics for private plan-time histories. These lemmas also describe a Markov decision process with discrete time and continuous states. The best-response policy $a^i_{0:}$ for the private occupancy-state Markov decision process $M''(a^{\neg i}_{0:})$ can also be used as the best-response policy for the private plan-time Markov game $M'(a^{\neg i}_{0:})$ (and $M(a^{\neg i}_{0:})$):
\begin{align}
&a^i_t(\varphi^i_{M'}(y^i_t)) \\
&\textstyle \in \argmax_{u^i_t\in \mathcal{U}^i} q^i_{M(a^{\neg i}_{0:}),\gamma,t}(o^i_t, u^i_t) \nonumber\\
&\textstyle\in \argmax_{u^i_t\in \mathcal{U}^i} q^i_{M'(a^{\neg i}_{t:}),\gamma,t}(y^i_t, u^i_t) \nonumber\\
&\textstyle\in \argmax_{u^i_t\in \mathcal{U}^i} r^i_{M'}(y^i_t, u^i_t) + \gamma \sum_{z^i_{t+1}\in \mathcal{Z}^i} \omega^i_{M'}(z^i_{t+1}|y^i_t,u^i_t) \cdot \upsilon^i_{M'(a^{\neg i}_{t+1:}),\gamma,t+1}(y^i_t,u^i_t,z^i_{t+1}) \label{eqn:thm:private:sufficiency:occupancy:state:1}\\
&\textstyle\in \argmax_{u^i_t\in \mathcal{U}^i} r^i_{M''}(s^i_t, u^i_t) + \gamma \sum_{z^i_{t+1}\in \mathcal{Z}^i} \omega^i_{M''}(z^i_{t+1}|s^i_t,u^i_t) \cdot \upsilon^i_{M''(a^{\neg i}_{t+1:}),\gamma,t+1}(s^i_t,u^i_t,z^i_{t+1}) \label{eqn:thm:private:sufficiency:occupancy:state:2}\\
&\textstyle\in \argmax_{u^i_t\in \mathcal{U}^i} q^i_{M''(a^{\neg i}_{t:}),\gamma,t}(s^i_t, u^i_t). \nonumber
\end{align}
By applying Lemmas \ref{lem:sufficiency:occupancy:state:observation}, \ref{lem:sufficiency:private:occupancy:state:markov}, and \ref{lem:sufficiency:occupancy:state:reward} to Equation \eqref{eqn:thm:private:sufficiency:occupancy:state:1}, we can arrive at Equation \eqref{eqn:thm:private:sufficiency:occupancy:state:2}. Therefore, it can be concluded that private occupancy states act as sufficient statistics of private histories and private plan-time histories, which enables to optimally solving slave games according to Theorem \ref{thm:private:sufficiency:occupancy:state}.

\subsection{Occupancy States As Mixtures of Private Occupancy States}
The subsection explores the relationship between sufficient statistics for a master game and its corresponding slave games. It demonstrates that occupancy states can be expressed using various basis sets. While the typical representation of occupancy states uses the canonical basis, a finite collection of unit vectors, one can also rewrite occupancy states based on the private occupancy states of a fixed agent. These basis changes are crucial when identifying the convexity properties of state- and action-value functions for the master game As we will illustrate, certain state- and action-value functions only reveal their structure when expressed under a specific basis set.

\begin{restatable}{lem}{lemrelationbetweenoccupancystateandprivateones}
\label{lem:relation:between:occupancy:state:and:private:ones}
Consider a joint policy $a_{0:}$ and initial belief state $s_0$. For every agent $i$ and $t$-step private history $o_t^i$, let $\lambda^{i,o^i_t} \doteq \Pr\{o^i_t | y_t \doteq (s_0, a_{:t-1},w_{1:t}) \} $ be the probability of $o_t^i$ given $y_t$ and $s^{i,o^i_t}_t \doteq \Pr\{X_t,O_t | y^i_t \doteq (s_0, o^i_t, a_{:t-1}) \}$ be the private occupancy state given $y^i_t$. For any $t$-step occupancy state $s_t \doteq \Pr\{X_t,O_t | y_t \doteq (s_0, a_{:t-1},w_{1:t})\}$ there exists $(\lambda^{i,o^i_t}, s^{i,o^i_t}_t )_{o^i_t\in \mathcal{O}^i_t}$ such that $s_t$ can be expressed as a mixture of private occupancy states of any agent $i$, \ie
\begin{align}
s_t &= \textstyle \sum_{o_t^i\in \mathcal{O}^i_t}~ \lambda^{i,o_t^i} \cdot s^{i,o_t^i}_t \quad\text{such that }\quad \sum_{o_t^i\in \mathcal{O}^i_t}~ \lambda^{i,o_t^i} = 1 \text{ and } \lambda^{i,o_t^i}\geq 0.
\label{eq:lem:relation:between:occupancy:state:and:private:ones:0}
\end{align}
\end{restatable}

\begin{proof}
The proof starts with the definition of occupancy states, \ie for any time step $t$ and plan-time history $y_t \doteq (s_0,a_{:t-1}, w_{1:t})$, we have:
\begin{align}
s_t(x_t, o_t ) &\doteq \Pr\{x_t, o_t | y_t\},&\forall x_t\in \mathcal{X}, o_t\in \mathcal{O}_t.
\label{eq:lem:relation:between:occupancy:state:and:private:ones:1}
\end{align}
The expansion of \eqref{eq:lem:relation:between:occupancy:state:and:private:ones:1} with histories $o_t^i\in \mathcal{O}^i_t$ of agent $i$ leads to
\begin{align}
s_t(x_t, o_t )&= {\textstyle \sum_{o_t^i\in \mathcal{O}^i_t}}~\Pr\{x_t, o_t, o_t^i| y_t\}.
\label{eq:lem:relation:between:occupancy:state:and:private:ones:2}
\end{align}
The expansion of \eqref{eq:lem:relation:between:occupancy:state:and:private:ones:2} as the product of conditional probability distributions gives
\begin{align}
s_t(x_t, o_t )&= {\textstyle \sum_{o_t^i\in \mathcal{O}^i_t}}~\Pr\{x_t, o_t, | o_t^i, y_t\} \cdot \Pr\{ o_t^i | y_t\}.
\label{eq:lem:relation:between:occupancy:state:and:private:ones:3}
\end{align}
The first factor in \eqref{eq:lem:relation:between:occupancy:state:and:private:ones:3} depends on $a^i_{:t-1}$ only through $o_t^i$. It simply corresponds to private occupancy state given by $s^{i,o_t^i}_t(x_t, o_t) \doteq \Pr\{x_t, o_t | s_0,a^{\neg i}_{:t-1},o_t^i \}$, \ie
\begin{align}
s_t(x_t, o_t )&= {\textstyle \sum_{o_t^i\in \mathcal{O}^i_t}}~s^{i,o_t^i}_t(x_t, o_t) \cdot \Pr\{ o_t^i| y_t\}.
\label{eq:lem:relation:between:occupancy:state:and:private:ones:4}
\end{align}
The second factor in \eqref{eq:lem:relation:between:occupancy:state:and:private:ones:4} is simply a non-negative coefficient, denoted $\lambda^{i,o_t^i}$, defining the probability of history $o_t^i$ when agents follow joint policy $a_{:t-1}$ starting in an initial belief state $s_0$, \ie
\begin{align*}
s_t(x_t, o_t )&= {\textstyle \sum_{o_t^i\in \mathcal{O}^i_t}}~\lambda^{i,o_t^i} \cdot s^{i,o_t^i}_t(x_t, o_t).
\end{align*}
It is worth noticing that $\sum_{o_t^i\in \mathcal{O}^i_t}~\lambda^{i,o_t^i} = \sum_{o_t^i\in \mathcal{O}^i_t}~ \Pr\{ o_t^i| y_t\} = 1$. Which ends the proof.
\end{proof}

Lemma \ref{lem:relation:between:occupancy:state:and:private:ones} provides an alternative perspective on the nature of occupancy states. An occupancy state is a mathematical entity that captures the state of an occupancy-state Markov game $M'$ from various viewpoints.

\begin{figure}[!htbp]
\centering
\begin{tikzpicture}[scale=1]
\draw[fill=yellow!20, opacity=.5, rounded corners, draw=blue] (-5.5,-4.75) rectangle (1,.5);
\draw[fill=yellow!20, opacity=.5, rounded corners, draw=blue] (1.75,-4.75) rectangle (9.25,.5);
\node[scale=1] at (-2.2,1) {$s^{\mathcal{B}^0}_1$};
\node[scale=1] at (5.5,1) {$s^{\mathcal{B}^1}_1$};
[sibling distance=2cm,-, thick]
\node[scale=.8] at (3.5,0.35) {$.5$};
\node[scale=.85] at (3.5,0) {$\textcolor{green!60!black}{\emptyset}\textcolor{red!60!black}{\emptyset}$}
child {node[scale=.85] {$\textcolor{green!60!black}{\boldsymbol{u_\textsc{l}}},\textcolor{red!60!black}{\boldsymbol{u_\textsc{ol}}}$}
child {node[scale=.75] {$\textcolor{green!60!black}{\boldsymbol{z_\textsc{hl}}},\textcolor{red!60!black}{\boldsymbol{z_\textsc{hl}}}$}
[sibling distance=.5cm,-]
child {node[scale=.65] {.25} edge from parent node[left, scale=.75] {$x_\textsc{tl}$}}
child {node[scale=.65] {.25} edge from parent node[right, scale=.75] {$x_\textsc{tr}$}}}
child {node[scale=.75] { $\textcolor{green!60!black}{\boldsymbol{z_\textsc{hl}}},\textcolor{red!60!black}{\boldsymbol{z_\textsc{hr}}}$}
[sibling distance=.5cm,-]
child {node[scale=.65] {.25} edge from parent node[left, scale=.75] {$x_\textsc{tl}$}}
child {node[scale=.65] {.25} edge from parent node[right, scale=.75] {$x_\textsc{tr}$}}}};
\node[scale=.8] at (7.5,0.35) {$.5$};
\node[scale=.85] at (7.5,0) {$\textcolor{green!60!black}{\emptyset}\textcolor{red!60!black}{\emptyset}$}
child {node[scale=.85] {$\textcolor{green!60!black}{\boldsymbol{u_\textsc{l}}},\textcolor{red!60!black}{\boldsymbol{u_\textsc{ol}}}$}
child {node[scale=.75] {$\textcolor{green!60!black}{\boldsymbol{z_\textsc{hr}}},\textcolor{red!60!black}{\boldsymbol{z_\textsc{hl}}}$}
[sibling distance=.5cm,-]
child {node[scale=.65] {.25} edge from parent node[left, scale=.75] {$x_\textsc{tl}$}}
child {node[scale=.65] {.25} edge from parent node[right, scale=.75] {$x_\textsc{tr}$}}}
child {node[scale=.75] { $\textcolor{green!60!black}{\boldsymbol{z_\textsc{hr}}},\textcolor{red!60!black}{\boldsymbol{z_\textsc{hr}}}$}
[sibling distance=.5cm,-]
child {node[scale=.65] {.25} edge from parent node[left, scale=.75] {$x_\textsc{tl}$}}
child {node[scale=.65] {.25} edge from parent node[right, scale=.75] {$x_\textsc{tr}$}}}};
\node[scale=.85] (n1) at (-2.25,0) {$\textcolor{green!60!black}{\emptyset}\textcolor{red!60!black}{\emptyset}$}
child {node[scale=.85] (n2) {$\textcolor{green!60!black}{\boldsymbol{u_\textsc{l}}},\textcolor{red!60!black}{\boldsymbol{u_\textsc{ol}}}$}
child {node[scale=.75] {$\textcolor{green!60!black}{\boldsymbol{z_\textsc{hl}}},\textcolor{red!60!black}{\boldsymbol{z_\textsc{hl}}}$}
[sibling distance=.5cm,-]
child {node[scale=.65] {.125} edge from parent node[left, scale=.75] {$x_\textsc{tl}$}}
child {node[scale=.65] {.125} edge from parent node[right, scale=.75] {$x_\textsc{tr}$}}}
child {node[scale=.75] { $\textcolor{green!60!black}{\boldsymbol{z_\textsc{hr}}},\textcolor{red!60!black}{\boldsymbol{z_\textsc{hr}}}$}
[sibling distance=.5cm,-]
child {node[scale=.65] {.125} edge from parent node[left, scale=.75] {$x_\textsc{tl}$}}
child {node[scale=.65] {.125} edge from parent node[right, scale=.75] {$x_\textsc{tr}$}}}
child {node[scale=.75] { $\textcolor{green!60!black}{\boldsymbol{z_\textsc{hr}}},\textcolor{red!60!black}{\boldsymbol{z_\textsc{hl}}}$}
[sibling distance=.5cm,-]
child {node[scale=.65] {.125} edge from parent node[left, scale=.75] {$x_\textsc{tl}$}}
child {node[scale=.65] {.125} edge from parent node[right, scale=.75] {$x_\textsc{tr}$}}}
child {node[scale=.75] { $\textcolor{green!60!black}{\boldsymbol{z_\textsc{hl}}},\textcolor{red!60!black}{\boldsymbol{z_\textsc{hr}}}$}
[sibling distance=.5cm,-]
child {node[scale=.65] {.125} edge from parent node[left, scale=.75] {$x_\textsc{tl}$}}
child {node[scale=.65] {.125} edge from parent node[right, scale=.75] {$x_\textsc{tr}$}}}};
\end{tikzpicture}
\caption{Tree representations of occupancy state $s_1$ in both $\mathcal{B}^0$ and $\mathcal{B}^1$ for the case with no public observations. }
\label{fig:occupancy:state:basis:change}
\end{figure}
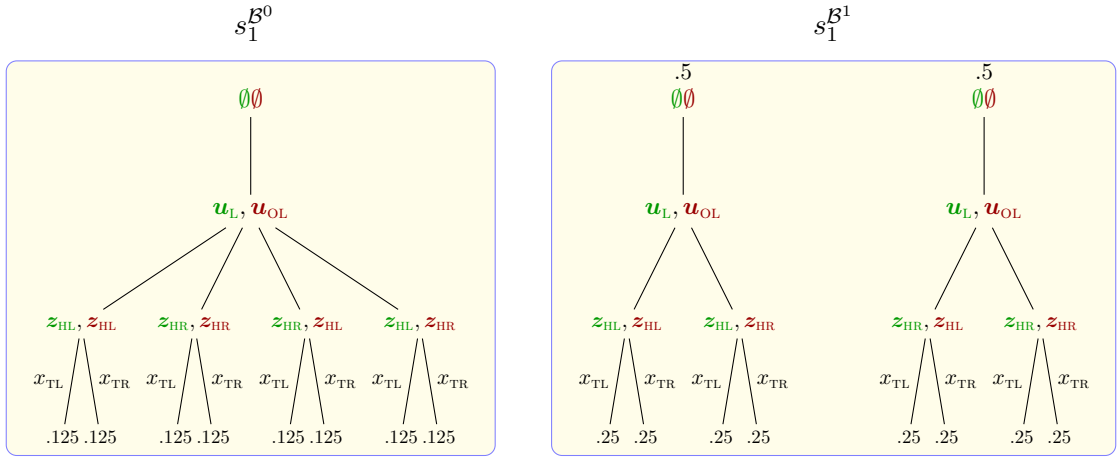

So far, we have described occupancy states from the perspective of the central planner, using the standard coordinate system, \ie
\begin{align}
s_t &= \sum_{(x_t,o_t)\in \mathcal{X}\times \mathcal{O}_t}~\lambda^{(x_t,o_t)} \cdot e^{(x_t,o_t)}\quad\text{such that }\quad \sum_{(x_t,o_t)\in \mathcal{X}\times \mathcal{O}_t}~ \lambda^{(x_t,o_t)} = 1 \text{ and } \lambda^{(x_t,o_t)}\geq 0
\label{eq:lem:relation:between:occupancy:state:and:private:ones:5}
\end{align}
where $e^{(x_t,o_t)}$ is a one-hot vector associated with pair $(x_t,o_t)\in \mathcal{X}\times \mathcal{O}_t$. In the canonical basis $\mathcal{B}^0$, which consists of all elements $e^{(x_t,o_t)}$ where $(x_t,o_t)$ is in $\mathcal{X} \times \mathcal{O}_t$, the occupancy state $s^{\mathcal{B}^0}_t$ is represented as $(\lambda^{(x_t,o_t)})_{(x_t,o_t) \in \mathcal{X} \times \mathcal{O}_t}$. This representation is objective because neither agents nor the central planner can access general information on states and joint histories denoted as $(x_t,o_t)$. This representation is useful in multi-agent systems without self-interested agents, such as \emph{dec}-POMDPs. However, more informative coordinate systems are needed in self-interested games like \emph{zs}-POSGs or \emph{st}-POSGs. Therefore, private occupancy states $s^{i,o_t^i}_t$ from agent $i$'s perspective are used to describe occupancy states in $M''$, as shown in Equation \eqref{eq:lem:relation:between:occupancy:state:and:private:ones:0}. These private occupancy states convey information on the occupancy state $s_t$ in $M''$ given that agent $i$ experienced private history $o^i$. Interestingly, using either the basis of canonical occupancy states in Equation \eqref{eq:lem:relation:between:occupancy:state:and:private:ones:5} or the basis of private occupancy states for agent $i$ in $\mathcal{B}^i = \{s^{i,o_t^i} \colon o_t^i\in \mathcal{O}_t^i\}$ results in the same information on the occupancy state of $M'$, but in different coordinate systems. The vector $s^{\mathcal{B}^i}_t$ expresses the occupancy state using the basis set $\mathcal{B}^i$, specific to agent $i$. Figure \ref{fig:occupancy:state:basis:change} illustrates two different representations of the same occupancy state, where $s^{\mathcal{B}^0}_1= (.125~.125~.125~.125~.125~.125~.125~.125~)$ and $s^{\mathcal{B}^1}_1= (.5\quad .5)$. Despite the significant differences between vectors $s^{\mathcal{B}^0}_t$ and $s^{\mathcal{B}^i}_t$, they convey the same occupancy state using different languages or basis sets.

\section{Convexity of Optimal Value Functions of Master Games}
\label{sec:main:results}
In this section, we are ready to present the main findings of our research. Our focus initially is on the uniform continuity features of state-value functions in slave games. Secondly, we will demonstrate the correlation between solutions for a master game and its corresponding slave games. Lastly, we will provide evidence for the three convexity properties for state-value functions in specific and significant master games.

\subsection{Piecewise Linearity and Convexity}
This subsection demonstrates the piecewise linear and convex nature of optimal state-value functions over private occupancy states in private occupancy-state Markov games. It is worth noting that the sufficiency of private occupancy states concerning private plan-time history for optimally solving private occupancy-state Markov games facilitates the ability to recover the solution of the original slave game.

\begin{restatable}{thm}{corpwlcoccslaveproblem}
\label{cor:pwlc:occ:slave:problem}
Let $M$ be a master game. Let $M''(a^{\neg i}_{0:})$ and $M'(a^{\neg i}_{0:})$ be slave games \wrt $M$ and policy $a^{\neg i}_{0:}\in \mathcal{A}_{0:}^{\neg i}$ for agents $\neg i$. For any arbitrary time step $t$, the slave game $M''(a^{\neg i}_{0:})$ has a piecewise linear and convex (PWLC) optimal state-value function of private occupancy states, \ie $\upsilon^i_{M''(a^{\neg i}_{t:}),\gamma,t}\colon \mathcal{S}^i_t\to \mathbb{R}$. Consequently, at each time step $t$ there exists a finite collection of linear functions of private occupancy states, denoted as $\mathcal{V}_{M''(a^{\neg i}_{t:}), \gamma, t}$, \st for any private occupancy state $s^i_t$, it follows:
\begin{align*}
\upsilon^i_{M''(a^{\neg i}_{t:}),\gamma,t}(s^i_t) &= \textstyle \max_{\upsilon_{M'',\gamma,t}^{i,a_{t:}}\in\mathcal{V}_{M''(a^{\neg i}_{t:}), \gamma, t}}~\upsilon_{M'',\gamma,t}^{i,a_{t:}}(s^i_t).
\end{align*}
\end{restatable}

\begin{proof}
To prove piecewise linearity and convexity, it is important to note that all best-response policies result in a unique state-value function. This means mixed best-response policies are just as good as pure policies in their support. Therefore, when searching for the optimal value function for a slave game, one can focus solely on pure best-response policies. We define the optimal value for any given time step $t$ and private occupancy state $s^i_t$ to derive the property.
\begin{align*}
\upsilon^i_{M''(a^{\neg i}_{t:}),\gamma,t}(s^i_t) &\doteq\textstyle \max_{a^i_{t:} \in \mathcal{A}^i_{t:}}~\mathbb{E}\{\upsilon^{i,a_{t:}}_{M,\gamma,t}(X_t,O_t)|s^i_t\}\\
&\overset{\text{pure policies}}{=} \textstyle\max_{a^i_{t:} \in \mathcal{A}^{i,\text{pure}}_{t:}}~ \mathbb{E}\{\upsilon^{i,a_{t:}}_{M,\gamma,t}(X_t,O_t)|s^i_t\}\\
&\overset{\text{Lemma \ref{lem:linearity:joint:policy}}}{=} \textstyle\max_{a^i_{t:} \in \mathcal{A}^{i,\text{pure}}_{t:}}~\upsilon^{i,a_{t:}}_{M'',\gamma,t}(s^i_t).
\end{align*}
The point-wise maximum of linear functions over private occupancy states is a piecewise linear and convex function of private occupancy states.
\end{proof}

Theorem \ref{cor:pwlc:occ:slave:problem} establishes that for any slave game, the optimal state-value functions are piecewise linear and convex over private occupancy states. In the case of a two-agent one-sided POSG, where agent 2 perceives both the underlying state of the game as well as the actions and perceptions of agent 1, the optimal state-value functions $\langle \upsilon^1_{M''(s_0,a^2_{0:}),\gamma,0}, \upsilon^1_{M''(a^2_{1:}),\gamma,1},\ldots, \upsilon^1_{M''(a^2_{\ell-1:}),\gamma,\ell-1}\rangle$ of the slave game are tabular functions that map underlying states to real values.

\begin{restatable}{cor}{corpwlcoccslaveactionproblem}
For any arbitrary time step $t$, the slave game $M''(a^{\neg i}_{0:})$ has a piecewise linear and convex (PWLC) optimal action-value function of private occupancy states $q^i_{M''(a^{\neg i}_{t:}),\gamma,t}\colon \mathcal{S}^i_t\times \mathcal{U}^i\to \mathbb{R}$. Consequently, at each time step $t$ there exists a finite collection of linear functions of private occupancy states, denoted as $\mathcal{Q}_{M''(a^{\neg i}_{t:}),\gamma,t}$, \st for any private occupancy state $s^i_t$, it follows:
\begin{align*}
q^i_{M''(a^{\neg i}_{t:}),\gamma,t}(s^i_t, u^i_t) &= \textstyle \max_{q_{M'',\gamma,t}^{i,a_{t:}}\in\mathcal{Q}_{M''(a^{\neg i}_{t:}),\gamma,t} }~q_{M'',\gamma,t}^{i,a_{t:}}(s^i_t, u^i_t).
\end{align*}
\end{restatable}

Corollary \ref{cor:pwlc:occ:slave:occupancy:problem} shows optimal state-value functions are piecewise linear and convex over occupancy states in the standard basis for any slave game.

\begin{restatable}{cor}{corpwlcoccslaveoccupancyproblem}
\label{cor:pwlc:occ:slave:occupancy:problem}
For any arbitrary time step $t$, the slave game $M''(a^{\neg i}_{0:})$ has a piecewise linear and convex (PWLC) optimal state-value function of occupancy states expressed in the standard basis, \ie $\upsilon^i_{M''(a^{\neg i}_{t:}),\gamma,t}\colon \mathcal{S}_t^{\mathcal{B}^0} \to \mathbb{R}$. Consequently, at each time step $t$ there exists a finite collection of linear functions of occupancy states, denoted as $\mathcal{V}_{M''(a^{\neg i}_{t:}),\gamma,t}$, \st for any occupancy state $s^{\mathcal{B}^0}_t$, it follows:
\begin{align*}
\upsilon^i_{M''(a^{\neg i}_{t:}),\gamma,t}(s^{\mathcal{B}^0}_t) &= \textstyle \max_{\upsilon_{M'',\gamma,t}^{i,a_{t:}}\in\mathcal{V}_{M''(a^{\neg i}_{t:}),\gamma,t}}~\upsilon_{M'',\gamma,t}^{i,a_{t:}}(s^{\mathcal{B}^0}_t).
\end{align*}
\end{restatable}

\subsection{Optimal Value Functions Are Convex}
This subsection presents our key finding. However, we must first demonstrate that the optimal state-value function of the slave game $M''(a^{\neg i}_{0:})$ (as well as $M'(a^{\neg i}_{0:})$ and $M(a^{\neg i}_{0:})$) can be expressed as a linear function of occupancy states from the perspective of agent $i$, using the basis $\mathcal{B}^i$. This proof is based on the premise that private occupancy states are sufficient for determining the optimal state-value function of slave game $M''(a^{\neg i}_{0:})$ and the representation of occupancy states as a combination of private occupancy states. Before proceeding, we would like to establish some fundamental preliminary results.

\begin{restatable}{lem}{lemrewardrelationbetweenoccupancystateandprivateones}
\label{lem:reward:relation:between:occupancy:state:and:private:ones}
For each agent $i$, the immediate expected reward function over occupancy states, denoted as $r^i_{M''} \colon \mathcal{S}^{\mathcal{B}^i} \times \mathcal{A}^i \to \mathbb{R}$, is linear over occupancy states expressed in $\mathcal{B}^i$ and decision rules of agent $i$---\ie for occupancy state $s^{\mathcal{B}^i}_t \doteq \sum_{o_t^i\in \mathcal{O}^i_t} \lambda^{i,o_t^i} \cdot s_t^{i,o_t^i}~ \text{ such that }\sum_{o_t^i\in \mathcal{O}^i_t} \lambda^{i,o_t^i} =1 ~\text{ and } \lambda^{i,o_t^i} \geq 0$, and decision rule $a^i_t$,
\begin{align*}
r^i_{M''}(s^{\mathcal{B}^i}_t, a^i_t) &=
\sum_{o_t^i\in \mathcal{O}^i_t} \lambda^{i,o_t^i} \sum_{u_t^i\in \mathcal{U}^i} a^i_t(u^i|o^i) \cdot r^i_{M''}(s_t^{i,o_t^i},u_t^i).
\end{align*}
\end{restatable}

\begin{proof}
The proof starts with the definition of the immediate expected reward at occupancy state $s^{\mathcal{B}^i}_t$, and decision rule $a^i_t$,
\begin{align*}
r^i_{M''}(s^{\mathcal{B}^i}_t, a^i_t) &\doteq \mathbb{E}\{ r^i_{X_t,U_t} | s^{\mathcal{B}^i}_t, a_t\}.
\end{align*}
The expansion of the expectation over all private occupancy states, hidden states, and private actions gives the following expression
\begin{align}
r^i_{M''}(s^{\mathcal{B}^i}_t, a^i_t) &= \sum_{o_t^i\in \mathcal{O}^i_t}\lambda^{i,o_t^i} \sum_{x_t\in \mathcal{X}}\sum_{o^{\neg i}_t\in \mathcal{O}^{\neg i}_t}s^{i,o_t^i}_t(x_t,o_t)\sum_{u_t\in \mathcal{U}} a^i_t(u^i_t|o^i)\cdot a^{\neg i}_t(u^{\neg i}_t|o^{\neg i}_t)\cdot r^i_{x_t,u_t}.
\label{eqn:lem:reward:relation:between:occupancy:state:and:private:ones:0}
\end{align}
Given that occupancy states are mixtures of private occupancy states, \ie
\begin{align*}
s^{\mathcal{B}^i}_t &= \sum_{o_t^i\in \mathcal{O}^i_t}~ \lambda^{i,o_t^i} \cdot s^{i,o_t^i}_t \quad \text{ such that }\sum_{o_t^i\in \mathcal{O}^i_t}~ \lambda^{i,o_t^i} =1 \quad \lambda^{i,o_t^i} \geq 0.
\end{align*}
Re-arranging terms leads to
\begin{align}
&\overset{\eqref{eqn:lem:reward:relation:between:occupancy:state:and:private:ones:0}}{=}
\sum_{o_t^i\in \mathcal{O}^i_t} \lambda^{i,o_t^i} \sum_{u^i_t\in \mathcal{U}^i} a^i_t(u^i_t|o_t^i)\left(\sum_{x_t\in \mathcal{X}}\sum_{o^{\neg i}_t\in \mathcal{O}^{\neg i}_t}s^{i,o_t^i}_t(x_t,o_t)\sum_{u^{\neg i}_t\in \mathcal{U}^{\neg i}} a^{\neg i}_t(u^{\neg i}_t|o^{\neg i}_t)\cdot r^i_{x_t,u_t}\right).
\label{eqn:lem:reward:relation:between:occupancy:state:and:private:ones:1}
\end{align}
The injection of $r^i_{M''}(s^{i,o_t^i}_t,u^i_t)$ into the expression in the parentheses of \eqref{eqn:lem:reward:relation:between:occupancy:state:and:private:ones:1} yields
\begin{align*}
r^i_{M''}(s^{\mathcal{B}^i}_t, a^i_t) &=
\sum_{o_t^i\in \mathcal{O}^i_t} \lambda^{i,o_t^i} \sum_{u^i_t\in \mathcal{U}^i} a^i_t(u^i_t|o_t^i) \cdot r^i_{M''}(s^{i,o_t^i}_t,u^i_t),
\end{align*}
which ends the proof.
\end{proof}

Lemma \ref{lem:reward:relation:between:occupancy:state:and:private:ones} demonstrates a property that is easily proven for occupancy states expressed in the standard basis. This lemma confirms that the linearity property of the immediate reward remains valid regardless of the agent viewpoint, as long as occupancy states are expressed on the corresponding basis.

\begin{restatable}{lem}{lemtransitionrelationbetweenoccupancystateandprivateones}
\label{lem:transition:relation:between:occupancy:state:and:private:ones}
For each agent $i$, the transition rule over occupancy states expressed in $\mathcal{B}^i$, denoted as $\rho^i_{M'} \colon \mathcal{S}^{\mathcal{B}^i} \times \mathcal{A}^i \to \mathcal{S}^{\mathcal{B}^i}$, is linear over occupancy states expressed in $\mathcal{B}^i$ and decision rules of agent $i$---\ie for occupancy state $s^{\mathcal{B}^i}_t \doteq \sum_{o_t^i\in \mathcal{O}^i_t} \lambda^{i,o_t^i} \cdot s_t^{i,o_t^i}~ \text{ such that }\sum_{o_t^i\in \mathcal{O}^i_t} \lambda^{i,o_t^i} =1 ~\text{ and } \lambda^{i,o_t^i} \geq 0$, and decision rule $a^i_t$,
\begin{align*}
\rho^i_{M''}(s^{\mathcal{B}^i}_t, a^i_t) &= \sum_{o^i_t\in \mathcal{O}_t^i} \lambda^{i,o^i_t} \sum_{u^i_t\in \mathcal{U}^i} a^i_t(u_t^i|o^i_t) \sum_{z^i_{t+1}\in \mathcal{Z}^i} p^i_{M''}(s^{i,o^i_t}_t,u^i_t,s^{i,(o^i_t,u^i_t,z^i_{t+1})}_{t+1})\cdot s^{i,(o^i_t,u^i_t,z^i_{t+1})}_{t+1}.
\end{align*}
\end{restatable}

\begin{proof}
The proof starts with the definition of the next-step occupancy state $s^{\mathcal{B}^i}_{t+1} \doteq \rho^i_{M''}(s^{\mathcal{B}^i}_t, a^i_t)$ upon taking decision rule $a^i_t$ in occupancy state $s^{\mathcal{B}^i}_t$, \ie
\begin{align}
s^{\mathcal{B}^i}_{t+1}( s_{t+1}^{i,(o^i_t,u_t^i,z_{t+1}^i)}) &\doteq \Pr\{o^i_t,u^i_t,z^i_{t+1} | s^{\mathcal{B}^i}_t, a^i_t\},
\label{eq:lem:transition:relation:between:occupancy:state:and:private:ones}
\end{align}
The expansion of \eqref{eq:lem:transition:relation:between:occupancy:state:and:private:ones} over all hidden states $x_t,x_{t+1}$, private action $u^{\neg i}_t$, joint history $o^{\neg i}_t$, and joint observation $z^{\neg i}_{t+1}$ yields,
\begin{align*}
s^{\mathcal{B}^i}_{t+1}( s_{t+1}^{i,(o^i_t,u_t^i,z_{t+1}^i)}) &=
\sum_{x_t,x_{t+1}\in \mathcal{X}}
\sum_{u^{\neg i}_t\in \mathcal{U}^{\neg i}}
\sum_{o^{\neg i}_t\in \mathcal{O}^{\neg i}_t}
\sum_{z^{\neg i}_{t+1}\in \mathcal{Z}^{\neg i}}
\Pr\{x_t,u_t,o_t,x_{t+1},z_{t+1} | s^{\mathcal{B}^i}_t, a^i_t\}.
\end{align*}
The expansion of the joint probability into the product of conditional probabilities gives
\begin{align*}
&\overset{\text{\eqref{eq:lem:transition:relation:between:occupancy:state:and:private:ones}}}{=}
\sum_{x_t,x_{t+1}\in \mathcal{X}}
\sum_{u^{\neg i}_t\in \mathcal{U}^{\neg i}}
\sum_{o^{\neg i}_t\in \mathcal{O}^{\neg i}_t}
\sum_{z^{\neg i}_{t+1}\in \mathcal{Z}^{\neg i}}
\Pr\{x_{t+1},z_{t+1}|x_t,u_t,o_t, s^{\mathcal{B}^i}_t, a^i_t\}\nonumber\\
&\qquad\cdot\Pr\{u_t|x_t,o_t, s^{\mathcal{B}^i}_t, a^i_t\}
\cdot \Pr\{x_t,o^{\neg i}_t| o^i_t, s^{\mathcal{B}^i}_t, a^i_t\} \cdot \Pr\{o^i_t| s^{\mathcal{B}^i}_t, a^i_t\}.
\end{align*}
The first factor in the decomposition of the joint probability is the dynamics model, \ie $p_{x_t,x_{t+1}}^{u_t,z_{t+1}}$, the second denotes the probability of taking joint action given decision rules, \ie $ a^i_t(u_t^i|o^i_t) \cdot a^{\neg i}_t(u_t^{\neg i}|o^{\neg i}_t)$, and the following factor is the value at private occupancy state $s_t^{i,o^i_t}(x_t,o_t)$, and the last factor is the value at occupancy state $s^{\mathcal{B}^i}_t(s^{i,o^i_t}_t)$, \ie
\begin{align}
&= a^i_t(u_t^i|o^i_t) \cdot \lambda^{i,o^i_t} \left( \sum_{x_t,x_{t+1}\in \mathcal{X}}\sum_{u^{\neg i}_t\in \mathcal{U}^{\neg i}}\sum_{o^{\neg i}_t\in \mathcal{O}^{\neg i}_t}\sum_{z^{\neg i}_{t+1}\in \mathcal{Z}^{\neg i}} p_{x_t,x_{t+1}}^{u_t,z_{t+1}}\cdot a^{\neg i}_t(u_t^{\neg i}|o^{\neg i}_t) \cdot s_t^{i,o^i_t}(x_t,o_t)\right).
\label{eqn:lem:transition:relation:between:occupancy:state:and:private:ones:0}
\end{align}
The injection of $\omega^i_{M''}(z^i_{t+1}|s^{i,o^i_t}_t,u^i_t)$, \cf Definition \ref{def:slave:problem:as:mdp}, into \eqref{eqn:lem:transition:relation:between:occupancy:state:and:private:ones:0} yields
\begin{align*}
&= a^i_t(u_t^i|o^i_t) \cdot s^{\mathcal{B}^i}_t(s_t^{i,o^i_t}) \cdot \omega^i_{M''}(z^i_{t+1}|s^{i,o^i_t}_t,u^i_t).
\end{align*}
If we group together all $s^{i,(o^i_t,u^i_t,z^i_{t+1})}_{t+1}$ with the same private occupancy state $s^i_{t+1}$ regardless the private history $(o^i_t,u^i_t,z^i_{t+1})$, we have that:
\begin{align}
s^{\mathcal{B}^i}_{t+1}(s^i_{t+1})&= \sum_{o^i_t\in \mathcal{O}_t^i} \lambda^{i,o^i_t} \sum_{u^i_t\in \mathcal{U}^i} a^i_t(u_t^i|o^i_t) \sum_{z^i_{t+1}\in \mathcal{Z}^i} \omega^i_{M''}(z^i_{t+1}|s^{i,o^i_t}_t,u^i_t) \cdot \delta_{\rho^i_{M''}(s_t^{i,o^i_t}, u^i_t,z^i_{t+1})}(s^i_{t+1}).
\label{eqn:lem:transition:relation:between:occupancy:state:and:private:ones:1}
\end{align}
The injection of $p^i_{M''}(s^{i,o^i_t}_t,u^i_t,s^{i,(o^i_t,u^i_t,z^i_{t+1})}_{t+1})$, \cf Definition \ref{def:slave:problem:as:mdp}, into \eqref{eqn:lem:transition:relation:between:occupancy:state:and:private:ones:1} yields:
\begin{align*}
s^{\mathcal{B}^i}_{t+1}&= \sum_{o^i_t\in \mathcal{O}_t^i} \lambda^{i,o^i_t} \sum_{u^i_t\in \mathcal{U}^i} a^i_t(u_t^i|o^i_t) \sum_{z^i_{t+1}\in \mathcal{Z}^i} p^i_{M''}(s^{i,o^i_t}_t,u^i_t,s^{i,(o^i_t,u^i_t,z^i_{t+1})}_{t+1})\cdot s^{i,(o^i_t,u^i_t,z^i_{t+1})}_{t+1}.
\end{align*}
which ends the proof.
\end{proof}

Like the previous lemma, Lemma \ref{lem:transition:relation:between:occupancy:state:and:private:ones} demonstrates that the transition from one occupancy state to another is a linear function, regardless of the basis used or the perspective of any agent. The selection of the decision rules, known as the policy of agent $i$, creates a non-observable and deterministic Markov game, which is another occupancy-state Markov game from the perspective of agent $i$.

\begin{definition}
A $n$-agent, simultaneous-move, occupancy Markov game \wrt~ $M''(a^{\neg i}_{0:})$ is given by a tuple $\langle \mathcal{S}^{\mathcal{B}^i}, \mathcal{A}^i, \rho^{i}_{M''}, r^{i}_{M''}\rangle$, where $\mathcal{S}^{\mathcal{B}^i}$ is the occupancy-state space, occupancy states being conditional probability distributions over private occupancy states given a joint policy followed so far; $\mathcal{A}^i$ is the space of actions describing decision rules of agent $i$; $\rho^i_{M''} \colon \mathcal{S}^{\mathcal{B}^i} \times \mathcal{A}^i \to \mathcal{S}^{\mathcal{B}^i}$ is the deterministic transition rule, and $r^i_{M''} \colon \mathcal{S}^{\mathcal{B}^i} \times \mathcal{A}^i \to \mathbb{R}$ describes the linear reward function over occupancy states.
\end{definition}

It is worth noting that $M''(a^{\neg i}_{0:})$ serves as a restricted version of $M''$, which solely takes into account the occupancy states based on the joint policy $a^{\neg i}_{0:}$ of all agents, except for agent $i$. Private occupancy states suffice to optimally solve $M''(a^{\neg i}_{0:})$ (and $M(a^{\neg i}_{0:})$). Additionally, $M'(a^{\neg i}_{0:})$ (and $M(a^{\neg i}_{0:})$) encompass a finite number of private occupancy states. Therefore, the optimal value function over private occupancy states can be represented as a tabular function, which ultimately becomes a finite vector of values. This leads to a crucial conclusion.

\begin{restatable}{thm}{thmrelationbetweenoccupancystateandprivateones}
\label{thm:relation:between:occupancy:state:and:private:ones}
For each agent $i$, the $t$-step optimal state-value function for slave game $M''(a^{\neg i}_{0:})$ (respectively $M'(a^{\neg i}_{0:})$ and $M(a^{\neg i}_{0:})$), denoted as $\upsilon^i_{M''(a^{\neg i}_{t:}),\gamma, t}\colon \mathcal{S}_t^{\mathcal{B}^i} \to \mathbb{R}$, is linear over occupancy states expressed in $\mathcal{B}^i$, \ie at any occupancy state $s^{\mathcal{B}^i}_t \doteq \sum_{o_t^i\in \mathcal{O}^i_t} \lambda^{i,o_t^i} \cdot s_t^{i,o_t^i}~ \text{ such that }\sum_{o_t^i\in \mathcal{O}^i_t} \lambda^{i,o_t^i} =1 ~\text{ and } \lambda^{i,o_t^i} \geq 0$,
\begin{align*}
\upsilon^i_{M''(a^{\neg i}_{t:}),\gamma, t}(s_t^{\mathcal{B}^i}) &=
{\textstyle
\sum_{o^i_t\in \mathcal{O}^i_t} \lambda^{i,o^i_t} \cdot \upsilon^i_{M''(a^{\neg i}_{t:}),\gamma, t}(s^{i,o^i_t}_t)
}.
\end{align*}
\end{restatable}

\begin{proof}
The proof proceeds by induction. The statement trivially holds at time step $t=\ell$, since $\upsilon^{i,\cdot}_{M'',\gamma, \ell}(\cdot) = 0$. Suppose the statement holds from time step $t+1$ onward---\ie
\begin{align*}
\upsilon^i_{M''(a^{\neg i}_{t+1:}),\gamma, t+1}(s_{t+1}^{\mathcal{B}^i}) &={\textstyle \sum_{o^i_{t+1}\in \mathcal{O}^i_{t+1}}}~ \lambda^{i,o^i_{t+1}} \cdot\upsilon^i_{M''(a^{\neg i}_{t+1:}),\gamma, t+1}(s_{t+1}^{i,o^i_{t+1}}).
\end{align*}
We can now prove it also holds for time step $t$. We start with the definition of the optimal state-value function for slave game $M''(a^{\neg i}_{0:})$ given any arbitrary occupancy state $s^{\mathcal{B}^i}_t$ at any time step $t$, \ie
\begin{align}
\upsilon^i_{M''(a^{\neg i}_{t:}),\gamma, t}(s^{\mathcal{B}^i}_t) &\doteq
{\textstyle \max_{a^i_{t:}\in \mathcal{A}^i_{t:}}}~\upsilon^{i,a_{t:}}_{M'',\gamma, t}(s^{\mathcal{B}^i}_t).
\label{eqn:proof:thm:relation:between:occupancy:state:and:private:ones}
\end{align}
Bellman's optimality equations allow us to rewrite \eqref{eqn:proof:thm:relation:between:occupancy:state:and:private:ones} as follows
\begin{align}
& \overset{\eqref{eqn:proof:thm:relation:between:occupancy:state:and:private:ones}}{=}
\textstyle \max_{a^i_t\in \mathcal{A}^i_t}~\left\{r^i_{M''}(s^{\mathcal{B}^i}_t, a^i_t) + \gamma \max_{a^i_{t+1:}\in \mathcal{A}^i_{t+1:}}~\upsilon^{i,a_{t+1:}}_{M'',\gamma, t+1}(\rho^i_{M''}(s^{\mathcal{B}^i}_t, a^i_t) )\right\}.
\label{proof:thm:relation:between:occupancy:state:and:private:ones:0}
\end{align}
Re-arranging terms in \eqref{proof:thm:relation:between:occupancy:state:and:private:ones:0} using \eqref{eqn:proof:thm:relation:between:occupancy:state:and:private:ones} leads us to the following expression
\begin{align}
& \overset{\eqref{eqn:proof:thm:relation:between:occupancy:state:and:private:ones}}{=}
\textstyle \max_{a^i_t\in \mathcal{A}^i_t}~\left\{r^i_{M''}(s^{\mathcal{B}^i}_t, a^i_t) + \gamma \upsilon^i_{M''(a^{\neg i}_{t+1:}),\gamma, t+1}(\rho^i_{M''}(s^{\mathcal{B}^i}_t, a^i_t) )\right\}.
\label{proof:thm:relation:between:occupancy:state:and:private:ones:1}
\end{align}
Using the induction hypothesis, Lemma \ref{lem:transition:relation:between:occupancy:state:and:private:ones}, and Lemma \ref{lem:reward:relation:between:occupancy:state:and:private:ones} into \eqref{proof:thm:relation:between:occupancy:state:and:private:ones:1} yields
\begin{align*}
& \overset{\eqref{eqn:proof:thm:relation:between:occupancy:state:and:private:ones}}{=}
{\textstyle \max_{a^i_t\in \mathcal{A}^i_t}~
\sum_{o^i_t\in \mathcal{O}^i_t}
\lambda^{i,o^i_t}
\sum_{u^i_t\in \mathcal{U}^i}
a^i_t ( u^i_t |o^i_t )
}\nonumber\\
&\quad
\textstyle \left\{
r^i_{M''}(s^{i,o^i_t}_t, u^i_t) +
\gamma \sum_{z^i_{t+1}\in \mathcal{Z}^i}
p^i_{M''}(s^{i,o^i_t}_t, u^i_t, s^{i,(o^i_t,u^i_t,z^i_{t+1})}_{t+1}) \cdot
\upsilon^i_{M''(a^{\neg i}_{t+1:}),\gamma, t+1}(s_{t+1}^{i,(o^i_t,u^i_t,z^i_{t+1})})
\right\}
\end{align*}
Re-arranging terms using definition of action-value function $q^i_{M''(a^{\neg i}_{t:}),\gamma, t}$ gives the following expression
\begin{align*}
& \overset{\eqref{eqn:proof:thm:relation:between:occupancy:state:and:private:ones}}{=}
{\textstyle \max_{a^i_t\in \mathcal{A}^i_t}~
\sum_{o^i_t\in \mathcal{O}^i_t}
\lambda^{i,o^i_t}
\sum_{u^i_t\in \mathcal{U}^i}
a^i_t ( u^i_t |o^i_t )
\cdot
q^i_{M''(a^{\neg i}_{t:}),\gamma, t}(s^{i,o^i_t}_t, u^i_t)
}
\end{align*}
One can reason independently for each private occupancy state $s^{i,o^i_t}_t$, which produces
\begin{align*}
& \overset{\eqref{eqn:proof:thm:relation:between:occupancy:state:and:private:ones}}{=}
{\textstyle
\sum_{o^i_t\in \mathcal{O}^i_t}
\lambda^{i,o^i_t}
\max_{u^i_t\in \mathcal{U}^i}
q^i_{M''(a^{\neg i}_{t:}),\gamma, t}(s^{i,o^i_t}_t, u^i_t)
}.
\end{align*}
Hence, we have
\begin{align*}
& \overset{\eqref{eqn:proof:thm:relation:between:occupancy:state:and:private:ones}}{=}{
\textstyle
\sum_{o^i_t\in \mathcal{O}^i_t} \lambda^{i,o^i_t} \cdot \upsilon^i_{M''(a^{\neg i}_{t:}),\gamma, t}(s^{i,o^i_t}_t)
}.
\end{align*}
Which ends the proof.
\end{proof}

Theorem \ref{thm:relation:between:occupancy:state:and:private:ones} states that the optimal state-value functions for any slave game, denoted as $M''(a^{\neg i}_{0:})$, $M'(a^{\neg i}_{0:})$ and $M(a^{\neg i}_{0:})$, are linear over occupancy states expressed in the right basis. This holds true regardless of whether agent $i$ maximizes or minimizes the $\gamma$-discounted sum of rewards. This theorem is a useful tool for establishing convexity properties in various classes of games.

\begin{restatable}{thm}{corconvexoptimalvaluefctmasterzerosum}
\label{cor:convex:optimal:value:fct:master:zerosum}
Let $M''$ be a two-agent, zero-sum, simultaneous-move, occupancy Markov game \wrt $M$ and $M'$, where occupancy states are expressed in the basis of an agent. The $t$-step optimal state-value function of any agent $i$, denoted as $\upsilon^{i,*}_{M'',\gamma, t}\colon \mathcal{S}_t^{\mathcal{B}^{\neg i}} \to \mathbb{R}$, is a convex function over occupancy states expressed in the basis of agent $\neg i$. For each agent $i$, at each time step $t$ there exists a collection $\mathcal{V}^i_t$ of (possibly uncountable) vectors \st for any occupancy state $s^{\mathcal{B}^{\neg i}}_t \doteq \sum_{o_t^{\neg i}\in \mathcal{O}^{\neg i}_t} \lambda^{{\neg i},o_t^{\neg i}} \cdot s_t^{{\neg i},o_t^{\neg i}}~ \text{ such that }\sum_{o_t^{\neg i}\in \mathcal{O}^{\neg i}_t} \lambda^{{\neg i},o_t^{\neg i}} =1 ~\text{ and } \lambda^{{\neg i},o_t^{\neg i}} \geq 0$,
\begin{align*}
\upsilon^{i,*}_{M'',\gamma, t}(s_t^{\mathcal{B}^{\neg i}}) &= \textstyle \max_{\upsilon^i_{M''(a^i_{t:}),\gamma, t}\in \mathcal{V}^i_t}~ \upsilon^i_{M''(a^i_{t:}),\gamma, t}(s_t^{\mathcal{B}^{\neg i}}),
\end{align*}
where $\upsilon^i_{M''(a^i_{t:}),\gamma, t} \doteq -\upsilon^{\neg i}_{M''(a^i_{t:}),\gamma, t}$, with boundary condition $\upsilon^{\cdot,*}_{M'',\gamma, \ell}(\cdot)=0$. Also, the $t$-step optimal state-value function of any agent $i$, denoted as $\upsilon^{i,*}_{M'',\gamma, t}\colon\mathcal{S}_t^{\mathcal{B}^{0}} \to\mathbb{R}$, are the maximum of concave (CConcave) functions over occupancy states expressed in the standard basis. For each agent $i$, at each time step $t$ there exists a (possibly uncountable) collection $\mathbb{V}^i_t = \{\mathcal{V}^i_{M''(a^i_{t:}),\gamma,t} | a^i_{t:} \in \mathcal{A}^i_{t:}\}$ of finite collections $\mathcal{V}^i_{M''(a^i_{t:}),\gamma,t}=\{\upsilon^{i,a_{t:}}_{M'',\gamma, t}| a^{\neg i}_{t:} \in \mathcal{A}^{\neg i,\text{pure}}_{t:}\}$ of vectors \st for any occupancy state $s^{\mathcal{B}^{0}}_t$,
\begin{align*}
\upsilon^{i,*}_{M'',\gamma, t}(s_t^{\mathcal{B}^{0}}) &= \textstyle \max_{\mathcal{V}^i_{M''(a^i_{t:}),\gamma,t}\in \mathbb{V}^i_t}\min_{\upsilon^{i,a_{t:}}_{M'',\gamma, t}\in \mathcal{V}^i_{M''(a^i_{t:}),\gamma,t}}~ \upsilon^{i,a_{t:}}_{M'',\gamma, t}(s_t^{\mathcal{B}^{0}}),
\end{align*}
where $ \upsilon^{i,a_{t:}}_{M'',\gamma, t} \doteq - \upsilon^{\neg i,a_{t:}}_{M'',\gamma, t}$, with boundary condition $\upsilon^{\cdot,*}_{M'',\gamma, \ell}(\cdot)=0$.
\end{restatable}

\begin{proof}
The proof holds at any arbitrary time step $t$, any (maximizing) agent $i$, and any arbitrary occupancy state $s^{\mathcal{B}^{\neg i}}_t$, we start with the definition of the optimal value and then use Lemma \ref{lem:linearity:joint:policy} and Theorem \ref{thm:relation:between:occupancy:state:and:private:ones}, \ie
\begin{align}
\upsilon^{i,*}_{M'',\gamma, t}(s_t^{\mathcal{B}^{\neg i}}) &\doteq
\textstyle \max_{a^i_{t:}\in \mathcal{A}^i_{t:}}\min_{a^{\neg i}_{t:}\in \mathcal{A}^{\neg i}_{t:}}~
\mathbb{E}\{ \upsilon^{i,a_{t:}}_{M,\gamma, t}(X_t,O_t)|s_t^{\mathcal{B}^{\neg i}}\}.
\label{eqn:cor:convex:optimal:value:fct:master:zerosum:1}
\end{align}
Replacing value function $\upsilon^{i,a_{t:}}_{M,\gamma, t}$ by $\upsilon^{\neg i,a_{t:}}_{M,\gamma, t}$ into \eqref{eqn:cor:convex:optimal:value:fct:master:zerosum:1} using the zero-sum property of the game, \ie $\upsilon^{i,a_{t:}}_{M,\gamma, t} + \upsilon^{\neg i,a_{t:}}_{M,\gamma, t}= 0$, yields:
\begin{align}
\upsilon^{i,*}_{M'',\gamma, t}(s_t^{\mathcal{B}^{\neg i}}) &\doteq
\textstyle \max_{a^i_{t:}\in \mathcal{A}^i_{t:}}-\max_{a^{\neg i}_{t:}\in \mathcal{A}^{\neg i}_{t:}}~
\mathbb{E}\{ \upsilon^{\neg i,a_{t:}}_{M,\gamma, t}(X_t,O_t)|s_t^{\mathcal{B}^{\neg i}}\}.
\label{eqn:cor:convex:optimal:value:fct:master:zerosum:2}
\end{align}
The application of Lemma \ref{lem:linearity:joint:policy} into \eqref{eqn:cor:convex:optimal:value:fct:master:zerosum:2} results in the following expression:
\begin{align}
\upsilon^{i,*}_{M'',\gamma, t}(s_t^{\mathcal{B}^{\neg i}}) &\textstyle= \max_{a^i_{t:}\in \mathcal{A}^i_{t:}}-\max_{a^{\neg i}_{t:}\in \mathcal{A}^{\neg i}_{t:}}~
\upsilon^{\neg i,a_{t:}}_{M'',\gamma, t}(s_t^{\mathcal{B}^{\neg i}}).
\label{eqn:cor:convex:optimal:value:fct:master:zerosum:3}
\end{align}
The application of Theorem \ref{thm:relation:between:occupancy:state:and:private:ones} into \eqref{eqn:cor:convex:optimal:value:fct:master:zerosum:3} results in the following expression:
\begin{align*}
\upsilon^{i,*}_{M'',\gamma, t}(s_t^{\mathcal{B}^{\neg i}}) &\textstyle = \max_{a^i_{t:}\in \mathcal{A}^i_{t:}}~ -
\upsilon^{\neg i}_{M''(a^i_{t:}),\gamma, t}(s_t^{\mathcal{B}^{\neg i}}).
\end{align*}
It is worth noticing that $\upsilon^{\neg i}_{M''(a^i_{t:}),\gamma, t}$ is linear in occupancy states expressed in $\mathcal{B}^{\neg i}$ but concave in occupancy states expressed in $\mathcal{B}^0$. If we let $\upsilon^i_{M''(a^i_{t:}),\gamma, t} \doteq - \upsilon^{\neg i}_{M''(a^i_{t:}),\gamma, t}$, then it follows that
\begin{align*}
\upsilon^{i,*}_{M'',\gamma, t}(s_t^{\mathcal{B}^{\neg i}}) &\textstyle = \max_{a^i_{t:}\in \mathcal{A}^i_{t:}}~ 
\upsilon^i_{M''(a^i_{t:}),\gamma, t}(s_t^{\mathcal{B}^{\neg i}}).
\end{align*}
If we let $\mathcal{V}_t^i \doteq \{ \upsilon^i_{M''(a^i_{t:}),\gamma, t} | a^i_{t:} \in \mathcal{A}^i_{t:} \}$ be a (possibly uncountable) collection of vectors, then we have that $\upsilon^{i,*}_{M'',\gamma, t}(s_t^{\mathcal{B}^{\neg i}}) = \max_{\upsilon^i_{M''(a^i_{t:}),\gamma, t}\in \mathcal{V}_t^i}~ \upsilon^i_{M''(a^i_{t:}),\gamma, t}(s_t^{\mathcal{B}^{\neg i}})$. If we let $\mathbb{V}_t^i \doteq \{\mathcal{V}^i_{M''(a^i_{t:}),\gamma,t} | a^i_{t:} \in \mathcal{A}^i_{t:} \}$ and be a (possibly uncountable) collection of finite collections $\mathcal{V}^i_{M''(a^i_{t:}),\gamma,t}=\{\upsilon^{i,a_{t:}}_{M'',\gamma, t}  | a^{\neg i}_{t:} \in \mathcal{A}^{\neg i}_{t:}\}$ of vectors, then we have that
\begin{align*}
\upsilon^{i,*}_{M'',\gamma, t}(s_t^{\mathcal{B}^0}) &= \max_{\mathcal{V}^i_{M''(a^i_{t:}),\gamma,t} \in \mathbb{V}_t^i}\min_{ \upsilon^{i,a_{t:}}_{M'',\gamma, t}\in \mathcal{V}^i_{M''(a^i_{t:}),\gamma,t} }~  \upsilon^{i,a_{t:}}_{M'',\gamma, t}(s_t^{\mathcal{B}^0}),
\end{align*}
which ends the proof.
\end{proof}

Theorem \ref{cor:convex:optimal:value:fct:master:zerosum} shows that the optimal state-value functions for an agent operating in a zero-sum game $M$, either maximizing or minimizing, are either convex or concave functions of the occupancy states expressed in the opponent viewpoint. It is inappropriate to assume that the optimal value functions of the zero-sum game $M$ are linear since the basis differs between agents. Remarkably, the stronger uniform continuity property is that these optimal state-value functions represent the maximum of concave functions over occupancy states expressed in the standard basis. As discussed later, when the former uniform continuity property fails, the latter still applies. In the case of a two-agent one-sided \emph{zs}-POSG, where agent 2 perceives both the underlying state of the game as well as the actions and perceptions of agent 1, the optimal state-value functions $\upsilon^{1,*}_{M'',\gamma,0:}$ of the master game are convex functions of the belief states. We shall now focus on another essential subclass of POSGs: \emph{st}-POSGs.

\begin{restatable}{thm}{corconvexoptimalvaluefctmasterstackelberg}
\label{cor:convex:optimal:value:fct:master:stackelberg}
Let $M''$ be a two-agent, Stackelberg, simultaneous-move, occupancy Markov game \wrt $M$ and $M'$, where occupancy states are expressed on the basis of the follower. The $t$-step optimal leader state-value function, denoted as $\upsilon^{1,*}_{M'',\gamma, t}\colon \mathcal{S}^{\mathcal{B}^2}_t \to\mathbb{R}$, is a convex function over occupancy states expressed in basis set $\mathcal{B}^2$. Consequently, at each time step $t$ there exists a collection $\mathcal{V}_t$ of (possibly uncountable) vectors \st for any occupancy state $s_t^{\mathcal{B}^2}$, $s^{\mathcal{B}^2}_t \doteq \sum_{o_t^2\in \mathcal{O}^2_t} \lambda^{2,o_t^2} \cdot s_t^{2,o_t^2}~ \text{ such that }\sum_{o_t^2\in \mathcal{O}^2_t} \lambda^{2,o_t^2} =1 ~\text{ and } \lambda^{2,o_t^2} \geq 0$
\begin{align*}
\upsilon^{1,*}_{M'',\gamma, t}(s_t^{\mathcal{B}^2}) &= \max_{ \upsilon^1_{M''(a^1_{t:}),\gamma, t}\in \mathcal{V}_t} \upsilon^1_{M''(a^1_{t:}),\gamma, t}(s_t^{\mathcal{B}^2})
\end{align*}
where $\upsilon^{1}_{M''(a^1_{t:}),\gamma, t}(s_t^2) \doteq \max_{a^2_{t:}\in \mathcal{A}^2_{t:}(s_t^2, a^1_{t:})}~\upsilon^{1,a_{t:}}_{M'',\gamma, t}(s_t^2)$ for every private occupancy state $s^2_t$, with boundary condition $\upsilon^{1,*}_{M'',\gamma, \ell}(\cdot)=\upsilon^{2,\cdot}_{M'',\gamma, \ell}(\cdot)\doteq 0$.
\end{restatable}

\begin{proof}
The proof holds for arbitrary time step $t$, leader agent $1$, and arbitrary occupancy state $s_t^{\mathcal{B}^0}$. We start with the value of strong Stackelberg equilibria at $s_t^{\mathcal{B}^2}$, \ie
\begin{align}
\upsilon^{1,*}_{M'',\gamma, t}(s_t^{\mathcal{B}^2}) & \doteq \max_{a^1_{t:} \in \mathcal{A}^1_{t:}} \max_{a^2_{t:}\in \mathcal{A}^2_{t:}}~
\mathbb{E}\{ \upsilon^{1,a_{t:}}_{M,\gamma, t}(X_t,O_t) | s_t^{\mathcal{B}^2} \}
\quad\text{\st}\quad
\upsilon^{2}_{M''(a^1_{t:}),\gamma, t}(s_t^{\mathcal{B}^2}) = \upsilon^{2,a_{t:}}_{M'',\gamma, t}(s_t^{\mathcal{B}^2}).
\label{eq:cor:convex:optimal:value:fct:master:stackelberg:1}
\end{align}
The application of Lemma \ref{lem:bellman:equation:m} into \eqref{eq:cor:convex:optimal:value:fct:master:stackelberg:1} results in the following expression:
\begin{align*}
&= \max_{a^1_{t:} \in A^1_{t:}} \max_{a^2_{t:}\in \mathcal{A}^2_{t:}}~
\upsilon^{1,a_{t:}}_{M'',\gamma, t}(s_t^{\mathcal{B}^2})
\quad\text{\st}\quad
\upsilon^{2}_{M''(a^1_{t:}),\gamma, t}(s_t^{\mathcal{B}^2}) = \upsilon^{2,a_{t:}}_{M'',\gamma, t}(s_t^{\mathcal{B}^2}).
\end{align*}
If we let $\mathcal{A}^2_{t:}(s_t^2, a^1_{t:}) \doteq \{a^2_{t:} \in \mathcal{A}^2_{t:} | \upsilon^{2,a_{t:}}_{M'',\gamma, t}(s_t^2) = \upsilon^{2}_{M''(a^1_{t:}),\gamma, t}(s_t^2) \}$ be the collection of best-response policies of the follower given the leader policy $a^1_{t:}$ alongside private occupancy state $s_t^2$, then the following expansion holds from Lemma \ref{lem:relation:between:occupancy:state:and:private:ones}:
\begin{align*}
\upsilon^{1,*}_{M'',\gamma, t}(s_t^{\mathcal{B}^2}) &\textstyle = \max_{a^1_{t:} \in \mathcal{A}^1_{t:}}
\sum_{s^2_t\in \mathtt{supp}(s_t^{\mathcal{B}^2})} \lambda^{2,s^2_t}
\cdot \left\{ \max_{a^2_{t:}\in \mathcal{A}^2_{t:}(s_t^2, a^1_{t:})}~\upsilon^{1,a_{t:}}_{M'',\gamma, t}(s_t^2) \right\}
\end{align*}
If we let $\upsilon^{1}_{M''(a^1_{t:}),\gamma, t}(s_t^2) \doteq \max_{a^2_{t:}\in \mathcal{A}^2_{t:}(s_t^2, a^1_{t:})}~\upsilon^{1,a_{t:}}_{M'',\gamma, t}(s_t^2)$, then re-arranging terms yields the following expression:
\begin{align}
\upsilon^{1,*}_{M'',\gamma, t}(s_t^{\mathcal{B}^2}) &\textstyle = \max_{a^1_{t:} \in \mathcal{A}^1_{t:}} \sum_{s^2_t\in \texttt{supp}(s_t^{\mathcal{B}^2})}
\lambda^{2, s^2_t} \cdot \upsilon^{1}_{M''(a^1_{t:}),\gamma, t}(s_t^2).
\label{eq:cor:convex:optimal:value:fct:master:stackelberg:2}
\end{align}
The application of Theorem \ref{thm:relation:between:occupancy:state:and:private:ones} into \eqref{eq:cor:convex:optimal:value:fct:master:stackelberg:2} results in the following expression:
\begin{align*}
\upsilon^{1,*}_{M'',\gamma, t}(s_t^{\mathcal{B}^2}) &\textstyle = \max_{a^1_{t:}\in \mathcal{A}^1_{t:}}~ \upsilon^{1}_{M''(a^1_{t:}),\gamma, t}(s_t^{\mathcal{B}^2}).
\end{align*}
If we let $\mathcal{V}_t^1 \doteq \{ \upsilon^{1}_{M''(a^1_{t:}),\gamma, t} | a^1_{t:} \in \mathcal{A}^1_{t:} \}$ be a (possibly uncountable) collection of vectors, then we have that
\begin{align*}
\upsilon^{1,*}_{M'',\gamma, t}(s_t^{\mathcal{B}^2}) &= \max_{\upsilon^{1}_{M''(a^1_{t:}),\gamma, t}\in \mathcal{V}_t^1}~\upsilon^{1}_{M''(a^1_{t:}),\gamma, t}(s_t^{\mathcal{B}^2}).
\end{align*}
Which ends the proof.
\end{proof}

According to Theorem \ref{cor:convex:optimal:value:fct:master:stackelberg}, the optimal state-value functions for strong Stackelberg equilibria from the leader's perspective are convex when expressed in terms of occupancy states from the follower's viewpoint. In the case of a two-agent one-sided \emph{st}-POSG, where agent 2 perceives both the underlying state of the game as well as the actions and perceptions of agent 1, the optimal state-value functions $\upsilon^{1,*}_{M'',\gamma,0:}$ of the master game are convex functions of the belief states. Moving forward, we will investigate the structure of the optimal state-value function in common-reward partially observable stochastic games.

\begin{restatable}{thm}{corconvexoptimalvaluefctmasterdecpomdp}
\label{cor:convex:optimal:value:fct:master:decpomdp}
Let $M''$ be a $n$-agent, common-reward, simultaneous-move, occupancy Markov game \wrt $M'$ and $M$. The $t$-step optimal state-value function, denoted as $\upsilon^*_{M'',\gamma, t}\colon \mathcal{S}^{\mathcal{B}^0}_t \to\mathbb{R}$, is a piecewise linear and convex function over occupancy states expressed in the basis set $\mathcal{B}^0$. Consequently, at each time step $t$ there exists a finite collection of vectors $\mathcal{V}_t$ \st for any occupancy state $s_t^{\mathcal{B}^0}$,
\begin{align}
\upsilon^*_{M'',\gamma, t}(s_t^{\mathcal{B}^0}) &= \max_{\upsilon^{a_{t:}}_{M'',\gamma, t}\in \mathcal{V}_t} \upsilon^{a_{t:}}_{M'',\gamma, t}(s_t^{\mathcal{B}^0})
\label{eq:cor:convex:optimal:value:fct:master:decpomdp}
\end{align}
with boundary condition $\upsilon^*_{M'',\gamma, \ell}(\cdot)\doteq 0$.
\end{restatable}

\begin{proof}
The proof holds for arbitrary time step $t$, arbitrary agent $i$, and arbitrary occupancy state $s_t^{\mathcal{B}^0}$. We start with the definition of the optimal value at $s_t^{\mathcal{B}^0}$, \ie
\begin{align}
\upsilon^*_{M'',\gamma, t}(s_t^{\mathcal{B}^0}) & \doteq \textstyle \max_{a^i_{t:} \in \mathcal{A}^i_{t:}}\max_{a^{\neg i}_{t:}\in \mathcal{A}^{\neg i}_{t:}}~\mathbb{E}\{ \upsilon^{a_{t:}}_{M,\gamma, t}(X_t,O_t)|s_t^{\mathcal{B}^0}\}.
\label{eq:cor:convex:optimal:value:fct:master:decpomdp:1}
\end{align}
The application of Lemma \ref{lem:bellman:equation:m} into \eqref{eq:cor:convex:optimal:value:fct:master:decpomdp:1} results in the following expression:
\begin{align*}
\upsilon^*_{M'',\gamma, t}(s_t^{\mathcal{B}^0}) &= \textstyle \max_{a^i_{t:} \in \mathcal{A}^i_{t:}}\max_{a^{\neg i}_{t:}\in \mathcal{A}^{\neg i}_{t:}}~\upsilon^{a_{t:}}_{M'',\gamma, t}(s_t^{\mathcal{B}^0}).
\end{align*}
With all agents having the same reward model, one can show that restricting attention to pure policies does not hurt optimality because any pure policy that is a best-response policy achieves the same performance index.
\begin{align*}
\upsilon^*_{M'',\gamma, t}(s_t^{\mathcal{B}^0}) &= \textstyle \max_{a_{t:}\in \mathcal{A}^{i,\text{pure}}_{t:}\otimes \mathcal{A}^{\neg i,\text{pure}}_{t:}}~\upsilon^{a_{t:}}_{M'',\gamma, t}(s_t^{\mathcal{B}^0}).
\end{align*}
Consequently, defining the finite collection $\mathcal{V}_t \doteq \{\upsilon^{a_{t:}}_{M'',\gamma, t} | a_{t:}\in \mathcal{A}^{i,\text{pure}}_{t:}\otimes \mathcal{A}^{\neg i,\text{pure}}_{t:}\}$, then \eqref{eq:cor:convex:optimal:value:fct:master:decpomdp} holds. Which ends the proof
\end{proof}

\subsection{Discussion}
The previous subsection elucidates the optimal state-value functions for \emph{zs}-POSGs, \emph{st}-POSGs, and \emph{dec}-POMDPs, demonstrating that they are convex functions of occupancy states expressed on the appropriate basis. A crucial element in determining the appropriate basis for occupancy states is its ability to reveal the convexity properties of the optimal state-value functions. The standard basis for occupancy states, which possess stronger generalization capabilities, is sufficient to disclose the convexity properties of POSGs with a single criterion. However, it often fails to reveal the underlying convexity properties of the optimal state-value functions for POSGs with two criteria, one per agent.

For instance, in \emph{dec}-POMDPs, the objective is to find a joint policy that maximizes the expected discounted cumulative reward.
\citet{continuous} established that the optimal state-value function for partially observable stochastic games with common reward is piecewise linear and convex in the occupancy state expressed in the standard basis. Theorem \ref{cor:convex:optimal:value:fct:master:decpomdp} explicitly connects the optimal state-value functions for the best-response problems.
Conversely, in \emph{zs}-POSGs and \emph{st}-POSGs, the aim is to discover joint policies that maximize the expected discounted cumulative reward of one agent while simultaneously constraining on the expected discounted cumulative reward of the other. In the latter, the basis of interest is determined by the private occupancy states of the other agents. Occupancy states expressed on this basis can reveal the convexity properties of the optimal state-value functions for \emph{zs}-POSGs and \emph{st}-POSGs. Notice, however, that not all bases have the same generalization power. The standard basis, a collection of one-hot vectors, exhibits a stronger generalization power than the basis from the opponent's viewpoint. Unfortunately, the optimal state-value functions for \emph{zs}-POSGs and \emph{st}-POSGs are not convex functions of occupancy states expressed on the standard basis as illustrated in the following example.

\begin{example}
\begin{figure}
\centering
\begin{tikzpicture}
\colorlet{color min rgb}[rgb]{sthlmRed}
\colorlet{color max rgb}[rgb]{sthlmGreen}
\colorlet{color min hsb}[hsb]{sthlmRed}
\colorlet{color max hsb}[hsb]{sthlmGreen}
\def\min{-20}
\def\max{101}
\node (agentCentral) at (-2,.5) {};
\node (agentA) at (-.25,2.25) {};
\node (agentD) at (-.25,-1) {};
\path[->, -latex, color=sthlmRed]
(agentCentral) edge [bend right] node[fill=white, scale=.65] {$r^1$} (agentA)
edge [bend left] node[fill=white, scale=.65] {$r^2$} (agentD);
\draw[white] (-7.5,1.5) -- (-3.5,1.5);
\draw[white] (-7.5,.5) -- (-3.5,.5);
\draw[white] (-7.5,-.5) -- (-3.5,-.5);
\draw[white] (-7.5,-1.5) -- (-3.5,-1.5);
\draw[white] (-6.5,2.5) -- (-6.5,-1.5);
\draw[white] (-5.5,2.5) -- (-5.5,-1.5);
\draw[white] (-4.5,2.5) -- (-4.5,-1.5);
\draw[white] (-3.5,2.5) -- (-3.5,-1.5);
\pgfmathtruncatemacro\lambda{(2-\min)/(\max-\min)*100}
\colorlet{my color rgb}[rgb]{color min rgb!\lambda!color max rgb}
\colorlet{my color hsb}[rgb]{color min hsb!\lambda!color max hsb}
\fill[fill=my color rgb,rounded corners] (-6.5,.5) rectangle +(1,1) node [draw=none, text centered, scale=.75, pos=.5, white] {+1};
\pgfmathtruncatemacro\lambda{(-9-\min)/(\max-\min)*100}
\colorlet{my color rgb}[rgb]{color min rgb!\lambda!color max rgb}
\colorlet{my color hsb}[rgb]{color min hsb!\lambda!color max hsb}
\fill[fill=my color rgb,rounded corners] (-5.5,.5) rectangle +(1,1) node [draw=none, text centered, scale=.75, pos=.5, white] {0};
\pgfmathtruncatemacro\lambda{(101-\min)/(\max-\min)*100}
\colorlet{my color rgb}[rgb]{color min rgb!\lambda!color max rgb}
\colorlet{my color hsb}[rgb]{color min hsb!\lambda!color max hsb}
\fill[fill=my color rgb,rounded corners] (-4.5,.5) rectangle +(1,1) node [draw=none, text centered, scale=.75, pos=.5, white] {0};
\pgfmathtruncatemacro\lambda{(-9-\min)/(\max-\min)*100}
\colorlet{my color rgb}[rgb]{color min rgb!\lambda!color max rgb}
\colorlet{my color hsb}[rgb]{color min hsb!\lambda!color max hsb}
\fill[fill=my color rgb,rounded corners] (-6.5,-.5) rectangle +(1,1) node [draw=none, text centered, scale=.75, pos=.5, white] {0};
\pgfmathtruncatemacro\lambda{(-20-\min)/(\max-\min)*100}
\colorlet{my color rgb}[rgb]{color min rgb!\lambda!color max rgb}
\colorlet{my color hsb}[rgb]{color min hsb!\lambda!color max hsb}
\fill[fill=my color rgb,rounded corners] (-5.5,-.5) rectangle +(1,1) node [draw=none, text centered, scale=.75, pos=.5, white] {+2};
\pgfmathtruncatemacro\lambda{(100-\min)/(\max-\min)*100}
\colorlet{my color rgb}[rgb]{color min rgb!\lambda!color max rgb}
\colorlet{my color hsb}[rgb]{color min hsb!\lambda!color max hsb}
\fill[fill=my color rgb,rounded corners] (-4.5,-.5) rectangle +(1,1) node [draw=none, text centered, scale=.75, pos=.5, white] {0};
\pgfmathtruncatemacro\lambda{(101-\min)/(\max-\min)*100}
\colorlet{my color rgb}[rgb]{color min rgb!\lambda!color max rgb}
\colorlet{my color hsb}[rgb]{color min hsb!\lambda!color max hsb}
\fill[fill=my color rgb,rounded corners] (-6.5,-1.5) rectangle +(1,1) node [draw=none, text centered, scale=.75, pos=.5, white] {0};
\pgfmathtruncatemacro\lambda{(100-\min)/(\max-\min)*100}
\colorlet{my color rgb}[rgb]{color min rgb!\lambda!color max rgb}
\colorlet{my color hsb}[rgb]{color min hsb!\lambda!color max hsb}
\fill[fill=my color rgb,rounded corners] (-5.5,-1.5) rectangle +(1,1) node [draw=none, text centered, scale=.75, pos=.5, white] {0};
\pgfmathtruncatemacro\lambda{(50-\min)/(\max-\min)*100}
\colorlet{my color rgb}[rgb]{color min rgb!\lambda!color max rgb}
\colorlet{my color hsb}[rgb]{color min hsb!\lambda!color max hsb}
\fill[fill=my color rgb,rounded corners] (-4.5,-1.5) rectangle +(1,1) node [draw=none, text centered, scale=.75, pos=.5, white] {-2};
\node[inner sep=0pt] (hobbes) at (-7,-1) {\includegraphics[width=.04\textwidth]{figures/Hobbes.png}};
\node[inner sep=0pt] (treasure) at (-7,0) {\includegraphics[width=.05\textwidth]{figures/Treasure.jpg}};
\node[inner sep=0pt] (treasure) at (-7,1) {\includegraphics[width=.03\textwidth]{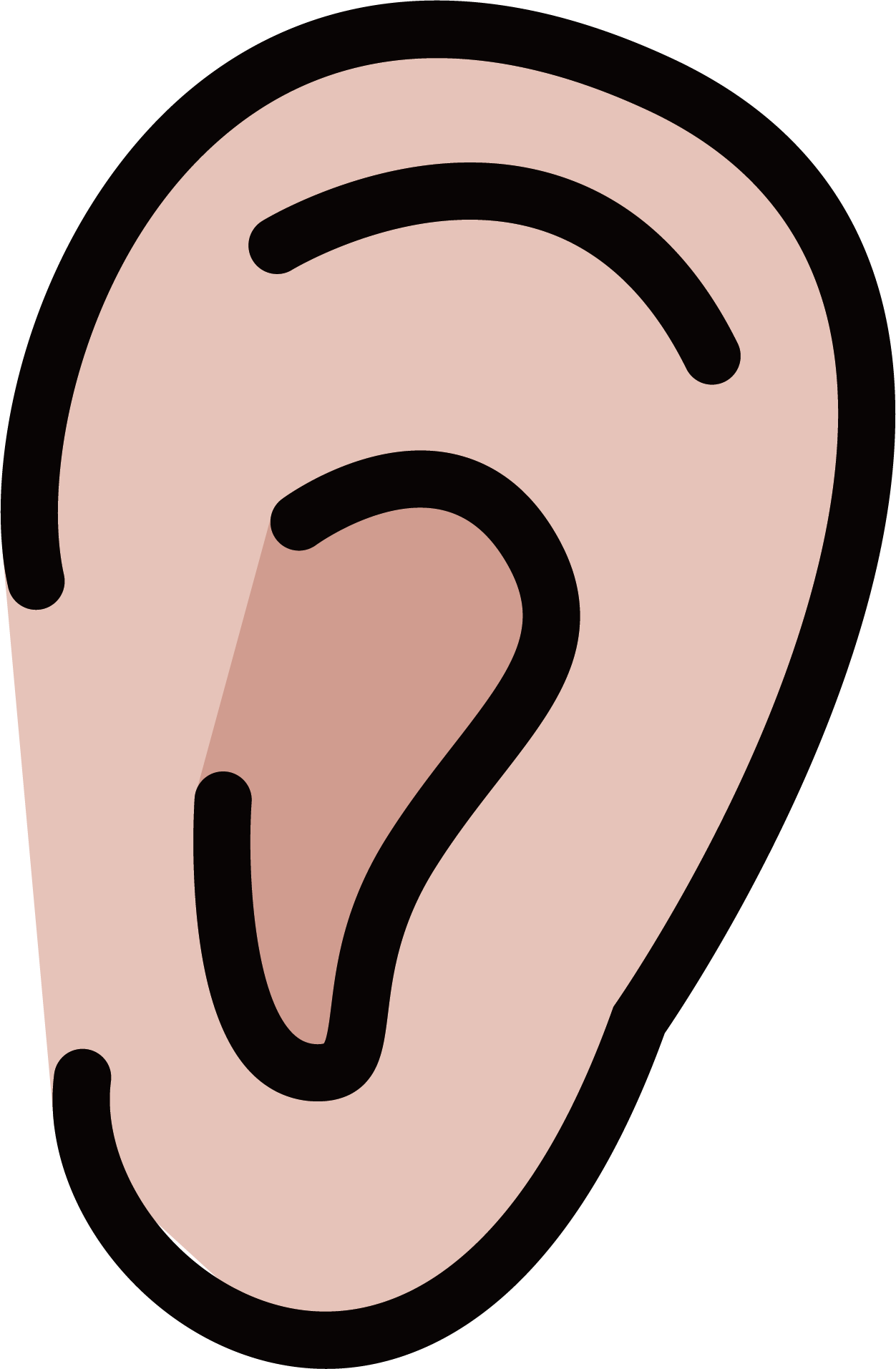}};
\node[inner sep=0pt] (hobbes) at (-4,2) {\includegraphics[width=.04\textwidth]{figures/Hobbes.png}};
\node[inner sep=0pt] (treasure) at (-5,2) {\includegraphics[width=.05\textwidth]{figures/Treasure.jpg}};
\node[inner sep=0pt] (treasure) at (-6,2) {\includegraphics[width=.03\textwidth]{figures/Ear.png}};
\draw[rounded corners, color=white] (-6.5,-1.5) rectangle +(3,3);
\begin{scope}[scale=.5]
\def \payoff{}
\def \agentA{(-.5,4.5) ellipse (1cm and 0.9cm);
\node[inner sep=0pt] (calvin) at (-.5,4.25) {\includegraphics[width=.075\textwidth]{figures/Calvin_Color.png}};}
\def \agentD{(-.5,-2) ellipse (1cm and 0.9cm);
\node[inner sep=0pt] (susie) at (-.5,-2.25) {\includegraphics[width=.075\textwidth]{figures/Susie.png}};}
\draw[draw=white] \agentA;
\draw[draw=white] \agentD;
\node[inner sep=0pt] (hobbes) at (9.4,3) {\includegraphics[width=.075\textwidth]{figures/Hobbes.png}};
\node[inner sep=0pt] (door1) at (7,3) {\includegraphics[width=.075\textwidth]{figures/Door.png}};
\node[inner sep=0pt] (door2) at (7,-1) {\includegraphics[width=.075\textwidth]{figures/Door.png}};
\node[inner sep=0pt] (treasure) at (9.4,-2) {\includegraphics[width=.075\textwidth]{figures/Treasure.jpg}};
\end{scope}
\node[scale=.75] at (-5,-1.75) {\textcolor{sthlmRed}{\sc payoff matrix}};
\draw[->,-latex, color=sthlmRed] (-3.5,.5) -- node[draw=none, fill=white] {$r$} (-2,.5);
\draw[->,-latex, color=sthlmGreen] (.5,-.5) -- node[draw=none, fill=white,,scale=.5] {$u^2$} (2,.25);
\draw[->,-latex, color=sthlmBlue] (2.75,.25) -- node[draw=none, fill=white,,scale=.5] {$z^2$} (1.25,-.5);
\draw[->,-latex, color=sthlmGreen] (.5,2) -- node[draw=none, fill=white,scale=.5] {$u^1$} (2,1.25);
\draw[->,-latex, color=sthlmBlue] (2.75,1.25) -- node[draw=none, fill=white,,scale=.5] {$z^1$} (1.25,2);
\end{tikzpicture}
\caption{Illustration of the payoff matrix for the two-agent tiger problem.}
\label{fig:counterexample}
\end{figure}

The following scenario involves two individuals, named \textcolor{green!60!black}{Calvin} and \textcolor{red!60!black}{Suzy}, standing before two closed doors as illustrated in Figure \ref{fig:counterexample}. The identity of each door and the contents behind them are unknown to \textcolor{green!60!black}{Calvin} and \textcolor{red!60!black}{Suzy}. However, \textcolor{green!60!black}{Calvin} and \textcolor{red!60!black}{Suzy} have the option to take independent actions, either to listen ($\boldsymbol{u_\textsc{l}}$) or to open ($\boldsymbol{u_\textsc{o}}$) one of the doors. When both individuals choose to listen, they receive a small reward of $+1$ ($r(-,\textcolor{green!60!black}{\boldsymbol{u_\textsc{l}}},\textcolor{red!60!black}{\boldsymbol{u_\textsc{l}}})$). If both individuals decide to open the door containing the valuable treasure $\includegraphics[width=.03\textwidth]{figures/Treasure.jpg}$, a reward of $+2$ ($r(\includegraphics[width=.03\textwidth]{figures/Treasure.jpg},\textcolor{green!60!black}{\boldsymbol{u_\textsc{o}}},\textcolor{red!60!black}{\boldsymbol{u_\textsc{o}}})$) is received. However, if both individuals choose to open the door containing the tiger, they will receive a penalty of $-2$ ($r(\includegraphics[width=.02\textwidth]{figures/Hobbes.png},\textcolor{green!60!black}{\boldsymbol{u_\textsc{o}}},\textcolor{red!60!black}{\boldsymbol{u_\textsc{o}}})$). In all other cases, no reward is given ($r(-,-,-) = 0$), regardless of the actions taken by \textcolor{green!60!black}{Calvin} and \textcolor{red!60!black}{Suzy}. After the individuals take their action, the game is over.

If \textcolor{green!60!black}{Calvin} and \textcolor{red!60!black}{Suzy} cooperate to maximize their common reward, \ie $r^1(\cdot) = r^2(\cdot) = r(\cdot)$, the optimal state-value function is a convex function of the initial occupancy states expressed in the standard basis, \ie for any initial occupancy state $s_0$,
\begin{align*}
\upsilon^{*,\emph{\text{\emph{dec}-POMDP}}}_0(s_0) &= \max_{\textcolor{green!60!black}{a}\in \mathcal{P}(\{\boldsymbol{u_\textsc{o}},\boldsymbol{u_\textsc{l}}\})}\max_{\textcolor{red!60!black}{a}\in \mathcal{P}(\{\boldsymbol{u_\textsc{o}},\boldsymbol{u_\textsc{l}}\})} ~ \upsilon^{\textcolor{green!60!black}{a}\textcolor{red!60!black}{a}}_0(\includegraphics[width=.02\textwidth]{figures/Hobbes.png}) \cdot s_0(\includegraphics[width=.02\textwidth]{figures/Hobbes.png}) + \upsilon^{\textcolor{green!60!black}{a}\textcolor{red!60!black}{a}}_0(\includegraphics[width=.03\textwidth]{figures/Treasure.jpg})\cdot s_0(\includegraphics[width=.03\textwidth]{figures/Treasure.jpg}),
\end{align*}
where
$\upsilon^{\textcolor{green!60!black}{a}\textcolor{red!60!black}{a}}_0(\includegraphics[width=.02\textwidth]{figures/Hobbes.png}) = \textcolor{green!60!black}{a(\boldsymbol{u_\textsc{l}})} \cdot \textcolor{red!60!black}{a(\boldsymbol{u_\textsc{l}})}-2 \cdot \textcolor{green!60!black}{a(\boldsymbol{u_\textsc{o}})}\cdot \textcolor{red!60!black}{a(\boldsymbol{u_\textsc{o}})}$ and $\upsilon^{\textcolor{green!60!black}{a}\textcolor{red!60!black}{a}}_0(\includegraphics[width=.03\textwidth]{figures/Treasure.jpg}) = \textcolor{green!60!black}{a(\boldsymbol{u_\textsc{l}})} \cdot \textcolor{red!60!black}{a(\boldsymbol{u_\textsc{l}})}+2 \cdot \textcolor{green!60!black}{a(\boldsymbol{u_\textsc{o}})}\cdot \textcolor{red!60!black}{a(\boldsymbol{u_\textsc{o}})}$.

The point-wise maximum of linear functions is a convex function. This insight is a direct application of Theorem \ref{cor:convex:optimal:value:fct:master:decpomdp}. It is worth noticing that the optimal state-value function generalizes over all occupancy states expressed in the standard basis.

If instead \textcolor{green!60!black}{Calvin} and \textcolor{red!60!black}{Suzy} compete to maximize their own interest, \ie $r^1(\cdot) = - r^2(\cdot) = r(\cdot)$, the optimal state-value function is \textbf{not} a convex function of the initial occupancy states regardless the chosen basis, \ie for any initial occupancy state $s_0$,
\begin{align*}
\upsilon^{*,\emph{\text{\emph{zs}-POSG}}}_0(s_0) &= \max_{\textcolor{green!60!black}{a}\in \mathcal{P}(\{\boldsymbol{u_\textsc{o}},\boldsymbol{u_\textsc{l}}\})}\left[ \min_{\textcolor{red!60!black}{a}\in \mathcal{P}(\{\boldsymbol{u_\textsc{o}},\boldsymbol{u_\textsc{l}}\})} ~ \upsilon^{\textcolor{green!60!black}{a}\textcolor{red!60!black}{a}}_0(\includegraphics[width=.02\textwidth]{figures/Hobbes.png}) \cdot s_0(\includegraphics[width=.02\textwidth]{figures/Hobbes.png}) + \upsilon^{\textcolor{green!60!black}{a}\textcolor{red!60!black}{a}}_0(\includegraphics[width=.03\textwidth]{figures/Treasure.jpg})\cdot s_0(\includegraphics[width=.03\textwidth]{figures/Treasure.jpg})\right].
\end{align*}

\begin{figure}
\centering
\includegraphics[width=.65\textwidth]{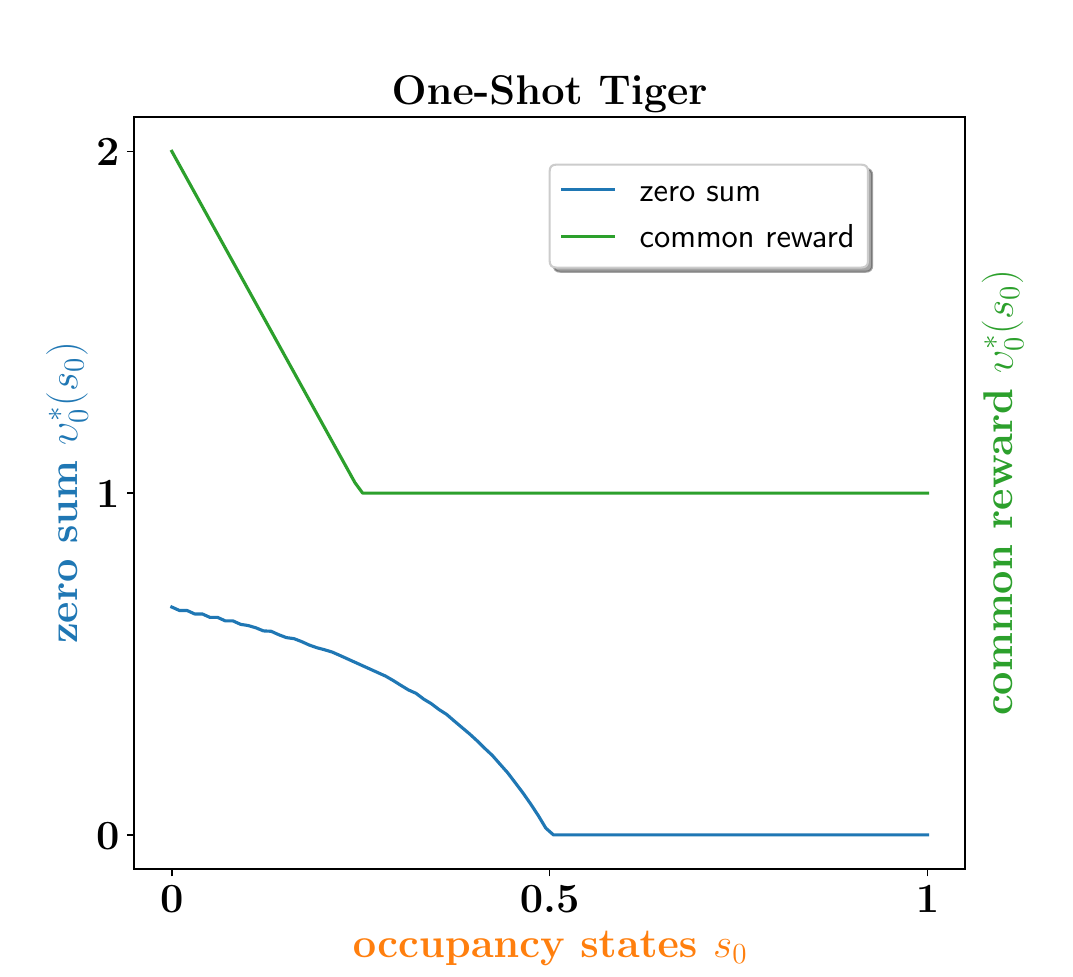}
\caption{The figure showcases the optimal state-value functions for the one-stage tiger game. \textbf{Best viewed in color}.}
\label{fig:one:shot:tiger:plot}
\end{figure}

\begin{figure}
\centering
\includegraphics[width=.65\textwidth]{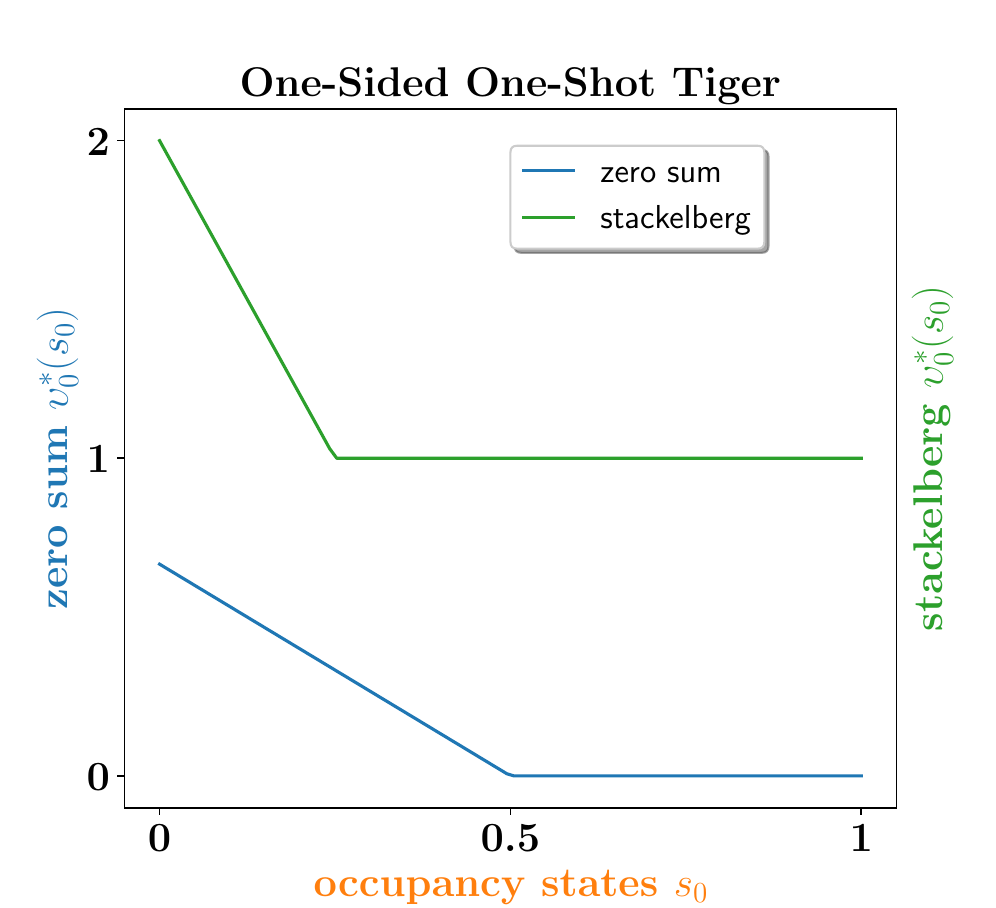}
\caption{The figure showcases the optimal state-value functions for the one-stage, one-sided tiger game. \textbf{Best viewed in color}.}
\label{fig:one:shot:onesided:tiger:plot}
\end{figure}

The function within the bracket represents the optimal value of the best-response at the initial occupancy state. It has been demonstrated that this function is a concave function of the initial occupancy states expressed in the standard basis. However, no more concise statistics can generally be found as the initial occupancy state is also the private occupancy state. Unfortunately, the point-wise maximum of concave functions is not convex. Therefore, the optimal state-value function over the initial occupancy states expressed in the standard basis cannot generally be convex. This negative property holds at the initial stage of the game and throughout its progression.

Theorem \ref{cor:convex:optimal:value:fct:master:zerosum} establishes the property of convexity for zero-sum partially observable stochastic games (zs-POSGs) over occupancy states expressed as distributions over private occupancy states. However, it is noteworthy that when the occupancy states match the private occupancy states, the convexity property vanishes. Convex functions of occupancy states expressed in the standard basis permit generalization from one private occupancy state to another, as seen in Figure \ref{fig:one:shot:tiger:plot}. Conversely, convex functions of distributions over private occupancy states do not allow for generalization from one private occupancy state to another, with the only known exception being when private occupancy states simplify into the underlying states, such as in one-sided two-agent POSGs, as illustrated in Figure \ref{fig:one:shot:onesided:tiger:plot}. In such games, the optimal state-value functions are convex functions of the belief states.

It is worth noting that the convexity property outlined in Theorem \ref{cor:convex:optimal:value:fct:master:zerosum} is not applicable during the initial stage of \emph{zs}-POSGs, as all initial occupancy states are also private occupancy states. Therefore, it is not a useful tool in this particular scenario. This same limitation applies to \emph{st}-POSGs. However, in \emph{zs}-POSGs, it is possible to leverage the maximum of concave functions and generalize from one occupancy state to another, even during the first stage of the game, as demonstrated in Figure \ref{fig:one:shot:tiger:plot}. Unfortunately, \emph{st}-POSGs do not offer this fallback solution.
\end{example}

\section{Conclusion}
Partially Observable Stochastic Games are extensively used to formalize various problems in multi-agent sequential decision-making. However, analyzing the underlying structure of optimal state-value functions is a difficult task, leading to inefficient planning and reinforcement learning algorithms. This paper proposes a three-step methodology to unveil the underlying structure of optimal state-value functions for subclasses of POSGs, including but not restricted to \emph{zs}-POSGs, \emph{st}-POSGs, and \emph{dec}-POMDPs.

The first step involves recasting the original game as a master game and best-response problems as slave games viewed from the perspective of the central planner. The master and slave games are then reformulated as Markov games, where states are the plan-time total data available to the central planner. The second step of the methodology summarizes the plan-time total data available to the central planner using the concept of occupancy states. Occupancy states refer to a posterior probability distribution over hidden states of the game and the stream of actions and observations that agents may encounter at execution time. The final step of the approach blends together the optimal state-value functions from both slave and master games. In particular, it shows that changing the basis on which occupancy states are expressed from the standard basis to agent viewpoints reveals the underlying structures. It enabled us to prove that the optimal state-value functions for \emph{zs}-POSGs, \emph{st}-POSGs, and \emph{dec}-POMDPs are convex functions of occupancy states expressed on the appropriate basis. \citet{structure} and \citet{zerosum} previously exhibited underlying uniform continuity properties in the optimal state-value functions of \emph{zs}-POSGs. Unfortunately, these properties are either hard to use in practice or provide poor generalization capabilities. Our findings provide uniform continuity properties of the optimal state-value functions that are either equal (\eg \emph{dec}-POMDPs) or stronger than pre-existed ones (\eg \emph{zs}-POSGs and \emph{st}-POSGs), \cf Table \ref{tab:uniform:continuity}.

\begin{table}[!ht]
\scriptsize
\tymin=0pt
\tymax=90pt
\begin{tabulary}{\linewidth}{p{2.4cm}LLLLL}
\toprule
& \phantom{Slave Games} & \multicolumn{3}{c}{\textbf{Master Games}} \\
\cmidrule{3-5}
& \textbf{Slave Games} & \emph{zs}-POSG & \emph{st}-POSG & \emph{dec}-POMDP \\
\citeauthor{continuous} & N.A. & N.A. & N.A. & $\mathcal{S}^{\mathcal{B}^0}$ PWLC \\
\citeauthor{structure} & $\mathcal{P}(\mathcal{O}^i)$ Linear & $\mathcal{P}(\mathcal{O}^i)$ Convex & N.A. & N.A.\\
\citeauthor{zerosum} & N.A. & $\mathcal{S}^{\mathcal{B}^0}$ Lipschitz & N.A. & N.A. \\
\cmidrule{2-5}
\textbf{Ours} & $\mathcal{S}^{\mathcal{B}^0}$ PWLC (Th. \ref{cor:pwlc:occ:slave:problem}) \par $\mathcal{S}^{\mathcal{B}^i}$ Linear (Th. \ref{thm:relation:between:occupancy:state:and:private:ones}) & $\mathcal{S}^{\mathcal{B}^2}$ Convex (Th. \ref{cor:convex:optimal:value:fct:master:zerosum}) \par $\mathcal{S}^{\mathcal{B}^0}$ CConcave & $\mathcal{S}^{\mathcal{B}^2}$ Convex (Th. \ref{cor:convex:optimal:value:fct:master:stackelberg}) & $\mathcal{S}^{\mathcal{B}^0}$ PWLC (Th. \ref{cor:convex:optimal:value:fct:master:decpomdp})\\
\bottomrule
\end{tabulary}
\caption{Uniform continuity properties of optimal value functions of POSG $M$. We use notations such as ``$\mathcal{S}^{\mathcal{B}}$ Convex'' to indicate the convexity of the optimal value functions over occupancy states when expressed in basis $\mathcal{B}$.}
\label{tab:uniform:continuity}
\end{table}

We have considered three perception assumptions to evaluate the accuracy and relevance of our main findings. Firstly, it must be acknowledged that in fully observable Markov games, optimal state-value functions are tabular functions of occupancy states for all two-agent and zero-sum, Stackelberg, and common reward criteria. These occupancy states are the underlying states of the game within this context, as stated by \citet{discrete,competitive}. Secondly, we have examined the optimal state-value functions of POSGs with public actions and observations. In this case, the application of our analysis has shown that the optimal state-value functions are convex functions of the belief states for common reward criteria, as previously established by \citet{finite,infinite}. Lastly, we have considered two-agent, one-sided POSGs. A direct application of our findings demonstrates that the optimal state-value functions of the one-sided POSGs are convex functions of the belief states for all two-agent and zero-sum, Stackelberg, as well as common reward criteria, as previously established by \citet{onesided,optimally}. These results are consistent with previously established properties, thus providing additional support for the accuracy of our findings, even when applied to more restrictive settings.

To summarize, the proposed methodology provides a comprehensive approach to analyzing the underlying structure of optimal state-value functions in POSGs. This approach is particularly useful for specific subclasses, including but not restricted to zero-sum, Stackelberg, and common reward partially observable stochastic games, where it is crucial to uncover the underlying structure to develop efficient planning and reinforcement learning algorithms. The authors anticipate that these findings will serve as a foundation to reveal the fundamental structures of optimal state-value functions in other games and to develop more efficient planning and reinforcement learning algorithms that leverage these structural properties. Future work in that direction includes describing \citeauthor{bellman}'s optimality equations leveraging upon the exhibited structures of the optimal state-value functions along the design of appropriate algorithms capable of exploiting these structures. This approach has proven particularly successful in many narrow settings including but not restricted to one-sided partially observable stochastic games \citep{cooperative,onesided,optimally}, zero-sum partially observable stochastic games \citep{zerosum}, and common-reward partially observable stochastic games \citep{heuristic,continuous}. The current theory focuses only on games with one or two criteria and unique optimal state-value functions. The authors are working on extending the present analysis to address POSGs with multiple criteria and no unique optimal state-value functions, \eg general-sum POSGs.

\section{Acknowledgements}
Jilles S. Dibangoye thanks God, who inspired this piece of work. This work was supported by ANR project \href{http://perso.citi-lab.fr/jdibangoy/#/plasma}{Planning and Learning to Act in Systems of Multiple Agents} under Grant ANR-19-CE23-0018, and ANR project \href{https://projet.liris.cnrs.fr/delicio/}{Data and Prior, Machine Learning and Control} under Grant ANR-19-CE23-0006, and ANR project \href{https://}{Multi-Agent Trust Decision Process for the Internet of Things} under Grant ANR-21-CE23-0016, all funded by French Agency ANR.

\bibliographystyle{unsrtnat}
\bibliography{Bibliography}

\begin{thebibliography}{27}
\providecommand{\natexlab}[1]{#1}
\providecommand{\url}[1]{\texttt{#1}}
\expandafter\ifx\csname urlstyle\endcsname\relax
  \providecommand{\doi}[1]{doi: #1}\else
  \providecommand{\doi}{doi: \begingroup \urlstyle{rm}\Url}\fi

\bibitem[Harsanyi(1997)]{bayesian}
John~C. Harsanyi.
\newblock \emph{Games with Incomplete Information Played by "{B}ayesian"
  Players}, pages 216--288.
\newblock Princeton University Press, 1997.

\bibitem[Hansen et~al.(2004)Hansen, Bernstein, and Zilberstein]{dynamic}
Eric~A. Hansen, Daniel~S. Bernstein, and Shlomo Zilberstein.
\newblock Dynamic programming for partially observable stochastic games.
\newblock In \emph{Proceedings of the 19th National Conference on Artificial
  Intelligence}, pages 709--715, 2004.

\bibitem[Shoham and Leyton-Brown(2008)]{multiagent}
Yoav Shoham and Kevin Leyton-Brown.
\newblock \emph{Multiagent systems: Algorithmic, game-theoretic, and logical
  foundations}.
\newblock Cambridge University Press, 2008.

\bibitem[Von~Neumann and Morgenstern(1944)]{economic}
John Von~Neumann and Oskar Morgenstern.
\newblock \emph{Theory of Games and Economic Behavior}.
\newblock Princeton University Press, 1944.

\bibitem[Bernstein et~al.(2000)Bernstein, Zilberstein, and
  Immerman]{decentralized}
Daniel~S. Bernstein, Shlomo Zilberstein, and Neil Immerman.
\newblock The complexity of decentralized control of {M}arkov decision
  processes.
\newblock In \emph{Proceedings of the 16th Conference on Uncertainty in
  Artificial Intelligence}, pages 32--37, 2000.

\bibitem[Tambe(2011)]{security}
Milind Tambe.
\newblock \emph{Security and game theory: algorithms, deployed systems, lessons
  learned}.
\newblock Cambridge University Press, 2011.

\bibitem[Smallwood and Sondik(1973)]{finite}
Richard~D. Smallwood and Edward~J. Sondik.
\newblock The optimal control of partially observable {M}arkov decision
  processes over a finite horizon.
\newblock \emph{Operations Research}, 21\penalty0 (5):\penalty0 1071--1088,
  1973.

\bibitem[Smith and Simmons(2005)]{pomdp}
Trey Smith and Reid Simmons.
\newblock Point-based {POMDP} algorithms: Improved analysis and implementation.
\newblock In \emph{Proceedings of the 21st Conference on Uncertainty in
  Artificial Intelligence}, pages 542--549, 2005.

\bibitem[Pineau et~al.(2006)Pineau, Gordon, and Thrun]{pb}
Joelle Pineau, Geoffrey~J. Gordon, and Sebastian Thrun.
\newblock Anytime point-based approximations for large {POMDP}s.
\newblock \emph{Journal of Artificial Intelligence Research}, 27\penalty0
  (1):\penalty0 335--380, 2006.

\bibitem[Shani et~al.(2013)Shani, Pineau, and Kaplow]{survey}
Guy Shani, Joelle Pineau, and Robert Kaplow.
\newblock A survey of point-based {POMDP} solvers.
\newblock \emph{Journal of Autonomous Agents and Multi-Agent Systems},
  27\penalty0 (1):\penalty0 1--51, 2013.

\bibitem[Gmytrasiewicz and Doshi(2005)]{framework}
Piotr~J. Gmytrasiewicz and Prashant Doshi.
\newblock A framework for sequential planning in multi-agent settings.
\newblock \emph{Journal of Artificial Intelligence Research}, 24\penalty0
  (1):\penalty0 49--79, 2005.

\bibitem[Fehr et~al.(2018)Fehr, Buffet, Thomas, and Dibangoye]{rho}
Mathieu Fehr, Olivier Buffet, Vincent Thomas, and Jilles Dibangoye.
\newblock rho-{POMDP}s have {L}ipschitz-continuous epsilon-optimal value
  functions.
\newblock In \emph{Advances in Neural Information Processing Systems 31}.
  Curran Associates, Inc., 2018.

\bibitem[Dibangoye et~al.(2012)Dibangoye, Amato, and Doniec]{heuristic}
Jilles~S. Dibangoye, Christopher Amato, and Arnaud Doniec.
\newblock Scaling up decentralized {MDP}s through heuristic search.
\newblock In \emph{Proceedings of the 28th Conference on Uncertainty in
  Artificial Intelligence}, pages 217--226, 2012.

\bibitem[Dibangoye et~al.(2016)Dibangoye, Amato, Buffet, and
  Charpillet]{continuous}
Jilles~S. Dibangoye, Christopher Amato, Olivier Buffet, and Fran{\c{c}}ois
  Charpillet.
\newblock Optimally solving dec-{POMDP}s as continuous-state {MDP}s.
\newblock \emph{Journal of Artificial Intelligence Research}, 55:\penalty0
  443--497, 2016.

\bibitem[Sorin(2003)]{stochastic}
Sylvain Sorin.
\newblock Stochastic games with incomplete information.
\newblock 570:\penalty0 375--395, 2003.

\bibitem[Hor{\'a}k et~al.(2017)Hor{\'a}k, Bo{\v{s}}ansk{\`y}, and
  P{\v{e}}chou{\v{c}}ek]{onesided}
Karel Hor{\'a}k, Branislav Bo{\v{s}}ansk{\`y}, and Michal
  P{\v{e}}chou{\v{c}}ek.
\newblock Heuristic search value iteration for one-sided partially observable
  stochastic games.
\newblock In \emph{Proceedings of the AAAI Conference on Artificial
  Intelligence}, volume~31, 2017.

\bibitem[Wiggers et~al.(2016)Wiggers, Oliehoek, and Roijers]{structure}
Auke~J. Wiggers, Frans~A. Oliehoek, and Diederik~M. Roijers.
\newblock Structure in the value function of two-player zero-sum games of
  incomplete information.
\newblock In \emph{Proceedings of the European Conference on Artificial
  Intelligence}, pages 1628--1629, 2016.

\bibitem[Delage et~al.(2023)Delage, Buffet, Dibangoye, and Saffidine]{zerosum}
Aur{\`e}lien Delage, Olivier Buffet, Jilles~S. Dibangoye, and Abdallah
  Saffidine.
\newblock {HSVI} can solve zero-sum partially observable stochastic games.
\newblock \emph{Dynamic Games and Applications}, pages 1--55, 2023.

\bibitem[Bellman(1957)]{bellman}
Richard~E. Bellman.
\newblock \emph{Dynamic Programming}.
\newblock Dover Publications, Inc., 1957.

\bibitem[Kova{\v{r}}{\`i}k et~al.(2022)Kova{\v{r}}{\`i}k, Schmid, Burch,
  Bowling, and Lis{\`y}]{rethinking}
Vojt{\v{e}}ch Kova{\v{r}}{\`i}k, Martin Schmid, Neil Burch, Michael Bowling,
  and Viliam Lis{\`y}.
\newblock Rethinking formal models of partially observable multiagent decision
  making.
\newblock \emph{Artificial Intelligence}, 303\penalty0 (C), 2022.

\bibitem[Nayyar et~al.(2011)Nayyar, Mahajan, and Teneketzis]{delayed}
Ashutosh Nayyar, Aditya Mahajan, and Demosthenis Teneketzis.
\newblock Optimal control strategies in delayed sharing information structures.
\newblock \emph{IEEE Transactions on Automatic Control}, 56\penalty0
  (7):\penalty0 1606--1620, 2011.

\bibitem[Oliehoek(2013)]{sufficient}
Frans~A. Oliehoek.
\newblock Sufficient plan-time statistics for decentralized {POMDP}s.
\newblock In \emph{Proceedings of the 24th International Joint Conference on
  Artificial Intelligence}, pages 302--308, 2013.

\bibitem[Hadfield-Menell et~al.(2016)Hadfield-Menell, Russell, Abbeel, and
  Dragan]{cooperative}
Dylan Hadfield-Menell, Stuart~J. Russell, Pieter Abbeel, and Anca Dragan.
\newblock Cooperative inverse reinforcement learning.
\newblock In \emph{Advances in Neural Information Processing Systems 29}, pages
  3916--3924. 2016.

\bibitem[Xie et~al.(2020)Xie, Dibangoye, and Buffet]{optimally}
Yuxuan Xie, Jilles Dibangoye, and Olivier Buffet.
\newblock Optimally solving two-agent decentralized pomdps under one-sided
  information sharing.
\newblock In \emph{Proceedings of the 37th International Conference on Machine
  Learning}, pages 10473--10482, 2020.

\bibitem[Puterman(1994)]{discrete}
Martin~L. Puterman.
\newblock \emph{Markov Decision Processes, Discrete Stochastic Dynamic
  Programming}.
\newblock John Wiley \& Sons, Inc., 1994.

\bibitem[Filar and Vrieze(2012)]{competitive}
Jerzy Filar and Koos Vrieze.
\newblock \emph{Competitive Markov decision processes}.
\newblock Springer, 2012.

\bibitem[Sondik(1978)]{infinite}
Edward~J. Sondik.
\newblock The optimal control of partially observable markov decision processes
  over the infinite horizon: Discounted cost.
\newblock \emph{Operations Research}, 12:\penalty0 282--304, 1978.

\end{thebibliography}
\end{document}